\documentclass[11pt]{article}%
\usepackage{amssymb}
\usepackage{comment}
\usepackage{amsfonts}
\usepackage{graphicx}
\usepackage{amsmath}
\usepackage{caption}
\usepackage{booktabs}
\usepackage{multirow}
\usepackage{makecell}
\usepackage{caption}
\usepackage{appendix}
\usepackage{tikz}
\usepackage{subcaption}
\usepackage{float}
\usepackage{indentfirst}
\usepackage{enumitem}
\usepackage[nohead]{geometry}
\usepackage[singlespacing]{setspace}
\usepackage[bottom]{footmisc}
\usepackage{indentfirst}
\usepackage[round]{natbib}
\usepackage{color}
\usepackage[colorlinks,linkcolor=red,anchorcolor=blue,citecolor=blue]%
{hyperref}
\usepackage{rotating}
\allowdisplaybreaks
\providecommand{\U}[1]{\protect \rule{.1in}{.1in}}

\newtheorem{theorem}{Theorem}[section]

\newtheorem{lemma}{Lemma}[section]

\newenvironment{proof}[1][Proof]{\noindent \textbf{#1.} }{\  \rule{0.5em}{0.5em}}
\numberwithin{equation}{section}
\newcommand{\blfootnote}[1]{
    \begingroup
    \renewcommand\thefootnote{}\footnote{#1}
    \addtocounter{footnote}{-1}
    \endgroup
}

\begin{document}
	
	\title{Asymptotic Analysis of Optimal Diversification in Catastrophe Risk Pooling}
	\author{Minh Chau Nguyen~ \thinspace Tony Wirjanto ~ \thinspace  Fan Yang\\
		{\small \ Department of Statistics and Actuarial Science, University of
			Waterloo }\\
		{\small Waterloo, ON N2L 3G1, Canada }\\
	}
	\date{{\small April 7, 2026}}
	\maketitle
	
	\begin{abstract}
Catastrophe risk has long been recognized to pose a serious threat to the
insurance sector. Catastrophe risk pooling offers an effective way to
diversify losses arising from catastrophic events. In this paper, we
investigate a catastrophe risk pooling structure and optimize it so that participants can
attain the maximum diversification benefit from joining the pool.
Determining the practical optimal pool entails solving a high-dimensional
optimization problem, for which analytical solutions are typically
unavailable and numerical methods can be computationally intensive and
potentially unreliable. To address this challenge, we evaluate the
diversification benefit in the limit and use it to derive an asymptotically
optimal pool which approximates the practical optimal pool. Through
simulation studies, we show that the asymptotically optimal pool provides an
accurate and reliable approximation to the practical optimal pool. We also
conduct an empirical analysis using data from the U.S. National Flood
Insurance Program to illustrate how the framework can be applied in practice.

	\end{abstract}

\blfootnote{We thank Lisa Gao and David Saunders for their comments and suggestions on an earlier version of the paper. We also thank participants at the following conferences: Climate Change and Insurance Conference 2025, International Centre for Mathematical Science (Edinburgh, U.K., September 10th - September 12th, 2025); RSS International Conference 2025, Royal Statistical Society (Edinburgh, U.K., September 1st - September 4th, 2025); 2025 Actuarial Research Conference, Society of Actuaries (Toronto, Canada, July 29th - August 1st, 2025); 2025 ICSA China Conference, International Chinese Statistical Association, Beijing Normal University (Zhuhai, Guangdong, China, June 28th  - June 30th, 2025); SSC Annual Meeting 2025, Statistical Society of Canada (Saskatoon, Canada, May 25th - May 28th, 2025). The usual disclaimer applies.}

\section{Introduction}

Climate change has increased the severity and frequency of natural disasters.
Since 2015, losses from natural disasters worldwide have been steadily
increasing, more than 50\% of which are uninsured (see Figure
\ref{fig:protection_gap}). A catastrophe risk pool is a promising mechanism for
managing extreme losses, as it helps diversify participants' catastrophe risk
and strengthen their resilience against natural disasters. Real-world
examples of such risk pools include the Florida Hurricane Catastrophe Fund,
the Caribbean Catastrophe Risk Insurance Facility (CCRIF), and the African Risk
Capacity (for further details of these examples, see \citealp{bollmann2019international} for instance). In
this paper, we aim at improving the efficiency of catastrophe risk pools by
optimally allocating the diversification benefits among participants.

\begin{figure}[htbp]
    \centering
    \includegraphics[width=0.75\linewidth]{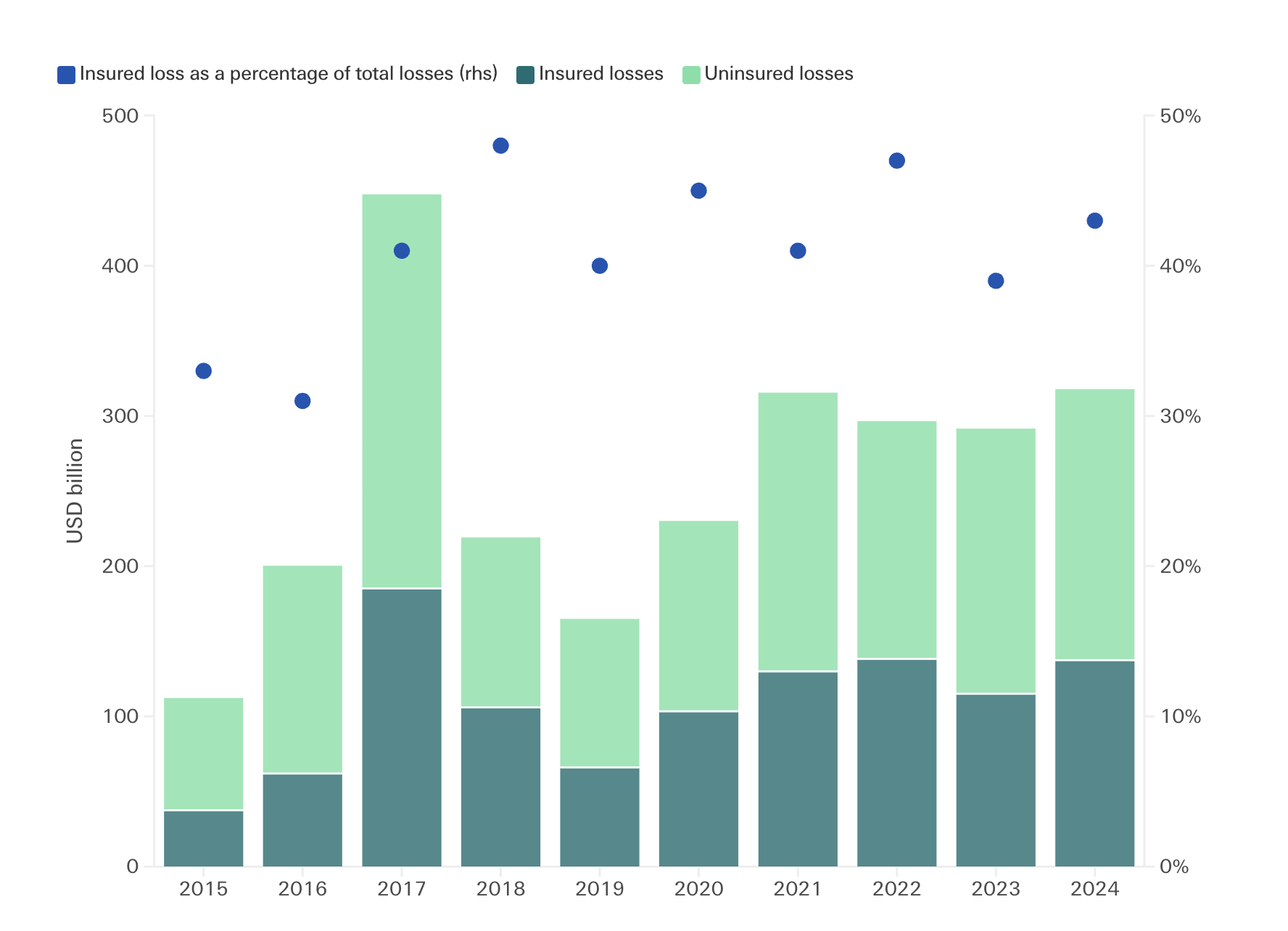}
    \caption{Global catastrophe losses from 2015 to 2024 (USD billion, 2024 prices)  \citep{swiss_re_sigma_2025}}
    \label{fig:protection_gap}
\end{figure}

It has been shown that economic losses from natural disasters exhibit
power-law tail behavior when the underlying hazard intensities
follow a power-law distribution. A power-law tail means that the survival
function of a loss variable $X$ satisfies%
\[
\mathbb{P} \left(  X>x\right)  \sim cx^{-\alpha},\qquad x>0,
\]
for some tail index $\alpha>0$. Such distributions are also referred to as
heavy-tailed. Empirical evidence supports this behavior for various types of
catastrophes, including earthquakes \citep{sornette1996rank}, hurricanes
\citep{hogg1983estimation,hsieh1999robustness}, wildfires
\citep{malamud2005characterizing}, and river floods
\citep{woo2011calculating}. 
 In this paper, we
assume that the catastrophic losses follow heavy-tailed distributions, and they need not be identical.

In practice, catastrophe risk pools are typically formed to diversify
exposures across geographically separated regions or across distinct perils,
thereby reducing cross-loss dependence. This is because independence among losses enhances
the potential for diversification, as verified under the Archimedean copula
structure in Theorem 2.3 of{\  \cite{embrechts2009additivity}. }Motivated by
this insight, we adopt the independence of the catastrophe losses as a
benchmark assumption in order to isolate the effect of heavy-tailed marginal
losses on diversification in a tractable manner. We acknowledge that
catastrophe losses may still exhibit residual dependence in practice due to
large-scale weather systems, climate patterns, or common macroeconomic
factors. In the empirical analysis, we show that our theoretical results
derived under the independence assumption continue to hold when the losses
exhibit weak dependence. A full analysis of pooling dependent catastrophic
losses, however, is beyond the scope of the present paper.

We consider a pooling structure in which the pool provides each participant
with coverage for a specific layer of loss, and each participant pays a
premium for this coverage which is proportional to the aggregated loss of the
pool. A formal definition of this pooling structure is provided in Section 
\ref{setup}. The diversification benefit is measured by using a so-called
diversification ratio (DR), defined in \eqref{DRp}, which is based on the
Value-at-Risk (VaR) measure. In this setup, after joining the pool, each
participant's risk consists of two components: a self-retained loss not
covered by the pool and the premium paid to participate in the pool. The DR
compares a participant's risk before and after joining the pool and thereby
measures the extent of the risk reduction achieved through participation.

Since the catastrophe losses are not assumed to be identically
distributed, the diversification benefit that each participant derives from the
pool depends on the layer of loss covered by the pool. In this paper, we
seek to determine the optimal layer for each participant so that all
participants simultaneously attain their maximum diversification benefit
from the pool. In practice, the DR is measured at any specific level $p$ of
the VaR measure. However, determining the practical pool that simultaneously
maximizes the diversification benefit for all participants leads to a
high-dimensional and multi-objective optimization problem. Analytical solutions are generally
unavailable, and numerical methods may be computationally expensive and
potentially unreliable. We therefore propose using the asymptotically
optimal pool, defined by evaluating DR in the limit as $p\rightarrow 1$, as
an approximation to the practical optimal pool for any finite $p$
sufficiently close to 1.

To be more specific, we consider two types of risk pools. In the first
model, all losses are assumed to be tail equivalent, i.e., they share the
same tail index, although they need not be identically distributed. In
the second model, we allow the losses to have different tail indices, in
which case a more substantial heterogeneity arises among participants. For
each model, we derive asymptotic expressions of the DR as the level $p$ of
the VaR measure approaches 1. These expressions characterize the
diversification benefit each participant obtains from the asymptotic pool.
Based on these results, we construct the asymptotically optimal pool by
selecting the optimal loss layer for each participant.

Through a set of simulation studies, we show that the asymptotically optimal
pool provides an accurate and reliable approximation to the practical
optimal pool when the level $p$ is close to 1. In particular, the
asymptotically optimal pooling strategy yields loss layers and diversification
ratios that are close to those obtained from the practical optimal pool. The
latter is computed at specific VaR levels $p$ using a global
optimization algorithm. In Appendix \ref{Appendix_algo}, we compare the
chosen algorithm with four alternative optimization methods and show that it
achieves a favorable balance between solution accuracy and computational
efficiency. In the empirical study, we construct three optimal flood risk
pools using the asymptotically optimal pooling strategy and flood loss data
from the U.S. National Flood Insurance Program (NFIP). We demonstrate the
full implementation of the proposed framework, including the validation of
model assumptions and the estimation of key parameters. The resulting pools
are shown to be consistent with our theoretical findings. 

The remainder of this paper is structured as follows. In Section \ref{thms},
we first provide the formal setup of the catastrophe risk pool. Then we derive
the asymptotic expressions for the DR when the level $p$ of VaR
approaches 1. This is done for two types of risk pools, one having tail
equivalent losses and the other having general heavy-tailed losses. Based on
these expressions, we derive the asymptotic optimal pool. In Section
\ref{sec:sim}, we examine the performance of the asymptotically optimal pool
through a set of simulations. Lastly, in Section \ref{sec:realdata}, we carry out an
empirical analysis to illustrate how this framework can be implemented in
practice. The appendices contain detailed proofs of the main results, a comparison
of algorithms used in the simulation study, and an additional exploratory analysis of the
flood loss data.

\section{Optimal pooling structure}\label{thms}

\subsection{Pool setup}\label{setup}

Suppose that there are $n$ participants in the pool, each with a ground-up
catastrophe loss $X_{i}$, $i=1,...,n$. Assume that $X_{1},X_{2},...,X_{n}$ are independent random variables with distribution functions $F_{1},...,F_{n}$.
As discussed in the introduction, these losses can be intentionally chosen
from different types of perils or from geographically well-separated
locations so as to justify the independence assumption.

Furthermore, these losses are assumed to follow heavy-tailed distributions. More
specifically,{ a loss} $X$ with a distribution function $F=1-\overline{F}$ {is
said to have a regularly varying tail with index }$\alpha>0${, denoted by}
$X\in \mathrm{RV}_{-\alpha}$ or $\overline{F}\in \mathrm{RV}_{-\alpha}${, if }
\[
\lim_{t\rightarrow \infty}\frac{\overline{F}(tx)}{\overline{F}(t)}=x^{-\alpha
},\qquad x>0.
\]
The tail index $\alpha$ represents the heaviness of the tail of the distribution, where the
smaller the value of $\alpha$ is, the heavier the tail is. See, for example,
\cite{haan2006extreme}  for more details on regularly varying (RV) functions.

By joining the pool, each participant is covered for a layer of $X_{i}$ with
an attachment point $d_{i}$ and a limit $l_{i}$, defined as
\begin{equation}
Y_{i}=\left \{
\begin{array}
[c]{lll}%
0, &  & X_{i}<d_{i},\\
X_{i}-d_{i}, &  & d_{i}\leq X_{i}<l_{i},\\
l_{i}-d_{i}, &  & l_{i}\leq X_{i}.
\end{array}
\right.  \label{Y}%
\end{equation}
This form of coverage is a common practice in many lines of property and casualty insurance, and is often applied in the contribution structure of existing climate
risk pools such as CCRIF  \citep{bollmann2019international}.

The aggregated loss in this pool is then given by $S_{n}=\sum_{i=1}^{n}Y_{i}$.
To receive the coverage from the pool, each participant pays a premium $P_{i}$
defined as
\[
P_{i}=\frac{E\left[  Y_{i}\right]  }{E\left[  S_{n}\right]  }\mathrm{VaR}%
_{p}(S_{n}),
\]
where $\mathrm{VaR}_{p}(Z)=G^{\leftarrow}(p)=\inf \{z:G(z) \geq p\}$ is the Value-at-Risk (VaR) or the general inverse function for a random variable $Z$ with a distribution function $G$ at a level $p\in(0,1)$. {This premium can be interpreted as the participants sharing the aggregated risk through a mean-proportional risk-sharing rule (see \cite{denuit2022risk} for instance). In particular, each participant shares a portion of $\frac{E\left[Y_{i}\right]  }{E\left[  S_{n}\right]  }$ of the total risk of the pool.} Furthermore, this premium can be rewritten as
\[
P_{i}=E[Y_{i}]\left(  1+\eta \right)  ,
\]
where $\eta=\frac{\mathrm{VaR}_{p}(S_{n})}{E[S_{n}]}-1$. When the level $p$ of VaR is close to 1, we have $\eta>0$. Thus, the premium $P_{i}$
follows the expected value principle, and the loading $\eta$ represents the
exceedance of the tail aggregated risk over the expected aggregated risk.

With the layer loss $Y_{i}$ covered by the pool, the risk retained by
participant $i$ is given by $\mathrm{VaR}_{p}(X_{i}-Y_{i})$. Then we quantify the
diversification benefit that participant $i$ obtains when joining the pool by using the
diversification ratio (DR) defined as
\begin{equation}
\mathrm{DR}_{i}(p)=\frac{\mathrm{VaR}_{p}(X_{i}-Y_{i})+\frac{E\left[
Y_{i}\right]  }{E\left[  S_{n}\right]  }\mathrm{VaR}_{p}(S_{n})}%
{\mathrm{VaR}_{p}(X_{i})}.\label{DRp}%
\end{equation}
DR compares the risk of a participant before and after joining the
pool. A DR value less than 1 indicates that the participant obtains a positive diversification benefit from joining the pool, and a smaller DR value means a
greater risk reduction for the participant. Note that the total risk that a
participant holds after joining the pool, including both the retained risk and
the portion of risk shared from the pool, is used for the definition (\ref{DRp}).
This is different from the diversification ratio in \cite{cui2021diversification} or
\cite{cui2022asymptotic}, where only the risk shared from the pool is
considered as the participant's post-pooling risk. By including the total risk that a participant holds after joining the pool, our newly proposed DR measure (\ref{DRp}) reflects the real situation faced by a pool participant and thus better captures the risk diversification effect of joining the pool.

The goal of this paper is to construct the pool by choosing the attachment
points $d_{i}$ and the limits $l_{i}$ such that all pool participants can
enjoy the maximum diversification benefit, that is, 
\begin{equation}
\min_{d_{i},l_{i},i=1,...,n}\mathrm{DR}_{i}(p)\text{,\qquad for all }%
i=1,...,n\text{.}  \label{pop}
\end{equation}%
However, an analytical solution to the high-dimensional and multi-objective optimization problem
(\ref{pop}) is generally not available. Considering that the pool is used to
manage the extreme risks for participants, we focus on the case where the level $p$ is close to 1. Denote the diversification benefit
that participant $i$ has in the limit when $p\rightarrow 1$ as 
\[
\mathrm{DR}_{i}(1):=\lim_{p\rightarrow 1}\mathrm{DR}_{i}(p).
\]%
In the following subsections, we set the attachment points in such a way
that each participant asymptotically has the same tail probability at these
points. Then we look for the optimal solution of the problem 
\begin{equation}
\min_{l_{i},i=1,...,n}\mathrm{DR}_{i}(1)\text{,\qquad }i=1,...,n\text{,}
\label{lpo}
\end{equation}%
for given $d_{i}$. This means that, for given attachment points $d_{i}$, we
solve for the optimal limits $l_{i}$ that allow the pool participants to
attain the maximum diversification benefit asymptotically. Through a set of
simulation studies, we demonstrate that this solution can be used to
approximate the practical optimal solution of problem (\ref{pop}) when $p$
is close to 1.

\subsection{Pool with tail equivalent losses}

In this subsection, we first consider a pool of losses which have the same heavy-tailedness but are not necessarily identically distributed, as detailed in the following model. By assuming a common tail index for all losses, we provide a simplified version of the general model, which helps illustrate the structure of the pool and the implications of the results.  We extend the study to a pool of losses with general heavy-tailedness in the next subsection.\bigskip

\textbf{Model 1}

\begin{itemize}
\item $\overline{F}_{i}\in \mathrm{RV}_{-\alpha}$ with $\alpha>0$, for
$i=1,...,n$. Furthermore, for
$i=1,...,n$, there exists $\theta_{i}>0$ such that
\begin{equation}
\lim_{t\rightarrow \infty}\frac{\overline{F}_{i}(t)}{\overline{F}_{1}%
(t)}=\theta_{i}>0.\label{TE}%
\end{equation}

\item The attachment point and limit are functions of $p$, meaning that $d_{i}=d_{i}(p)$
and $l_{i}=l_{i}(p)$, $i=1,...,n$, which satisfy, for some $\xi>0$, that
\begin{equation}
\lim_{p\rightarrow1}\frac{d_{i}(p)}{\mathrm{VaR}_{p}(X_{i})}=\xi,\label{a1}%
\end{equation}
and for some $\lambda_{i}>1$,
\begin{equation}
\lim_{p\rightarrow1}\frac{l_{i}(p)}{d_{i}(p)}=\lambda_{i}>1.\label{a2}%
\end{equation}

\end{itemize}

\bigskip

In Model 1, all losses are heavy-tailed with the same tail index $\alpha$. We further assume that the losses may have different scales,
captured by $\theta_{i}$, compared to the first loss distribution
$\overline{F}_{1}(t)$. A larger scale $\theta_{i}$ means a higher tail
probability of $X_{i}$ at the level $t$ compared to that of $X_{1}$. Under
the assumptions (\ref{a1}) and (\ref{a2}), both attachment points $d_{i}(p)$ and limits $l_{i}(p)$
diverge to infinity as the level $p$ approaches 1, or equivalently as the benchmark risk $\mathrm{VaR}_{p}(X_{i})$ increases to infinity. More specifically, for each loss
$X_{i}$, its attachment point $d_{i}(p)$ is assumed to be approximately $\xi$ times its
benchmark risk $\mathrm{VaR}_{p}(X_{i})$. This multiplier is set to be the same for
all losses. The limit $l_{i}(p)$ is then assumed to be approximately $\lambda_{i}$ times $d_{i}(p)$.
The multipliers $\lambda_i$ are allowed to be different for each loss and will be optimally chosen for fixed levels of $\xi$ to achieve an asymptotically optimal pool in Subsection \ref{optimal pool}.

Next, we obtain several relations that are helpful for further illustrating the implications of Model 1
and useful for establishing the proof of the main theorem in this subsection.

\begin{lemma}
\label{lem1}Under Model 1, we have
\begin{equation}
\lim_{p\rightarrow1}\frac{d_{i}(p)}{d_{1}(p)}=\theta_{i}^{1/\alpha}
,\qquad \lim_{p\rightarrow1}\frac{\overline{F}_{i}(d_{i})} {\overline{F}%
_{1}(d_{1})}=1, \label{r1}%
\end{equation}
and
\[
\lim_{p\rightarrow1}\frac{d_{i}(p)}{\mathrm{VaR}_{p}(X_{1})}=\theta
_{i}^{1/\alpha}\xi,\qquad \lim_{p\rightarrow1}\frac{l_{i} (p)}{\mathrm{VaR}%
_{p}(X_{1})}=\lambda_{i}\theta_{i}^{1/\alpha}\xi.
\]

\end{lemma}

From Lemma \ref{lem1}, we observe several additional implications of Model 1.
Firstly, when the level $p$ is close to 1, the attachment point
$d_{i}$ for loss $X_{i}$ is approximately $\theta_{i}^{1/\alpha}$ times
$d_{1}$ for $X_{1}$. Thus, for a loss $X_{i}$ with a larger scale $\theta_{i}$,
its attachment point $d_{i}$ is set at a higher level than those with
smaller scales. Secondly, the attachment points are set in such a way that the tail
probability of each loss at its attachment point is approximately the same.
Following a similar proof to that of Lemma \ref{lem1}, we can show that
\[
\lim_{p\rightarrow1}\frac{\overline{F}_{i}(l_{i})}{\overline{F}_{1}(l_{1}%
)}=\left(  \frac{\lambda_{1}}{\lambda_{i}}\right)  ^{\alpha}.
\]
This means that the tail probability of loss $X_{i}$ at its limit $l_{i}$ is $\left(
\lambda_{1}/\lambda_{i}\right)  ^{\alpha}$ times that of $X_{1}$, which
depends on $\lambda_{1}$ and $\lambda_{i}$. Thirdly, for loss $X_{i}$, its
attachment point $d_{i}(p)$ is approximately $\theta_{i}^{1/\alpha}\xi$ times
the benchmark risk $\mathrm{VaR}_{p}(X_{1})$ of $X_{1}$. It further
illustrates that losses with larger scales have higher attachment points.

Now we are ready to derive the asymptotic expression of $\mathrm{DR}_{i}(p)$
as $p\rightarrow1$.

\begin{theorem}
\label{theorem3}Under Model 1, we have
\begin{equation}
\mathrm{DR}_{i}(1)=\left \{
\begin{array}
[c]{lll}%
1+\delta_i\Delta_{\xi,\boldsymbol{\lambda}}, &  & \xi \geq1,\\
\xi+\delta_i\Delta_{\xi,\boldsymbol{\lambda}}, &  & 1/\lambda_{i}\leq \xi<1,\\
1-(\lambda_{i}-1)\xi+\delta_i\Delta_{\xi,\boldsymbol{\lambda}}, &  & \xi<1/\lambda
_{i},
\end{array}
\right.  \label{res1}%
\end{equation}
where
\[
\delta_i = \frac{\left(  \lambda_{i}^{1-\alpha
}-1\right)  }{\sum_{j=1}^{n}\left(  \lambda_{j}^{1-\alpha}-1\right)
\theta_{j}^{1/\alpha}},
\qquad
\Delta_{\xi,\boldsymbol{\lambda}}=\left(  \sum_{j\in Z}\left(  \theta_{j}^{-1/\alpha}%
+\xi \right)  ^{-\alpha}\right)  ^{1/\alpha}%
\]
with $Z=\left \{  j=1,2,...,n:\xi>\theta_{j}^{-1/\alpha}\left(  \lambda
_{j}-1\right)  ^{-1}\right \}  $.
\end{theorem}

A smaller $\mathrm{DR}_{i}(1)$ is preferred for pool participant $i$
because it means that more risk is diversified by joining the pool. From Theorem
\ref{theorem3}, since $\delta_i >0$ for any $\lambda_i>1$ and $\Delta_{\xi,\boldsymbol{\lambda}}\geq0$, we see that if
$\xi \geq1$, then $\mathrm{DR}_{i}(1) \geq 1$. In other words, if the deductible
$d_{i}$ is set at a higher level than $\mathrm{VaR}_{p}(X_{i})$, too much risk
is retained, and pool participant $i$ would fail to gain any risk diversification by joining
the pool. Since this is not an ideal situation, we focus on the case where $\xi<1$ in applications. Moreover, $\Delta_{\xi,\boldsymbol{\lambda}}$ can be considered as a representation of the overall risk diversification level of the pool, and in addition, a smaller $\Delta_{\xi,\boldsymbol{\lambda}}$ is also preferred as it can yield lower $\mathrm{DR}_{i}(1)$ {due to $\delta_i$ being positive}. 

\subsection{Pool with general losses}

In this subsection, we consider the general case where risks can have different
heavy-tailedness, which is specified in the following model.

\bigskip

\textbf{Model 2}

\begin{itemize}
\item $\overline{F}_{i}\in \mathrm{RV}_{-\alpha_{i}}$ with $\alpha_{i}>0$, for
$i=1,...,n$, and $\alpha_{1}=\min \left \{  \alpha_{1},...,\alpha_{n}\right \}
$. Furthermore, we assume that there exists $\theta_{i}\geq0$ such that
\[
\lim_{t\rightarrow \infty}\frac{\overline{F}_{i}(t)}{\overline{F}_{1}%
(t)}=\theta_{i}\geq0,
\]
provided that $\theta_{i}>0$ if $\alpha_{i}=\alpha_{1}$ and $\theta_{i}=0$ if
$\alpha_{i}>\alpha_{1}$.

\item The attachment point and limit are functions of $p$, i.e., $d_{i}=d_{i}(p)$
and $l_{i}=l_{i}(p)$, $i=1,...,n$, which satisfy that for some $\xi>0$,
\begin{equation}
\lim_{p\rightarrow1}\frac{d_{i}(p)}{\mathrm{VaR}_{p}(X_{i})}=\xi^{\alpha
_{1}/\alpha_{i}},\label{a3}%
\end{equation}
and for some $\lambda_{i}>1$,
\[
\lim_{p\rightarrow1}\frac{l_{i}(p)}{d_{i}(p)}=\lambda_{i}>1.
\]

\end{itemize}

\bigskip

In Model 2, the tails of $X_{1},X_{2},...,X_{n}$ can decay at different speeds
captured by the tail indices $\alpha_{i}$. There is a more substantial
difference among the risks than that in Model 1, {where the tails decay at the same speed (i.e., with same tail index), albeit with different scales}. A smaller tail
index means a heavier tail for the risk. Without loss of generality, the loss $X_1$ is
assumed to have the heaviest tail in the pool.

When a loss $X_{i}$ has the same tail index as $X_{1}$, they are
tail equivalent, which is the same as in Model 1, but may be
different in scale, again captured by $\theta_{i}>0$. When a loss $X_{i}$ has
a different tail index from $X_{1}$, by the assumption that $X_{1}$ has the
heaviest tail, its tail index $\alpha_{i}$ is strictly larger than $\alpha
_{1}$, and its scale parameter $\theta_{i}$ is still well-defined as $0$.

The attachment point $d_{i}$ is assumed to be $\xi^{\alpha_{1}/\alpha_{i}}$
times the benchmark risk level $\mathrm{VaR}_{p}(X_{i})$, which takes the difference in heavy-tailedness into account. When $\alpha_{i}=\alpha_{1}$, it is reduced to $\xi$ as seen in
Model 1. Again, we allow the limits, represented by $\lambda_{i}$'s, to be different for each loss and they
will be optimally chosen for fixed levels of $\xi$ in order to achieve an asymptotically optimal pool in Subsection
\ref{optimal pool}. 

Overall, when all tail indices in Model 2 are the same, Model 2 reduces to
Model 1. Similarly to Lemma \ref{lem1}, we obtain the following useful
relations for Model 2.

\begin{lemma}
\label{lem3}Under Model 2, we have
\[
\lim_{p\rightarrow1}\frac{d_{i}(p)}{d_{1}(p)}=\theta_{i}^{1/\alpha_{1}}
,\qquad \lim_{p\rightarrow1}\frac{\overline{F}_{i}(d_{i})}{\overline{F}
_{1}(d_{1})}=1,
\]
and
\[
\lim_{p\rightarrow1}\frac{d_{i}(p)}{\mathrm{VaR}_{p}(X_{1})}=\theta
_{i}^{1/\alpha_{1}}\xi,\qquad \lim_{p\rightarrow1}\frac{l_{i}(p)}{\mathrm{VaR}
_{p}(X_{1})}=\lambda_{i}\theta_{i}^{1/\alpha_{1}}\xi.
\]

\end{lemma}

When $\theta_{i} \neq 0$, Lemma
\ref{lem3} confers the same implications as Lemma \ref{lem1}. When $\theta_{i}=0$, three
limits in Lemma \ref{lem3} are reduced to 0. This result is intuitive because when
$\theta_{i}=0$, $X_{i}$ has a lighter tail than $X_{1}$
and its attachment point $d_{i}$ (or its limit $l_{i}$) grows at a lower speed than
$d_{1}$ (or $l_{1}$) as $p$ approaches 1. Nonetheless, for any $\theta_{i}\geq0$,
the attachment points are still set in such a way that the tail probability of each
loss at its attachment point is approximately the same. Moreover, following
the proof of Lemma \ref{lem3}, we can show that for $\theta_{i}\geq0$,
\[
\lim_{p\rightarrow1}\frac{\overline{F}_{i}(l_{i})}{\overline{F}_{1}(l_{1}%
)}=\frac{\lambda_{1}^{\alpha_{1}}}{\lambda_{i}^{\alpha_{i}}}.
\]
This means that the tail probability of loss $X_{i}$ at its limit $l_{i}$ is
$\lambda_{1}^{\alpha_{1}}/\lambda_{i}^{\alpha_{i}}$ times that of $X_{1}$,
which depends on $\lambda_{1}$, $\lambda_{i}$, and the tail indices. 

Now, we are ready to show the asymptotic expression of $\mathrm{DR}_{i}(p)$ as
$p\rightarrow1$ under Model 2.

\begin{theorem}
\label{difthm}Under Model 2, we have
\[
\mathrm{DR}_{i}(1)=\left \{
\begin{array}
[c]{lll}%
1+\widetilde{\delta}_i \Delta_{\xi,\boldsymbol{\lambda}}, &  & \xi \geq1,\\
\xi^{\alpha_{1}/\alpha_{i}}+\widetilde{\delta}_i \Delta_{\xi,\boldsymbol{\lambda}}, &
& 1/\lambda_{i}\leq \xi^{\alpha_{1}/\alpha_{i}}<1,\\
1-(\lambda_{i}-1)\xi^{\alpha_{1}/\alpha_{i}}+\widetilde{\delta}_i \Delta_{\xi
,\boldsymbol{\lambda}}, &  & \xi^{\alpha_{1}/\alpha_{i}}<1/\lambda_{i}.
\end{array}
\right.
\]
where
\begin{equation}
\widetilde{\delta}_i = \frac{\left(  1-\alpha
_{1}\right)  \left(  \lambda_{i}^{1-\alpha_{i}}-1\right)  \xi^{\alpha
_{1}/\alpha_{i}-1}}{(1-\alpha_{i})\sum_{j=1}^{n}\left(  \lambda_{j}%
^{1-\alpha_{1}}-1\right)  \theta_{j}^{1/\alpha_{1}}},
\qquad
\Delta_{\xi,\boldsymbol{\lambda}}=\left(  \sum_{j\in
Z}\left(  \theta_{j}^{-1/\alpha_{1}}+\xi \right)  ^{-\alpha_{1}}\right)
^{1/\alpha_{1}}\label{delta}%
\end{equation}
with $Z=\left \{  j=1,2,...,n:\xi>\theta_{j}^{-1/\alpha_{1}}\left(  \lambda
_{j}-1\right)  ^{-1}\right \}  $.
\end{theorem}

When all tail indices $\alpha_{i}$ are the same, Theorem \ref{difthm} reduces to Theorem \ref{theorem3}. Similarly to Theorem
\ref{theorem3}, we observe that a 
smaller $\Delta_{\xi,\boldsymbol{\lambda}}$
can yield a lower $\mathrm{DR}_{i}(1)$, which is the preferred case. 
Note that
if $\theta_{j}=0$, then risk $j$ does not belong to the set $Z$ in
(\ref{delta}). Thus, lighter-tailed risks have no contribution to
$\Delta_{\xi,\boldsymbol{\lambda}}$, and only losses as heavy-tailed
as $X_{1}$ determine the overall risk diversification level of the pool.

\subsection{Asymptotically optimal pool\label{optimal pool}}

In this subsection, we derive the optimal pooling structure in the limit for a pool
with general losses under Model 2. More specifically, we aim at minimizing
$\mathrm{DR}_{i}(1)$ for each participant based on the asymptotic expression
derived earlier in Theorem \ref{difthm}, where the value of $\xi$ is fixed and we solve for the optimal $\lambda_{i}$. The resulting pool, called the asymptotically
optimal pool, enables a simultaneous maximization of the diversification benefit
for all participants in the pool. It provides an
approximation to the practical problem (\ref{pop}) at a level $p$ close to 1.

\begin{theorem}
\label{thm:opt_sol}Assume that $0\leq \theta_i \leq 1$ for $i=1,...,n$. Under Model 2, we have for any $\xi>0$,%
\[
\min_{\boldsymbol{\lambda}}\mathrm{DR}_{i}(1)=\left \{
\begin{array}
[c]{lll}%
1, &  & \xi \geq1,\\
\xi^{\alpha_{1}/\alpha_{i}}, &  & 0<\xi<1,
\end{array}
\right.
\]
where the optimal $\boldsymbol{\lambda}^{\ast}=\arg \min_{\boldsymbol{\lambda}%
}\mathrm{DR}_{i}(1)$ satisfies that
\[
\boldsymbol{\lambda}^{\ast}\in \{(\lambda_{1},...,\lambda_{n}):\min \left \{
\xi^{-\alpha_{1}/\alpha_{i}},1\right \}  \leq \lambda_{i}\leq1+\theta
_{i}^{-1/\alpha_{1}}\xi^{-1},i=1,...,n\}.
\]

\end{theorem}

We note that in Theorem \ref{thm:opt_sol}, the condition $0\leq \theta_i \leq 1$ is without loss of generality. Indeed, $0<\theta_i \leq 1$ means that $X_1$ is set so that all other risks of same heavy-tailedness as $X_1$ has smaller scales, while $\theta_i=0$ means that the risk of the $i$-th participant has a lighter tail than $X_1$. As such, this theorem can be applied to any pool of risks satisfying Model 2.

Theorem \ref{thm:opt_sol} shows that at the limit, for any given $\xi $, each $\mathrm{DR}_{i}(1)$ is minimized,
which means that all participants in the pool can simultaneously obtain the
maximum diversification benefit from the pool. We show with a set of
simulations in the next section that the asymptotically optimal solution is a good
approximation to the solution of the practical optimization problem (%
\ref{pop}).

If the given $\xi$ is greater than or equal to $1$, then each participant has a
$\mathrm{DR}_{i}(1)$ of at least $1$, which is not ideal as no risk
reduction is obtained from joining the pool. When the given $\xi$ is less than
1, each participant can achieve the lowest $\mathrm{DR}_{i}(1)$ as
$\xi^{\alpha_{1}/\alpha_{i}}<1$. Therefore, each
participant obtains the largest risk reduction by joining the pool and the
pool distributes the diversification benefit efficiently among participants.

From Theorem \ref{thm:opt_sol}, we see that the lower $\xi$ is, the more
diversification benefit each participant obtains. However, it is important to point out that the mathematical proofs underlying the current methodology do
not apply to the case where $\xi=0$. {Therefore, the asymptotic pool in this paper is
derived for any given $\xi$ that is less than 1 and strictly greater than 0.}

From the proof of Theorem \ref{thm:opt_sol}, we see that the feasible set
$\boldsymbol{\lambda}^{\ast}$ is obtained from solving $\Delta
_{\xi,\boldsymbol{\lambda}}=0$ for any given $\xi$. It means that the asymptotically
optimal pool can completely diversify away the aggregated risk. As a result, each participant is only left with the retained risk. Participants with heavier tails (i.e., smaller
$\alpha_{i}$) can gain a larger diversification benefit (smaller $\xi
^{\alpha_{1}/\alpha_{i}}$) in this pool. For a participant with a lighter
tail (i.e., $\alpha_{i}>\alpha_{1}$), its scale $\theta_{i}$ is $0$, which leads to the
upper bound of $\lambda_{i}^{\ast}$ being $\infty$. This means that a 
participant with lighter-tailed loss $X_i$ can bring to the pool a layer loss with a limit $l_i$ unlimitedly larger than the benchmark risk $\mathrm{VaR}_p(X_i)$, which
provides greater flexibility for these pool participants.

\section{Simulation study}\label{sec:sim}
In this section, we use simulations to examine the effectiveness of the
asymptotically optimal pool derived in Section \ref{optimal pool} as an approximation to the practical optimal pool defined
in (\ref{pop}). Note that the problem (\ref{pop}) is naturally a
multi-objective optimization problem, since it seeks to minimize the
participant-specific DRs simultaneously. In general, such a problem may not
admit a single common minimizer, and therefore a scalarization is needed for
numerical implementation. We adopt the weighted-sum formulation of the practical optimization problem with equal weights as follows
\begin{equation}
\min_{\boldsymbol{\lambda }}\sum_{i=1}^{n}\mathrm{DR}_{i}(p).  \label{PAP}
\end{equation}%
This criterion can be interpreted as the minimization of the average diversification
ratio across pool participants, and thus provides a natural and tractable
solution to the original objective. We therefore compare $\boldsymbol{%
\lambda }^{\ast }$, obtained in Theorem \ref{thm:opt_sol}, with the optimal
solution $\boldsymbol{\lambda }(p)$ to problem (\ref{PAP}).

To solve the multidimensional optimization problem (\ref{PAP}), we apply the
Generalized Simulated Annealing (GSA) algorithm implemented in the R package 
\texttt{GenSA}. GSA is a generalized version of simulated annealing, a
widely used stochastic optimization method inspired by the annealing process
in metallurgy, where heat treatment is used to reduce a material's internal
energy (see e.g. \citealp{tsallis1996generalized} and \citealp{xiang_generalized_2013}). In Appendix %
\ref{Appendix_algo}, we compare GSA with four other optimization algorithms
and find that it provides accurate solutions while remaining computationally
efficient.

In this study, we assume that each loss $X_{i}$ follows a Fr\'{e}chet distribution
\[
F_i(x)=e^{-(x/s_{i})^{-\alpha_{i}}},\qquad s_{i}>0\text{,}%
\]
denoted by $F_i\sim \mathrm{Fr\acute{e}%
chet}\left(  \alpha_{i},s_{i}\right)  $. It can be
shown that $\overline{F}_i\in \mathrm{RV}_{-\alpha_{i}}$ and that the analytical form of its quantile function is $F_i^\leftarrow(p) = s_{i}[-\ln(p)]^{-1/\alpha_{i}}$. To compute the numerical solution of \eqref{PAP} at a level $p$, we first
simulate a sample of one million observations for the random vector
$(X_{1},...,X_{n})$, denoted by $\{(X_{j,1},...,X_{j,n})\}_{\{j=1,...,m\}}$,
with $m=1,000,000$. For a given $\xi$ and each $(\lambda_{1},...,\lambda_{n}%
)$, the attachment point $d_{i}$ is computed as $\xi^{\alpha_1/\alpha_{i}}%
F_{i}^{\leftarrow}(p)$, and the limit $l_{i}$ is $\lambda_{i}d_{i}$ for
$i=1,...,n$. The observed layer loss $Y_{j,i}$ for each observation $X_{j,i}$
are then computed by using (\ref{Y}). Thus the $j$-th observation of the
aggregated loss of the pool is $S_{j}=\sum_{i=1}^{n}Y_{j,i}$. Let $Z_{(1)}\leq
Z_{(2)}\leq \cdots \leq Z_{(m)}$ denote the order statistics of the sample
$Z_{1},...,Z_{m}$, $Z_{(1),i}\leq Z_{(2),i}\leq \cdots \leq Z_{(m),i}$
denote the order statistics of the sample $Z_{1,i},...,Z_{m,i}$, and $\lfloor x \rfloor$ denote the integer value of $x$. Then
$\mathrm{DR}_{i}(p)$ is estimated as
\[
\widehat{\mathrm{DR}}_{i}(p)=\frac{\mathrm{VaR}_{p}(X_{i}-Y_{i})}%
{\mathrm{VaR}_{p}(X_{i})}+\frac{\widehat{E[Y_{i}]}}{\sum_{i=1}^{n}%
\widehat{E[Y_{i}]}}\frac{S_{(\lfloor pm\rfloor)}}{X_{(\lfloor pm\rfloor),i}},
\]
where
\[
\frac{\mathrm{VaR}_{p}(X_{i}-Y_{i})}{\mathrm{VaR}_{p}(X_{i})}=\left \{
\begin{array}
[c]{lll}%
1, &  & F_{i}^{\leftarrow}(p)\leq d_{i},\\
\frac{d_{i}}{F_{i}^{\leftarrow}(p)}, &  & d_{i}<F_{i}^{\leftarrow}(p)\leq
l_{i},\\
1-\frac{l_{i}-d_{i}}{F_{i}^{\leftarrow}(p)}, &  & F_{i}^{\leftarrow}(p)>l_{i},
\end{array}
\right.
\]
and $\widehat{E[Y_{i}]}$ is the approximation of
\[
E[Y_{i}]=\int_{d_{i}}^{l_{i}}\overline{F}_{i}(x)dx=\int_{d_{i}}^{l_{i}}\left(
1-e^{-(x/s_{i})^{-\alpha_{i}}}\right)  dx
\]
 by using function \texttt{integrate()}
in R. 
Then we apply the GSA algorithm to search for the solution
$\boldsymbol{\widehat{\lambda}}(p)$ of (\ref{PAP}).

Since the asymptotical solutions $\boldsymbol{\lambda}^{\ast}$ form a set and the
GSA algorithm only produces one single value of $\boldsymbol{\widehat{\lambda}%
}(p)$ for each sample, we define the distance between $\boldsymbol{\lambda}^{\ast}$ and
$\boldsymbol{\widehat{\lambda}}(p)$ as
\begin{equation}
\Pi(p)=\min_{\boldsymbol{\lambda}^{\ast}}\left\Vert \boldsymbol{\widehat{\lambda}}(p) - \boldsymbol{\lambda}^{\ast}\right\Vert, \label{rd}%
\end{equation}
where $\left\Vert \cdot \right\Vert$ denotes a Euclidean norm.

In the first study, we consider two tail equivalent losses $X_{1}\sim \mathrm{Fr\acute{e}%
chet}\left(  8.5,100\right)  $ and $X_{2}\sim \mathrm{Fr\acute{e}%
chet}\left(
8.5,90\right)  $. They have the same tail index of $8.5$ but with
different scales. It can be shown that $\theta_{1}=1$ and $\theta
_{2}=(9/10)^{8.5}$ according to (\ref{TE}). Thus $X_{1}$ and $X_{2}$ satisfy
Model 1. We use 50 samples, for each of which we
repeat the process of finding $\boldsymbol{\widehat{\lambda}}(p)= (\widehat{\lambda}_{1}(p),\widehat{\lambda}_{2}(p))$ for a selection of levels $p$: $\{0.8,0.825,0.85,0.875,0.9,0.925,0.95,0.975,0.99\}$.
A summary of the results is shown in Table \ref{tb1}.

\begin{table}[htbp]
    \footnotesize
    \centering
    \begin{tabular}{|c|c|c|c|}
        \hline
        & $\widehat{\lambda}_1(p)$ & $\widehat{\lambda}_2(p)$ & $\Pi(p)$ \\ \hline

        \multicolumn{4}{|l|}{}\\
        \multicolumn{4}{|l|}{$\xi=0.1^{1/8.5}$, $\boldsymbol{\lambda}^{\ast}\in \{(\lambda_{1},\lambda_{2}):1.3111\leq \lambda_{1}\leq2.3111,1.3111\leq \lambda_{2}\leq2.4568\}$}\\ \hline
        $p=0.800$ & 1.311136 ($3.485\times10^{-6}$) & 1.311134 ($8.637\times10^{-7}$) & $2.086\times10^{-6}$ ($3.353\times10^{-6}$) \\ 
        $p=0.825$ & 1.311135 ($2.492\times10^{-6}$) & 1.311134 ($6.727\times10^{-7}$) & $2.006\times10^{-6}$ ($2.012\times10^{-6}$) \\ 
        $p=0.850$ & 1.311136 ($3.293\times10^{-6}$) & 1.311134 ($8.034\times10^{-7}$) & $2.559\times10^{-6}$ ($2.754\times10^{-6}$) \\ 
        $p=0.875$ & 1.311137 ($7.945\times10^{-6}$) & 1.311135 ($3.952\times10^{-6}$) & $4.415\times10^{-6}$ ($8.294\times10^{-6}$) \\ 
        $p=0.900$ & 1.311135 ($4.989\times10^{-6}$) & 1.311135 ($1.123\times10^{-6}$) & $3.313\times10^{-6}$ ($4.138\times10^{-6}$) \\
        $p=0.925$ & 1.311137 ($5.973\times10^{-6}$) & 1.311135 ($1.157\times10^{-6}$) & $4.658\times10^{-6}$ ($4.837\times10^{-6}$) \\
        $p=0.950$ & 1.311139 ($7.343\times10^{-6}$) & 1.311135 ($1.067\times10^{-6}$) & $6.000\times10^{-6}$ ($6.355\times10^{-6}$) \\
        $p=0.975$ & 1.280022 ($2.252\times10^{-6}$) & 1.311135 ($2.536\times10^{-6}$) & 0.031112 ($2.252\times10^{-6}$) \\
        $p=0.990$ & 1.290521 (0.01476691) & 7198.80 (13290.34) & 7197.11 (13289.91) \\ \hline
        
        \multicolumn{4}{|l|}{}\\
        \multicolumn{4}{|l|}{$\xi=0.3^{1/8.5}$, $\boldsymbol{\lambda}^{\ast}\in \{(\lambda_{1},\lambda_{2}):1.1522\leq \lambda_{1}\leq2.1522,1.1522\leq \lambda_{2}\leq2.2802\}$}\\ \hline
        $p=0.800$ & 1.147701 ($2.929\times10^{-4}$) & 1.152167 ($2.018\times10^{-6}$) & 0.004465 ($2.929\times10^{-4}$) \\ 
        $p=0.825$ & 1.136951 ($2.245\times10^{-6}$) & 1.152168 ($2.457\times10^{-6}$) & 0.015215 ($2.245\times10^{-6}$) \\ 
        $p=0.850$ & 1.136950 ($6.760\times10^{-8}$) & 1.152166 ($3.083\times10^{-8}$) & 0.015216 ($6.760\times10^{-8}$) \\ 
        $p=0.875$ & 1.136951 ($9.338\times10^{-7}$) & 1.152168 ($9.206\times10^{-7}$) & 0.015215 ($9.338\times10^{-7}$) \\ 
        $p=0.900$ & 1.136950 ($1.093\times10^{-6}$) & 1.152167 ($2.628\times10^{-7}$) & 0.015216 ($1.093\times10^{-6}$) \\ 
        $p=0.925$ & 1.136978 ($1.400\times10^{-4}$) & 1.152168 ($7.317\times10^{-6}$) & 0.015189 ($1.400\times10^{-4}$) \\ 
        $p=0.950$ & 1.136953 ($1.332\times10^{-5}$) & 1.152170 ($1.483\times10^{-5}$) & 0.015213 ($1.328\times10^{-5}$) \\ 
        $p=0.975$ & 1.139445 (0.0054) & 9880.94 (46848.42) & 9879.60 (46848.22) \\ 
        $p=0.990$ & 1.139079 (0.0053) & 2005.45 (5243.49) & 2004.15 (5243.10) \\ \hline
        
        \multicolumn{4}{|l|}{}\\
        \multicolumn{4}{|l|}{$\xi=0.5^{1/8.5}$, $\boldsymbol{\lambda}^{\ast}\in \{(\lambda_{1},\lambda_{2}):1.0850\leq \lambda_{1}\leq2.0850,1.0850\leq \lambda_{2}\leq2.2055\}$}\\ \hline
        $p=0.800$ & 1.076468 ($2.489\times10^{-8}$) & 1.084964 ($5.657\times10^{-9}$) & 0.008496 ($2.489\times10^{-8}$) \\
        $p=0.825$ & 1.076472 ($1.466\times10^{-6}$) & 1.084964 ($1.980\times10^{-8}$) & 0.008492 ($1.466\times10^{-6}$) \\
        $p=0.850$ & 1.076472 ($1.359\times10^{-8}$) & 1.084964 ($5.094\times10^{-9}$) & 0.008492 ($1.359\times10^{-8}$) \\
        $p=0.875$ & 1.076472 ($1.078\times10^{-6}$) & 1.084968 ($1.159\times10^{-6}$) & 0.008492 ($1.078\times10^{-6}$) \\
        $p=0.900$ & 1.076471 ($1.880\times10^{-6}$) & 1.084968 ($2.022\times10^{-6}$) & 0.008493 ($1.880\times10^{-6}$) \\
        $p=0.925$ & 1.076691 (0.001558) & 5274.84 (37291.11) & 5273.74 (37290.95) \\
        $p=0.950$ & 1.076806 (0.001671) & 470.53 (2389.10) & 469.41 (2388.89) \\
        $p=0.975$ & 1.120331 (0.299812) & 35621.42 (111784.24) & 35620.07 (111783.96) \\
        $p=0.990$ & 2.022278 (0.996342) & 1223.92 (6821.75) & 1.223.26 (6821.47) \\ \hline

        \multicolumn{4}{|l|}{}\\
        \multicolumn{4}{|l|}{$\xi=0.7^{1/8.5}$, $\boldsymbol{\lambda}^{\ast}\in \{(\lambda_{1},\lambda_{2}):1.0428\leq \lambda_{1}\leq2.0428,1.0428\leq \lambda_{2}\leq2.1587\}$}\\ \hline
        $p=0.800$ & 1.038573 ($1.515\times10^{-9}$) & 1.042858 ($6.061\times10^{-10}$) & 0.004282 ($1.515\times10^{-9}$) \\
        $p=0.825$ & 1.038573 ($8.314\times10^{-9}$) & 1.042858 ($8.512\times10^{-9}$) & 0.004282 ($8.306\times10^{-9}$) \\
        $p=0.850$ & 1.038570 ($1.584\times10^{-8}$) & 1.042855 ($5.233\times10^{-9}$) & 0.004285 ($1.584\times10^{-8}$) \\
        $p=0.875$ & 1.038571 ($2.295\times10^{-6}$) & 1.042856 ($2.832\times10^{-6}$) & 0.004284 ($2.292\times10^{-6}$) \\
        $p=0.900$ & 1.038570 ($5.663\times10^{-7}$) & 1.042855 ($4.288\times10^{-7}$) & 0.004284 ($5.663\times10^{-7}$) \\
        $p=0.925$ & 1.038571 ($2.935\times10^{-6}$) & 1.042855 ($8.213\times10^{-7}$) & 0.004283 ($2.935\times10^{-6}$) \\
        $p=0.950$ & 1.070454 (0.223433) & 1615.48 (4892.52) & 1614.34 (4892.18) \\
        $p=0.975$ & 2.679588 (6.628282) & 1.042856 ($2.396\times10^{-6}$) & 1.181359 (6.50907) \\
        $p=0.990$ & 3.352461 (7.815203) & 936.62 (6615.53) & 937.17 (6615.14) \\ \hline

    \end{tabular}
    \caption{\small Pool with two tail equivalent losses $X_{1}\sim \mathrm{Fr\acute{e}%
chet}\left(  8.5,100\right)  $ and $X_{2}\sim \mathrm{Fr\acute{e}
chet}\left(
8.5,90\right)$. Sample mean and standard error are reported.}
    \label{tb1}
\end{table}

In the second study, we have two losses $X_{1}\sim \mathrm{Fr\acute{e}%
chet}(8.5,100)$ and $X_{2}\sim \mathrm{Fr\acute{e}chet}(9,100)$. Thus, they have
different tail indices of $8.5$ and $9$. It can be shown that $\theta_{1}=1$
and $\theta_{2}=0$ according to (\ref{TE}). So $X_{1}$ and $X_{2}$ satisfy
Model 2. 
We use 50 samples, for each of which we
repeat the process of finding $\boldsymbol{\widehat{\lambda}}(p)= (\widehat{\lambda}_{1}(p),\widehat{\lambda}_{2}(p))$ for the same selection of levels $p$ as that in the first study.
A summary of the results is shown in Table \ref{tb2}.

\begin{table}[htbp]
    \footnotesize
    \centering
    \begin{tabular}{|c|c|c|c|}
        \hline
        & $\widehat{\lambda}_1(p)$ & $\widehat{\lambda}_2(p)$ & $\Pi(p)$  \\ \hline

        \multicolumn{4}{|l|}{}\\
        \multicolumn{4}{|l|}{$\xi=0.1^{1/8.5}$, $\boldsymbol{\lambda}^{\ast}\in \{(\lambda_{1},\lambda_{2}):1.3111\leq \lambda_{1}\leq2.3111,\lambda_{2}\geq1.2915\}$}\\ \hline
        $p=0.800$ & 1.311135 ($1.776\times10^{-6}$) & 1.291550 ($9.240\times10^{-7}$) & $1.209982\times10^{-6}$ ($1.557\times10^{-6}$) \\
        $p=0.825$ & 1.311134 ($1.456\times10^{-6}$) & 1.291550 ($6.781\times10^{-7}$) & $1.064790\times10^{-6}$ ($1.081\times10^{-6}$) \\
        $p=0.850$ & 1.311135 ($3.336\times10^{-6}$) & 1.291550 ($1.539\times10^{-6}$) & $2.075998\times10^{-6}$ ($2.961\times10^{-6}$) \\
        $p=0.875$ & 1.311135 ($2.428\times10^{-6}$) & 1.291550 ($9.873\times10^{-7}$) & $2.042666\times10^{-6}$ ($1.982\times10^{-6}$) \\
        $p=0.900$ & 1.311137 ($8.617\times10^{-6}$) & 1.291551 ($2.215\times10^{-6}$) & $4.206583\times10^{-6}$ ($8.181\times10^{-6}$) \\
        $p=0.925$ & 1.311138 ($9.893\times10^{-6}$) & 1.291551 ($4.177\times10^{-6}$) & $4.433363\times10^{-6}$ ($9.544\times10^{-6}$) \\
        $p=0.950$ & 1.311138 ($9.893\times10^{-6}$) & 1.291551 ($4.177\times10^{-6}$) & $4.433363\times10^{-6}$ ($9.544\times10^{-6}$) \\
        $p=0.975$ & 1.288958 ($1.336\times10^{-5}$) & 1.291560 ($1.252\times10^{-5}$) & 0.022176 ($1.336\times10^{-5}$) \\
        $p=0.990$ & 1.287235 (0.000133) & 1.291553 ($4.712\times10^{-6}$) & 0.023899 (0.000133) \\ \hline

        \multicolumn{4}{|l|}{}\\
        \multicolumn{4}{|l|}{$\xi=0.3^{1/8.5}$, $\boldsymbol{\lambda}^{\ast}\in \{(\lambda_{1},\lambda_{2}):1.1522\leq \lambda_{1}\leq2.1522,\lambda_{2}\geq1.1431\}$}\\ \hline
        $p=0.800$ & 1.152009 (0.000186) & 1.143137 ($2.242\times10^{-6}$) & 0.000158 (0.000185) \\
        $p=0.825$ & 1.142726 ($1.207\times10^{-5}$) & 1.143141 ($1.186\times10^{-5}$) & 0.009440 ($1.207\times10^{-5}$) \\
        $p=0.850$ & 1.142595 ($1.468\times10^{-5}$) & 1.143144 ($1.057\times10^{-5}$) & 0.009571 ($1.468\times10^{-5}$) \\
        $p=0.875$ & 1.142404 ($6.793\times10^{-5}$) & 1.143157 ($6.801\times10^{-5}$) & 0.009762 ($6.793\times10^{-5}$) \\
        $p=0.900$ & 1.142168 ($9.421\times10^{-6}$) & 1.143143 ($9.787\times10^{-6}$) & 0.009998 ($9.421\times10^{-6}$) \\
        $p=0.925$ & 1.141885 ($1.320\times10^{-5}$) & 1.143140 ($1.337\times10^{-5}$) & 0.010281 ($1.320\times10^{-5}$) \\
        $p=0.950$ & 1.141505 ($2.996\times10^{-5}$) & 1.143144 ($3.104\times10^{-5}$) & 0.010662 ($2.996\times10^{-5}$) \\
        $p=0.975$ & 1.141307 (0.002242) & 1607.40 (9992.78) & 0.010861 (0.002237) \\
        $p=0.990$ & 1.140792 (0.002672) & 3176.93 (12684.18) & 0.011374 (0.002672) \\ \hline

        \multicolumn{4}{|l|}{}\\
        \multicolumn{4}{|l|}{$\xi=0.5^{1/8.5}$, $\boldsymbol{\lambda}^{\ast}\in \{(\lambda_{1},\lambda_{2}):1.0850\leq \lambda_{1}\leq2.0850,\lambda_{2}\geq1.0801\}$}\\ \hline
        $p=0.800$ & 1.079644 ($2.206\times10^{-6}$) & 1.080063 ($2.184\times10^{-6}$) & 0.005320 ($2.206\times10^{-6}$) \\
        $p=0.825$ & 1.079565 ($1.590\times10^{-6}$) & 1.080062 ($1.819\times10^{-6}$) & 0.005400 ($1.590\times10^{-6}$) \\
        $p=0.850$ & 1.079474 ($1.047\times10^{-7}$) & 1.080060 ($1.188\times10^{-8}$) & 0.005489 ($1.047\times10^{-7}$) \\
        $p=0.875$ & 1.079375 ($4.760\times10^{-6}$) & 1.080061 ($4.415\times10^{-6}$) & 0.005589 ($4.760\times10^{-6}$) \\
        $p=0.900$ & 1.079390 ($1.974\times10^{-5}$) & 1.080060 ($5.233\times10^{-9}$) & 0.005574 ($1.974\times10^{-5}$) \\
        $p=0.925$ & 1.079095 ($3.024\times10^{-6}$) & 1.080061 ($2.361\times10^{-6}$) & 0.005869 ($3.024\times10^{-6}$) \\
        $p=0.950$ & 1.078885 ($2.786\times10^{-6}$) & 1.080061 ($3.106\times10^{-6}$) & 0.006079 ($2.786\times10^{-6}$) \\
        $p=0.975$ & 1.078530 ($5.088\times10^{-5}$) & 1.080068 ($1.049\times10^{-5}$) & 0.006434 ($5.088\times10^{-5}$) \\
        $p=0.990$ & 1.401185 (1.013326) & 2646.74 (18707.64) & 0.187982 (0.760469) \\ \hline

        \multicolumn{4}{|l|}{}\\
        \multicolumn{4}{|l|}{$\xi=0.7^{1/8.5}$, $\boldsymbol{\lambda}^{\ast}\in \{(\lambda_{1},\lambda_{2}):1.0428\leq \lambda_{1}\leq2.0428,\lambda_{2}\geq1.0404\}$}\\ \hline
        $p=0.800$ & 1.040178 ($1.922\times10^{-5}$) & 1.040427 ($7.051\times10^{-8}$) & 0.002677 ($1.922\times10^{-5}$) \\
        $p=0.825$ & 1.040100 ($8.694\times10^{-6}$) & 1.040434 ($1.095\times10^{-5}$) & 0.002754 ($8.694\times10^{-6}$) \\
        $p=0.850$ & 1.040061 ($1.776\times10^{-5}$) & 1.040436 ($1.224\times10^{-5}$) & 0.002793 ($1.776\times10^{-5}$) \\
        $p=0.875$ & 1.039997 ($2.079\times10^{-6}$) & 1.040431 ($3.339\times10^{-6}$) & 0.002858 ($2.079\times10^{-6}$) \\
        $p=0.900$ & 1.039934 ($3.940\times10^{-6}$) & 1.040429 ($2.463\times10^{-6}$) & 0.002921 ($3.940\times10^{-6}$) \\
        $p=0.925$ & 1.039856 ($6.023\times10^{-6}$) & 1.040428 ($1.509\times10^{-6}$) & 0.002999 ($6.023\times10^{-6}$) \\
        $p=0.950$ & 1.040296 (0.001181) & 1242.01 (4665.02) & 0.002560 (0.001176) \\
        $p=0.975$ & 1.107629 (0.372054) & 1.040431 ($2.081\times10^{-5}$) & 0.034428 (0.208918) \\
        $p=0.990$ & 3.289461 (6.477553) & 671.27 (2688.12) & 1.927083 (6.251131) \\ \hline

    \end{tabular}
    \caption{\small Pool with two losses of different heavy-tailedness $X_{1}\sim \mathrm{Fr\acute{e}%
chet}(8.5,100)$ and $X_{2}\sim \mathrm{Fr\acute{e}chet}(9,100)$. Sample mean and standard error are reported.}
    \label{tb2}
\end{table}

As mentioned above, an auxiliary comparison between optimization algorithms is performed to select the best candidate for these two studies (see Appendix \ref{Appendix_algo} for further details). In the experiments within this comparison, we use a smaller set of 20 samples with parameters from the second study of risks of different indices, and set $\xi=0.3^{1/8.5}$ and $p=0.95$. On average, it takes 3.74 hours to obtain each $\boldsymbol{\widehat{\lambda}}(p)$ corresponding to each of the 20 samples using the GSA algorithm\footnote{These experiments were carried out using the group of servers ``biglinux.math'' of University of Waterloo. More details on the computing power of these servers can be found at: \url{https://uwaterloo.ca/math-faculty-computing-facility/services/service-catalogue-research-linux/research-linux-server-hardware\#biglinux}}. Therefore, by using parallel computing, we can say that each line in Tables \ref{tb1} and \ref{tb2} (which employs 50 samples) requires between 3.74 and 187 hours depending on the available computing power. As such, it is clear that the explicit expression of the approximation proposed in this paper provides a much needed reduction in computational cost when solving the practical problem (\ref{PAP}). We also observe next that the approximation $\boldsymbol{\lambda}^\ast$ has a high level of accuracy in comparison to $\widehat{\boldsymbol{\lambda}}(p)$ provided by the algorithm.

From Tables~\ref{tb1} and \ref{tb2}, we have some observations which are similar for both studies. The lower bounds of $\boldsymbol{\lambda}^\ast$ generally provide a good approximation to $\widehat{\boldsymbol{\lambda}}(p)$ when $p$ is between 0.8 and 0.95, especially for losses with a lower scale or a lighter tail ($X_2$ in both studies). As such, in practical implementations, one could use $\lambda_i = \xi^{-\alpha_1/\alpha_i}$ for $i=1,...,n$ and would have a good approximation to the optimal pool.

However, we notice that as $p$ gets closer to 1, $\widehat{\boldsymbol{\lambda}}(p)$ no longer represents a reliable solution to the practical problem (\ref{PAP}) due to its large variations. This can be explained by the fact that a higher level $p$ lowers the number of observations used for estimations, thus leading to a higher level of uncertainty, which is evident in the magnitudes of the standard deviations of $\widehat{\lambda}_1(p)$, $\widehat{\lambda}_2(p)$ and $\Pi(p)$. This effect is further amplified by a larger value of $\xi$ for the same reason. As such, these two effects potentially contribute to the poorer performance of the approximation when $p$ is close to 1.

\section{Empirical analysis} \label{sec:realdata}
In this section, we examine the application of the theoretical framework by using flood loss data from the U.S. National Flood Insurance Program (NFIP). Created in 1968, this program aims at sharing the losses from flood damages with homeowners and limiting flood damages by reducing development in floodplains. NFIP is managed and administered by the Federal Emergency Management Agency (FEMA)\footnote{These data are available from: \url{https://www.fema.gov/openfema-data-page/fima-nfip-redacted-claims-v2}. It includes multiple details on each claim since 1978: date, damage amounts, payment amounts, characteristics of the building, details on the location, information on the coverage in the insurance policy, etc.}. We use the damage amounts, which are aggregated from the building damage (\texttt{buildingDamageAmount}) and contents damage (\texttt{contentsDamageAmount}) for each state and each month from 1978 to 2023. Hence, each state has a total of 552 observations. 

In this study, we consider three states: New York (NY), California (CA), Florida (FL). We first conduct some preliminary data analysis, and the details of each statistical test are provided in Appendix \ref{Appendix_data_test}. The main results of this analysis are summarized here in the main text. Using the test proposed in \cite{dietrich2002testing}, we find that there is statistically significant evidence that NY, CA, FL have regularly varying tail distributions. This result confirms the assertion that the three losses satisfy the assumption of having regularly varying tails in Model 1 and Model 2. Using Pearson correlation tests, we also find that there is statistically significant evidence of pairwise linear independence between the three losses. However, the Spearman correlation test shows a significant monotonic relation between NY and FL, meaning that these three losses may not be strictly independent. Nonetheless, the following empirical studies show that our results still hold, indicating that the independence assumption among the losses may be relaxed in our theorems.

Moreover, we apply the tail equivalence test proposed in \cite{daouia2024optimal} and show the $p$-values of the tests calculated at different values of $k$ in Figure \ref{fig:sametail_test}. FL and CA show strong evidence of having the same tail index, while NY and CA can be considered either tail equivalent or not depending on the particular choice of $k$.  

\begin{figure}[h]
\centering
\begin{subfigure}{0.5\textwidth}
  \centering
   \includegraphics[width=\linewidth]{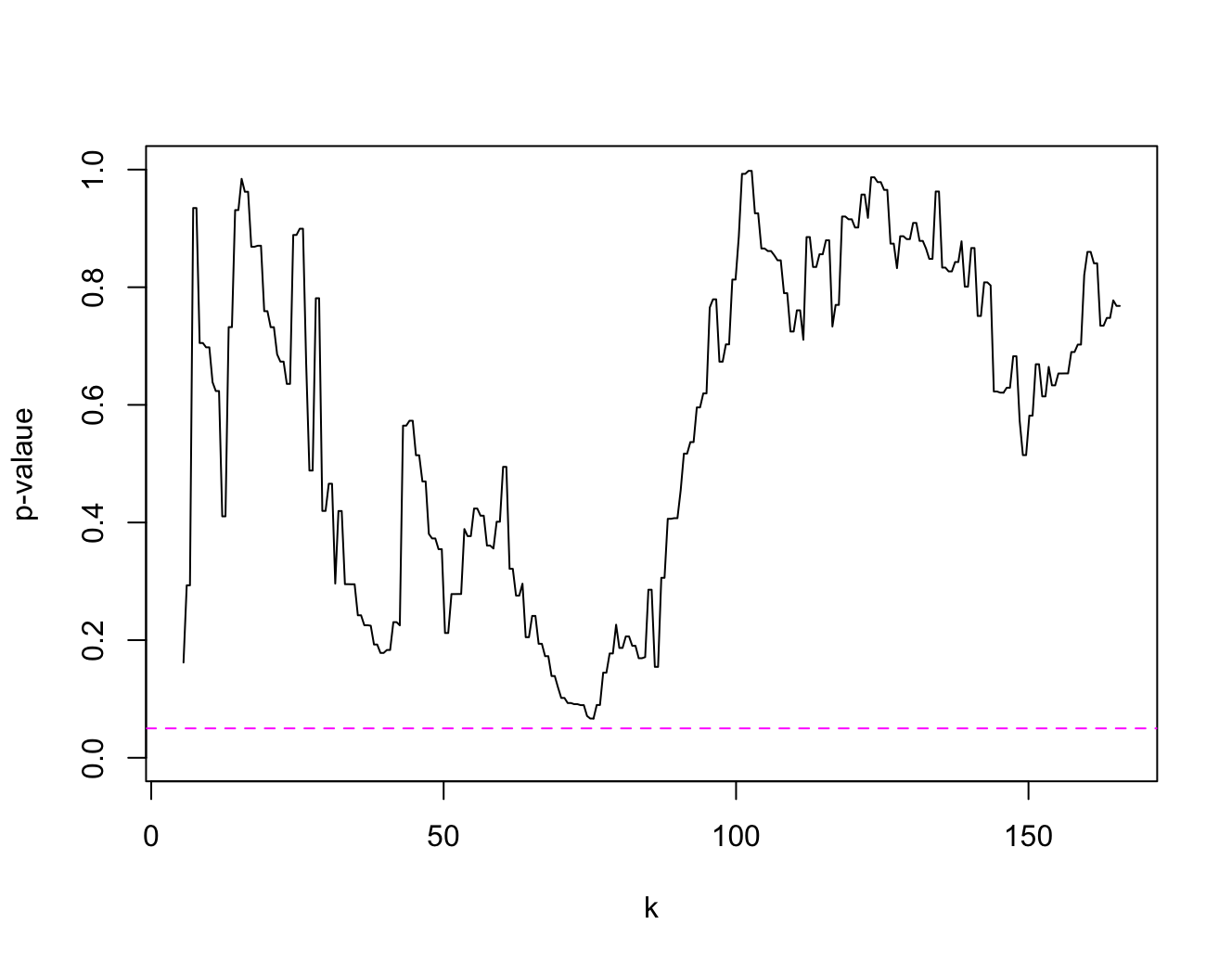}
    \caption{\small FL and CA}
    \label{fig:sametail_test1}
\end{subfigure}%
\begin{subfigure}{.5\textwidth}
  \centering
   \includegraphics[width=\linewidth]{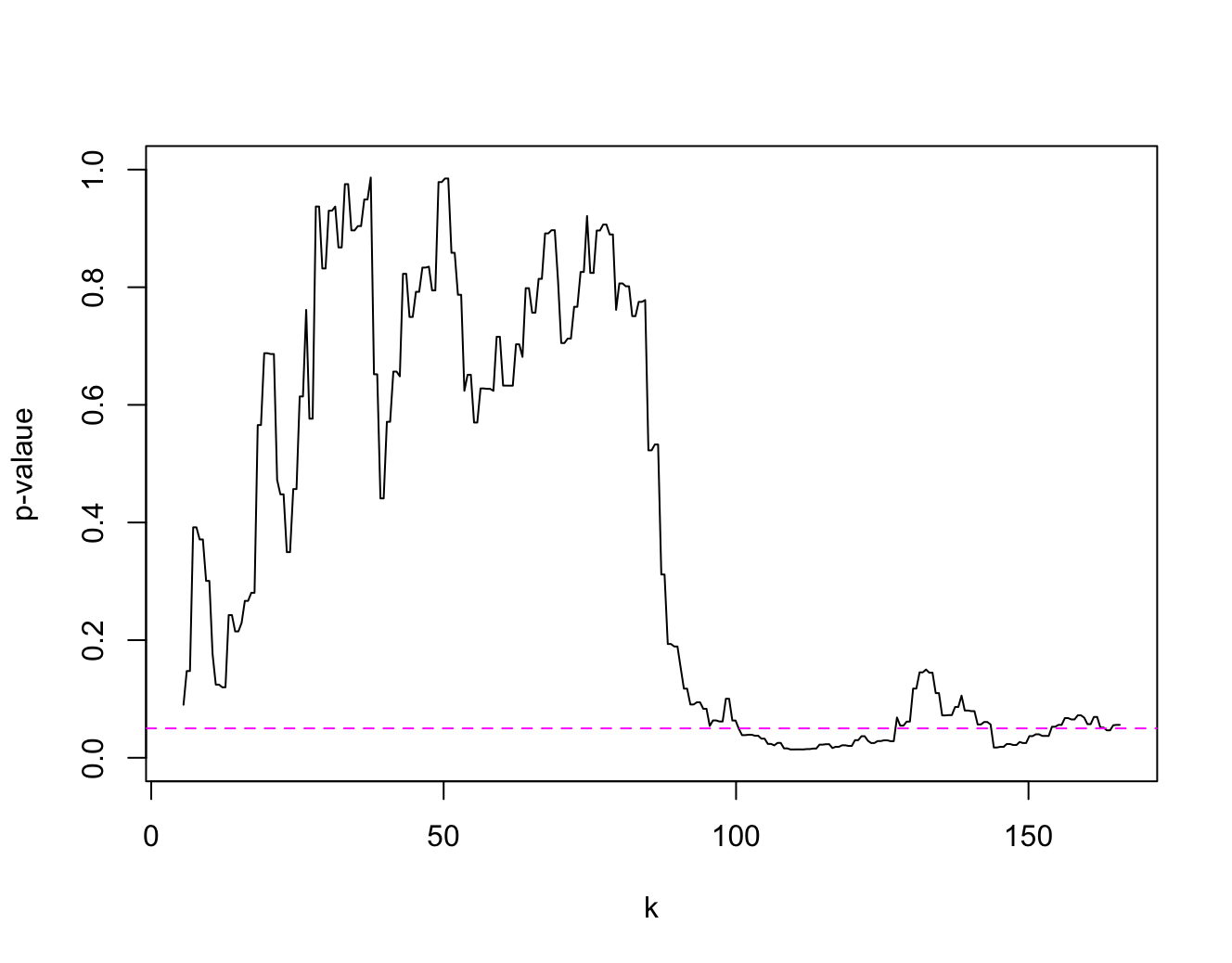}
 \caption{\small NY and CA}%
    \label{fig:sametail_test2}
\end{subfigure}%
\caption{\small Results of equivalent tail tests. The $p$-value is computed based on $k$ largest observations of $X_1$ and $k$ largest observations of $X_{2}$. The null hypothesis is that $X_1$ and $X_2$ have equivalent tails. The dashed line indicates a $p$-value of 0.05.}
\label{fig:sametail_test}
\end{figure}

Based on the above analysis, we study the following three pools of losses.
Pool 1 consists of the losses from FL ($X_{1}$) and CA ($X_{2}$), which
satisfy the assumptions in Model 1. Pool 2 consists of the losses from CA
($X_{1}$) and NY ($X_{2}$), which satisfy the assumptions in Model 2. Pool 3
consists of all three losses, FL ($X_{1}$), CA ($X_{2}$), and NY ($X_{3}$),
which satisfy the assumptions in Model 2. In each pool, let $X_{1,i}%
,X_{2,i},...,X_{m,i}$ be the $m$ $(=552)$ observations for loss $X_{i}$. We
use the following steps to calculate the DR for each participant.

First we determine the tail index of each loss. In Pool 1, since the losses
are considered to be tail equivalent, we use the pooled tail estimator defined
as \eqref{Hill_pooled_formula} in Appendix \ref{Appendix_data_test},
$\widehat{\alpha}_{FL-CA}$, which was proposed in \cite{daouia2024optimal}. In
Pool 2 and Pool 3, since the losses are considered to have different tail
indices, we estimate the tail index with the Hill's estimator for each loss
separately. The Hill's estimator \citep{hill1975} for a sample $X_{1},...,X_{m}$
is defined as
\[
\widehat{\alpha}=\left(  \frac{1}{k}\sum_{j=1}^{k}\log X_{(m-j+1)}-\log
X_{(m-k)}\right)  ^{-1}.
\]
The plots for each estimator with varying $k$ are shown in Figure
\ref{fig:hill_pool}. From these plots, we set $k$ to be $55$, which is the
10\% of the entire data points for each loss. Then we obtain: (i) for Pool 1,
$\widehat{\alpha}_{FL-CA}=\widehat{\alpha}_{1}=\widehat{\alpha}_{2}=0.604$, (ii) for Pool 2, $\widehat{\alpha}_{CA}=\widehat{\alpha}_{1}=0.646$ and
$\widehat{\alpha}_{NY}=\widehat{\alpha}_{2}=0.719$, and (iii) for Pool 3, $\widehat{\alpha}_{FL}=\widehat{\alpha
}_{1}=0.555$, $\widehat{\alpha}%
_{CA}=\widehat{\alpha}_{2}=0.646$ and $\widehat{\alpha}_{NY}=\widehat{\alpha}_{3}=0.719$. We can see that all of these losses are extremely heavy-tailed
with an infinite first moment.

\begin{figure}[htbp]
  \centering

  \begin{subfigure}{0.45\textwidth}
    \includegraphics[width=\linewidth]{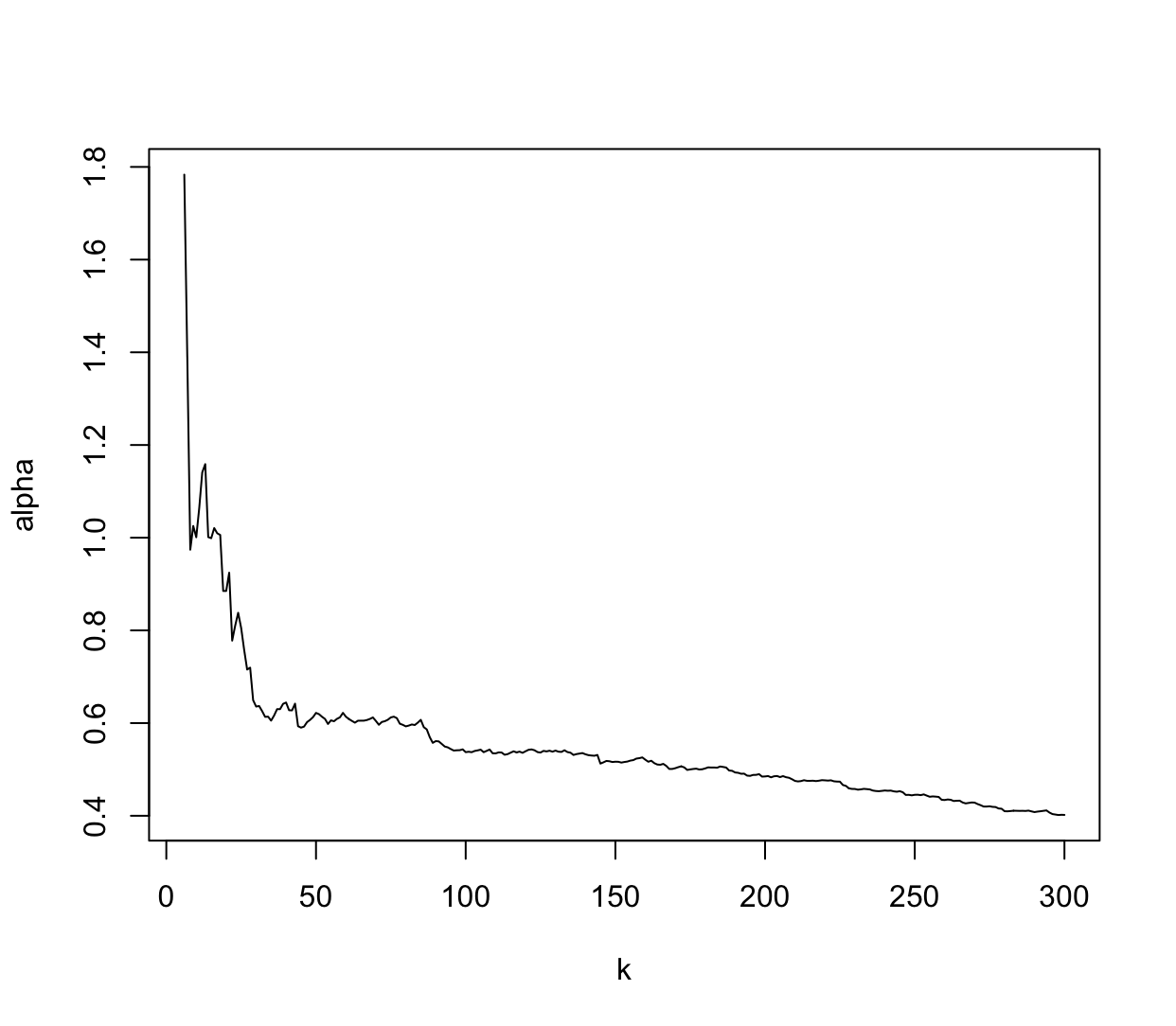}
    \caption{\small Pooled $\widehat{\alpha}_{FL-CA}$}
    \label{fig:hill_pool_FLCA}
  \end{subfigure}
    \hfill
  \begin{subfigure}{0.45\textwidth}
    \includegraphics[width=\linewidth]{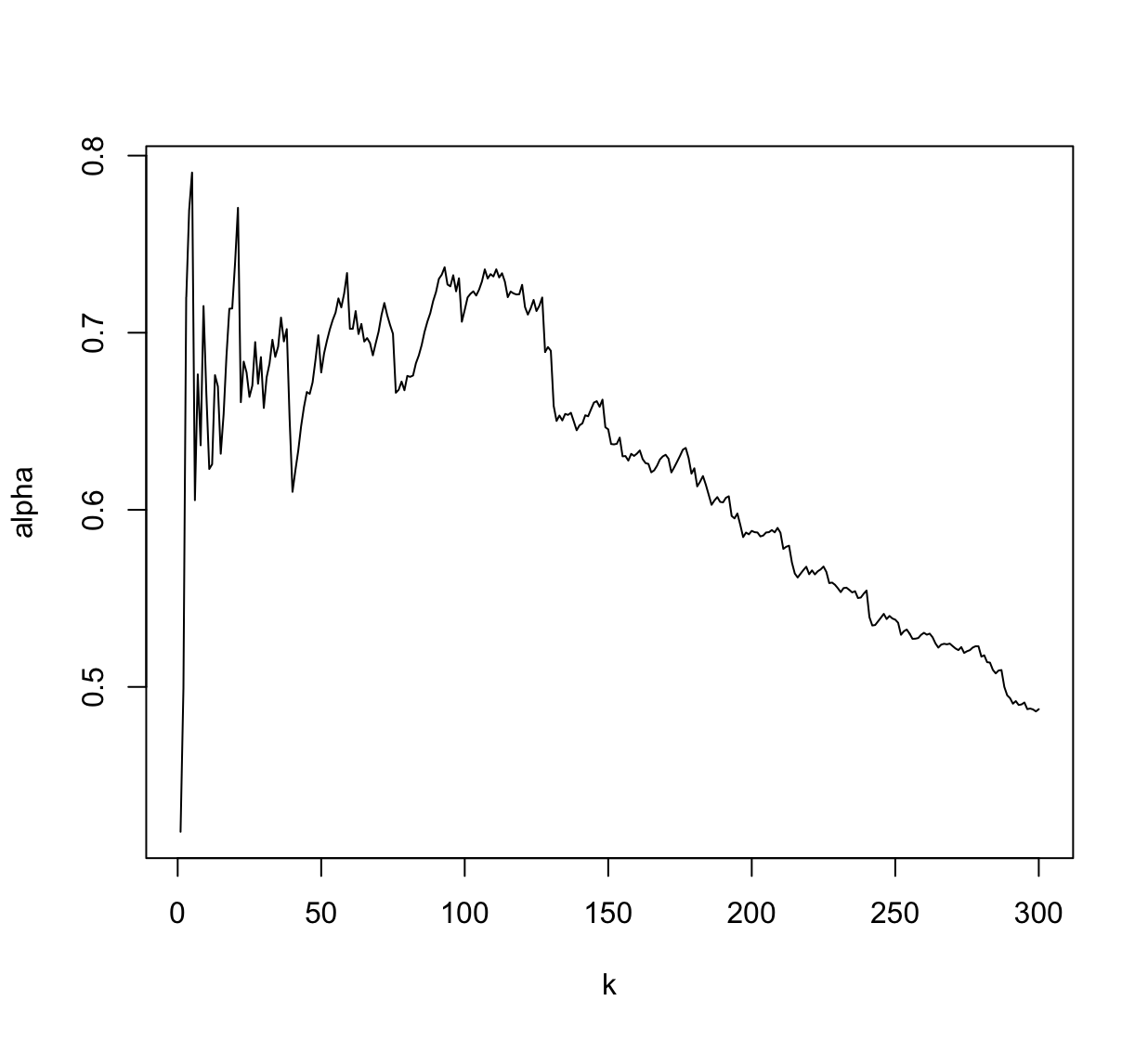}
    \caption{\small $\widehat{\alpha}_{NY}$}
    \label{fig:hill_pool_NY}
  \end{subfigure}

  \begin{subfigure}{0.45\textwidth}
    \includegraphics[width=\linewidth]{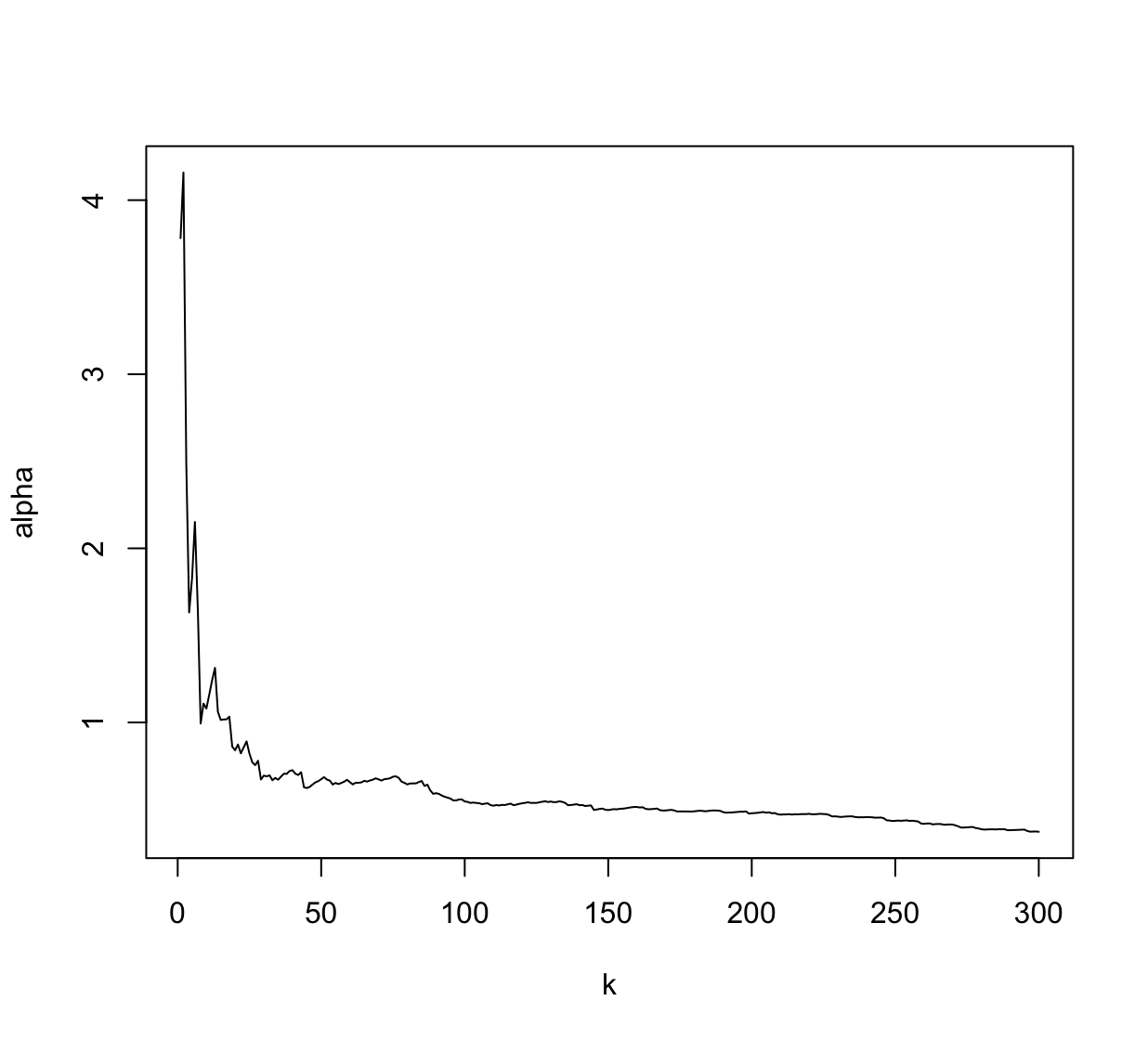}
    \caption{\small $\widehat{\alpha}_{CA}$}
    \label{fig:hill_pool_CA}
  \end{subfigure}
    \hfill
  \begin{subfigure}{0.45\textwidth}
    \includegraphics[width=\linewidth]{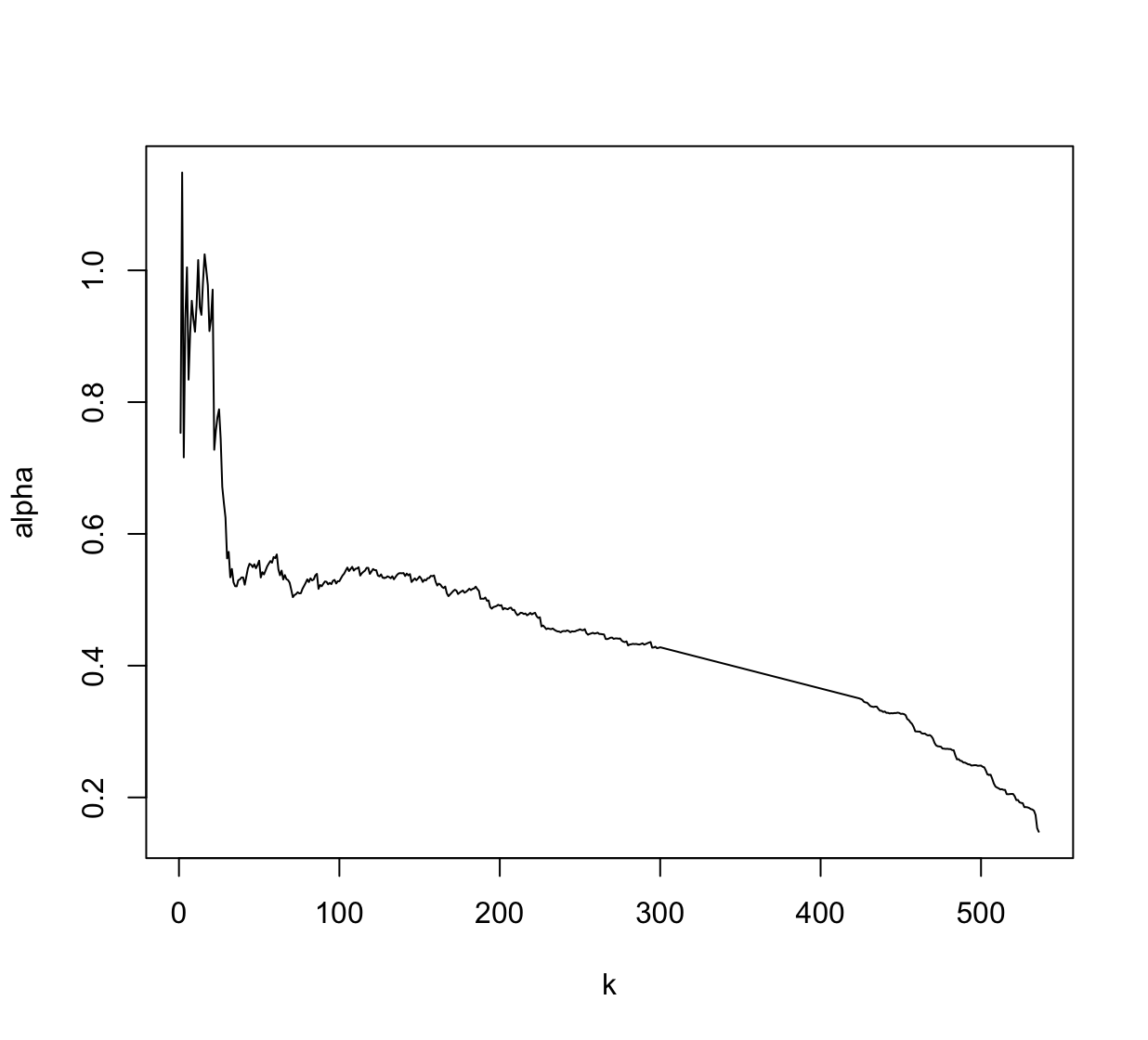}
    \caption{\small$\widehat{\alpha}_{FL}$}
    \label{fig:hill_pool_FL}
  \end{subfigure}
\captionsetup{justification=raggedright,singlelinecheck=false}
  \caption{Plots of Hill's estimator for each loss}
  \label{fig:hill_pool}
\end{figure}

Next, for losses in Pool 1, which are considered to be tail equivalent, we also
need to estimate the scale parameter $\theta_{CA}$ for CA as defined in
(\ref{TE}). Note that by definition $\theta_{FL}$ in Pool 1 is 1. We propose
the following empirical estimator for $\theta_{CA}$
\begin{equation}\label{eq:theta_est}
\widehat{\theta}_{CA}=\frac{\widehat{\overline{F}}_{CA}(X_{(m-h),FL})}{\widehat
{\overline{F}}_{FL}(X_{(m-h),FL})}=\frac{\sum_{j=1}^{m}1_{\left \{  X_{j,CA}\geq
X_{(m-h),FL}\right \}  }}{\sum_{j=1}^{m}1_{\left \{  X_{j,FL}\geq X_{(m-h),FL}\right \}  }},
\end{equation}
where $m=552$. By varying $h$, we plot the values of $\widehat{\theta}_{CA}$
in Figure \ref{fig:theta.plot}. Then by taking $h=55$, which is the 10\% of the entire data
points, $\widehat{\theta}_{CA}$ is calculated as $0.3637$.

\begin{figure}
    \centering
    \includegraphics[width=0.7\linewidth]{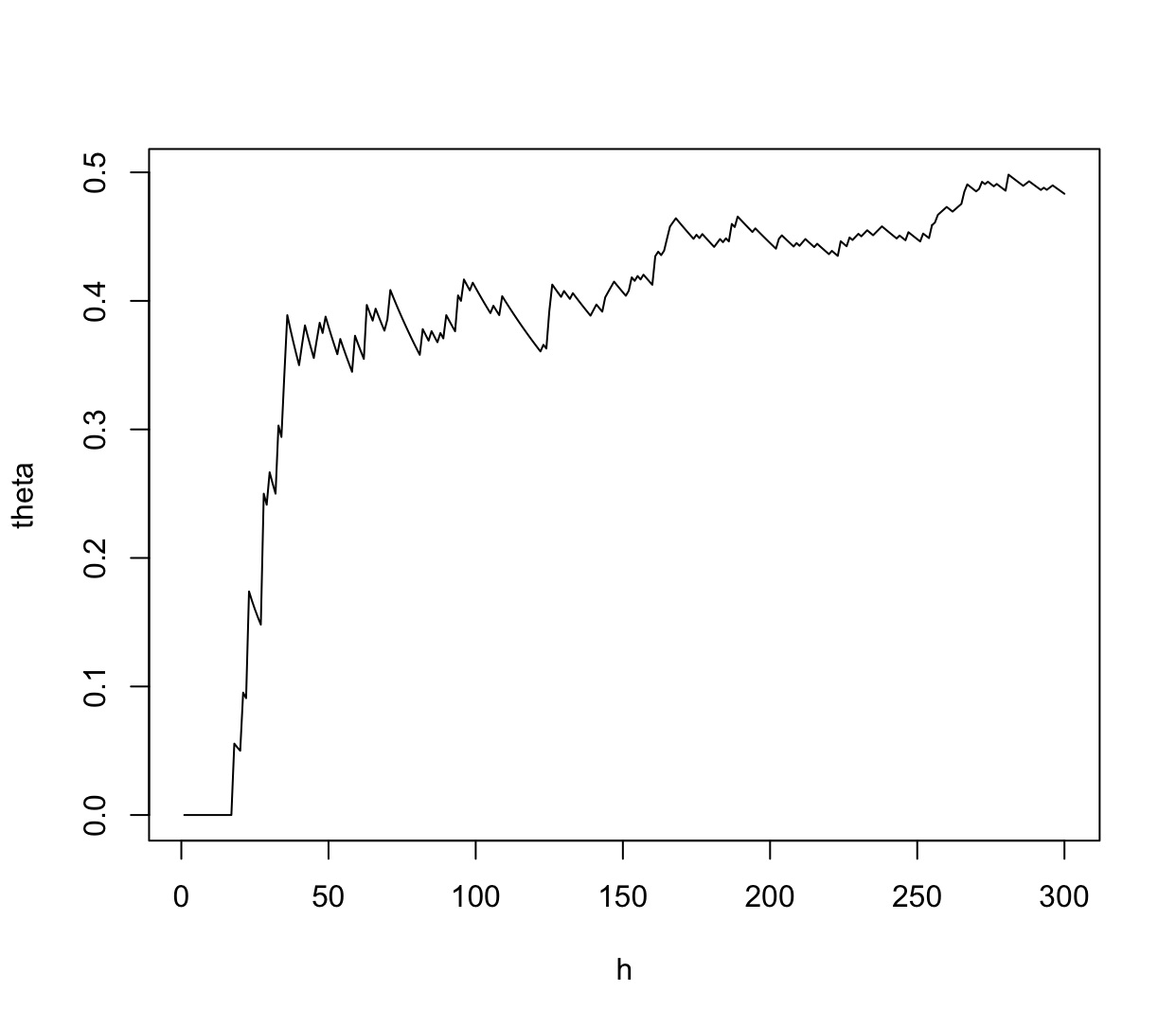}
    \captionsetup{justification=raggedright,singlelinecheck=false}
    \caption{\small Plot of $\widehat{\theta}_{CA}$ with varying $h$}
    \label{fig:theta.plot}
\end{figure}

Now we are ready to compute the DR for each participant in a pool. We apply
the following estimator for $\mathrm{DR}_{i}(p)$
\[
\widehat{\mathrm{DR}}_{i}(p)=\frac{\widehat{\mathrm{VaR}}_{p}(X_{i}-Y_{i}%
)}{\widehat{\mathrm{VaR}}_{p}(X_{i})}+\frac{\widehat{E[Y_{i}]}}{\sum_{i=1}%
^{n}\widehat{E[Y_{i}]}}\frac{S_{(\lfloor pm\rfloor)}}{X_{(\lfloor pm\rfloor),i}},
\]
with each term being calculated as follows. First, to make the estimation more stable for high
levels of $p$, we adopt the following EVT-based quantile estimator for
$\mathrm{VaR}_{p}(X_{i})$ with $p>0.8$%
\[
\widehat{\mathrm{VaR}}_{p}(X_{i})=X_{(\lfloor0.8m\rfloor),i}\left(  \frac
{0.2}{1-p}\right)  ^{1/\widehat{\alpha}_{i}},
\]
where $m=552$. This means that to
estimate $\mathrm{VaR}_{p}(X_{i})$ at a confidence level $p>0.8$ close to 1,
we extrapolate the empirical estimator of $\mathrm{VaR}_{0.8}(X_{i})$,
$X_{(\lfloor0.8m\rfloor),i}$, to the level $p$ (for details of this estimator
see for example Theorem 4.3.8 of \cite{haan2006extreme}). Then, following the assumptions in
Model 1 and 2, we set the levels of the attachment points in each pool as
$d_{i}=\xi^{\widehat{\alpha}_{1}/\widehat{\alpha}_{i}}\widehat{\mathrm{VaR}%
}_{p}(X_{i})$ for $i=1,2,3$, where $\xi$ is
chosen to be of certain values, and the limits in each pool as $l_i=\lambda_i d_i$, where $\lambda_i = \xi^{-\widehat{\alpha}_{1}/\widehat{\alpha}_{i}}$ are the lower bounds of $\boldsymbol{\lambda}^\ast$ from Theorem \ref{thm:opt_sol}. This is based on the observation in the previous section that the lower bounds provide the most accurate approximation to the solution of the practical optimization problem. Then, for each loss
observation $X_{j,i}$, the layer loss $Y_{j,i}$ is calculated according to
(\ref{Y}). Let $R_{j,i}=X_{j,i}-Y_{j,i}$. Thus we have
\[
\widehat{\mathrm{VaR}}_{p}(X_{i}-Y_{i})=R_{(\lfloor pm\rfloor),i}.
\]
and%
\[
\widehat{E[Y_{i}]}=\frac{1}{m}{\displaystyle \sum \limits_{j=1}^{m}} Y_{j,i}.
\]
Lastly, the aggregated loss in the pool is $S_{j}=
{\displaystyle \sum \limits_{i=1}^{n}}
Y_{j,i}$, where $n=2$ or $3$.

\begin{figure} [htbp]
    \centering

    \begin{subfigure}{0.45\textwidth}
    \includegraphics[width=\linewidth]{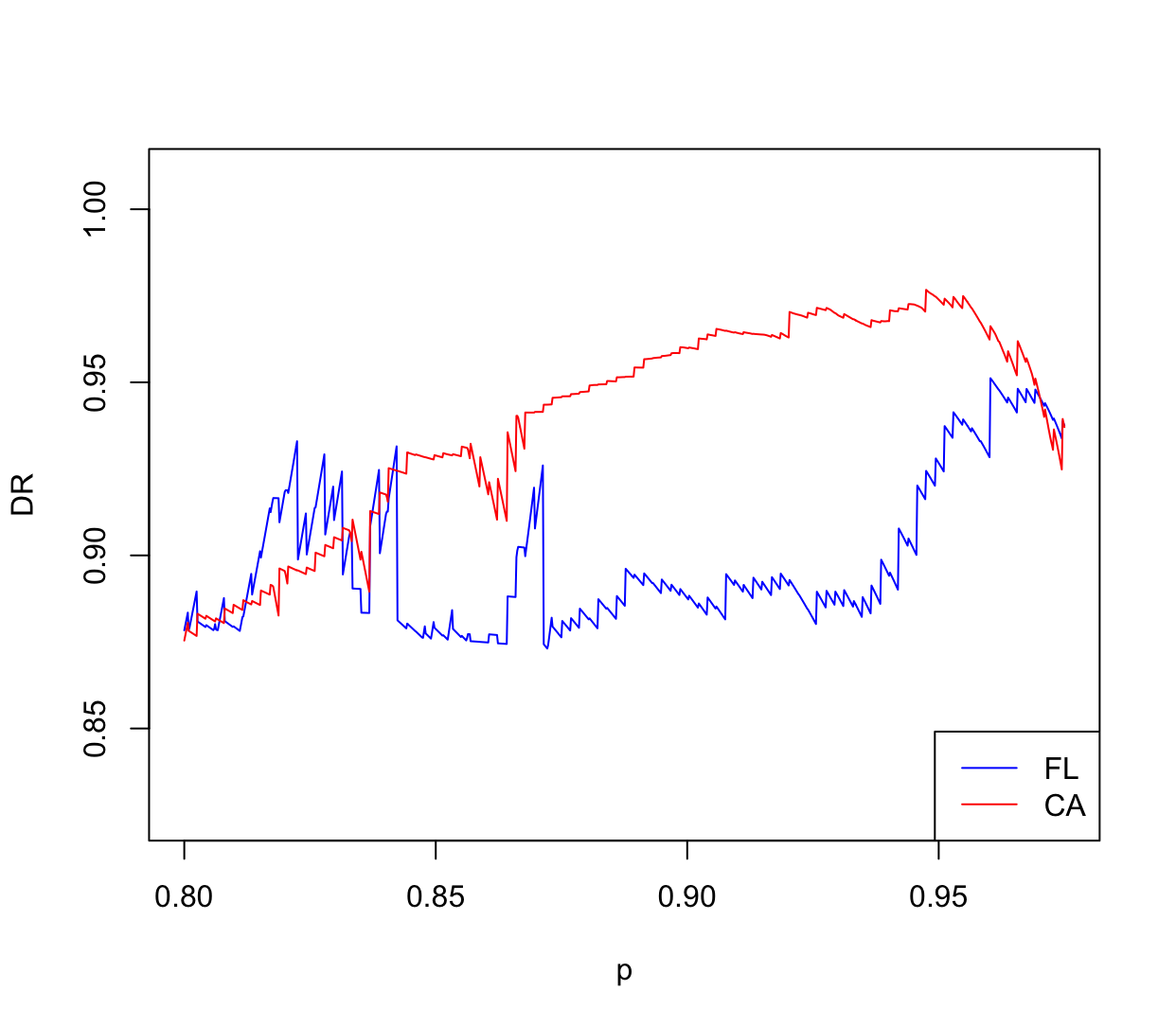}
    \caption{\small$\xi=0.1^{1/\widehat{\alpha}_{FL-CA}}$}
    \label{fig:pool1_result1}
    \end{subfigure}
    \hfill
    \begin{subfigure}{0.45\textwidth}
    \includegraphics[width=\linewidth]{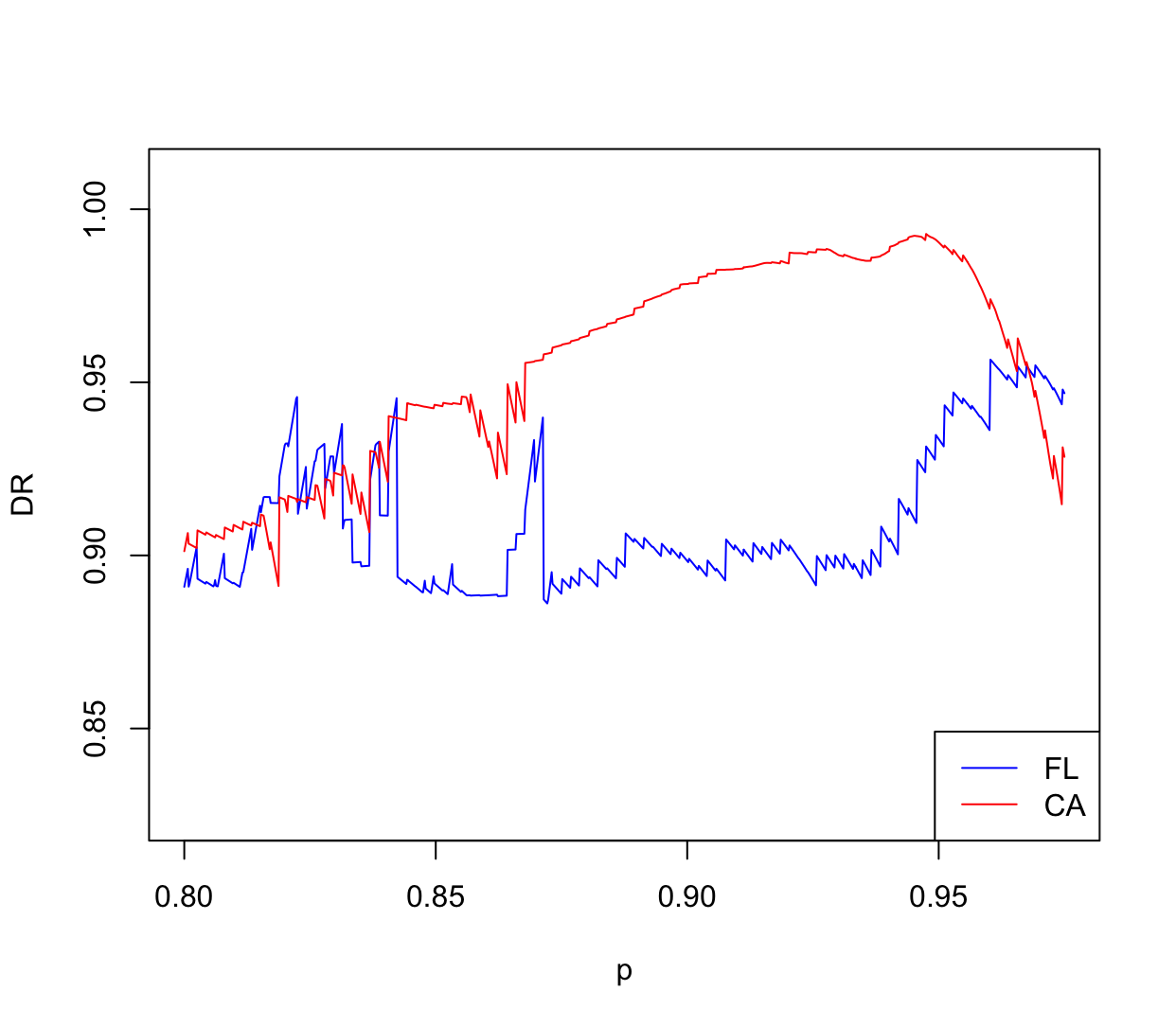}
    \caption{\small$\xi=0.3^{1/\widehat{\alpha}_{FL-CA}}$}
    \label{fig:pool1_result3}
    \end{subfigure}

    \begin{subfigure}{0.45\textwidth}
    \includegraphics[width=\linewidth]{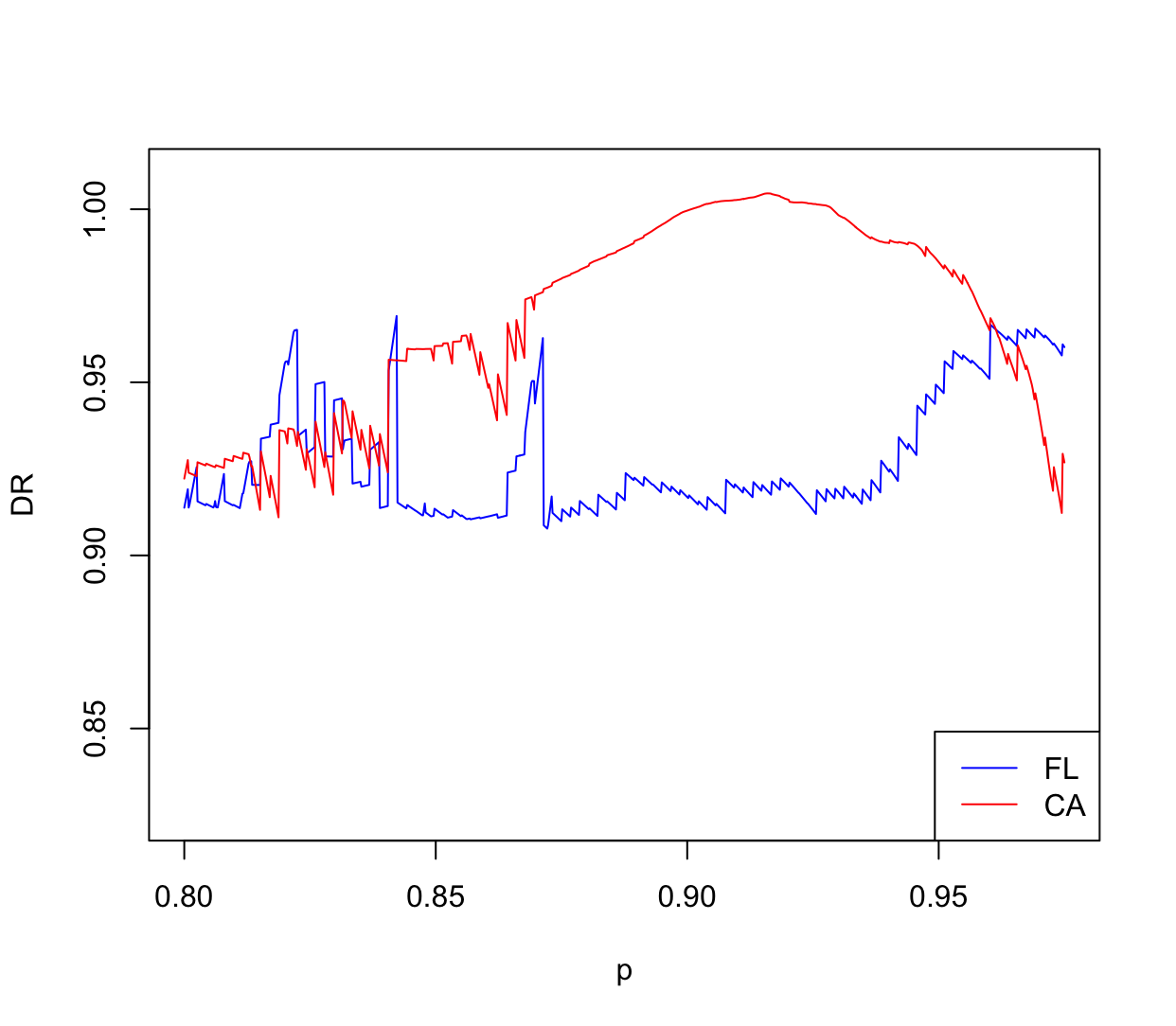}
    \caption{\small$\xi=0.5^{1/\widehat{\alpha}_{FL-CA}}$}
    \label{fig:pool1_result5}
    \end{subfigure}
    \hfill
    \begin{subfigure}{0.45\textwidth}
    \includegraphics[width=\linewidth]{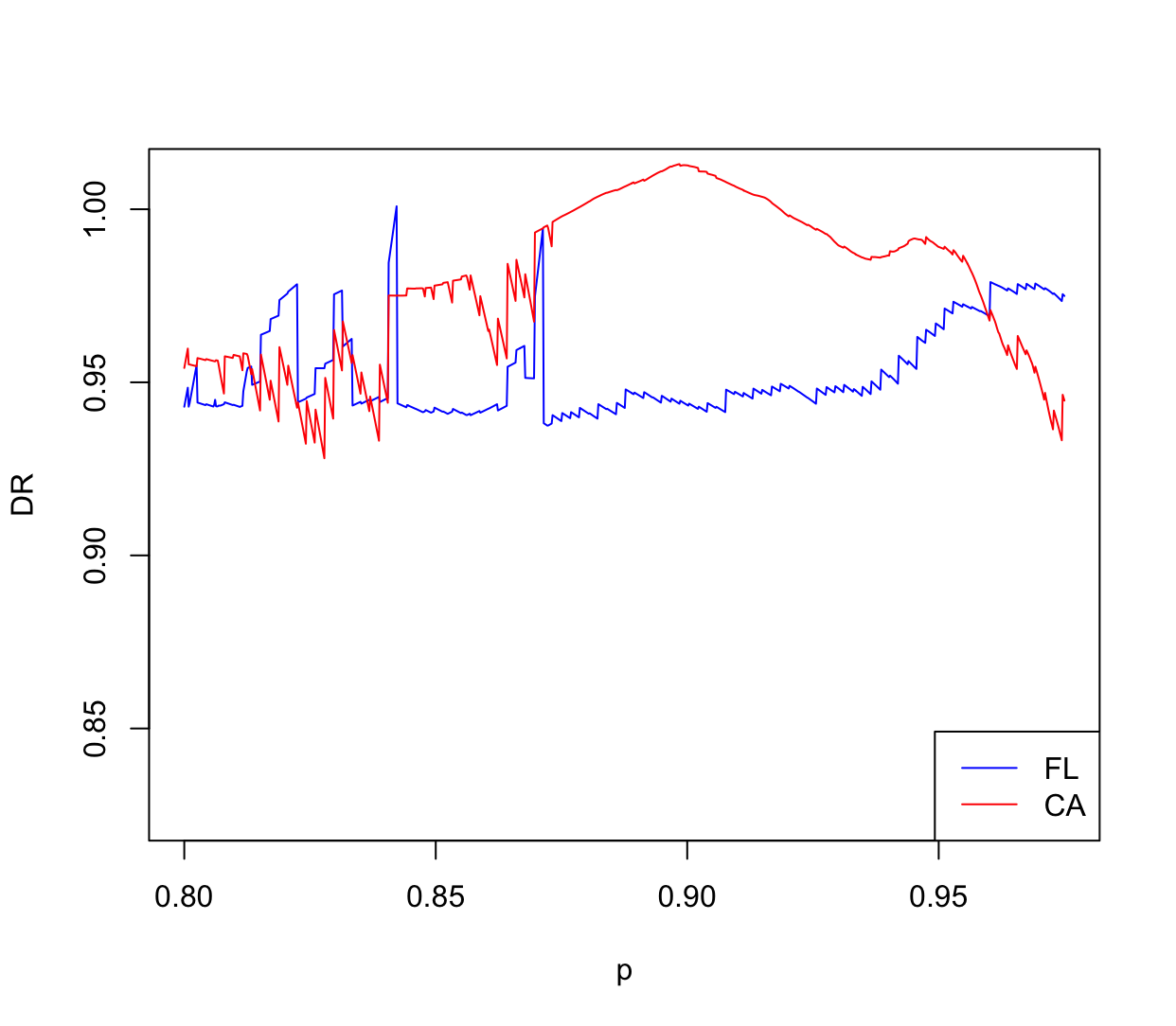}
    \caption{\small$\xi=0.7^{1/\widehat{\alpha}_{FL-CA}}$}
    \label{fig:pool1_result7}
    \end{subfigure}

    \captionsetup{justification=raggedright,singlelinecheck=false}
    \caption{\small$\mathrm{DR}_i(p)$ in Pool 1 for selected values of $\xi$ }
    
    \label{fig:pool1_result}
\end{figure}

{Figure \ref{fig:pool1_result} shows $\mathrm{DR}_{i}(p)$ for each participant
in Pool 1 (FL and CA) for {$p$ ranging from 0.8 to 0.975 at a stepsize of 0.001}, where $\xi$ is set as $0.1^{1/\widehat{\alpha}_{1}},$
$0.3^{1/\widehat{\alpha}_{1}},$ $0.5^{1/\widehat{\alpha}_{1}},$ or
$0.7^{1/\widehat{\alpha}_{1}}$. We observe that, in all scenarios, $\mathrm{DR}_{i}(p)$ is generally smaller
than 1 for both FL and CA, which means that all participants in the pool obtain a
diversification benefit from joining the pool. For smaller values of $\xi$,
$\mathrm{DR}_{i}(p)$ of FL and CA are smaller as well. This means that a higher
diversification benefit is achieved when the lower layer loss is brought to
the pool, which is consistent with our insight from Theorem \ref{thm:opt_sol}.
Furthermore, we notice that $\mathrm{DR}%
_{FL}(p)$ is generally smaller than $\mathrm{DR}%
_{CA}(p)$, which means that FL obtains more diversification benefit
than CA. As $\widehat
{\theta}_{CA}<1$, this implies that a participant in the pool with a larger loss enjoys a greater diversification benefit from the pool.}

\begin{figure} [htbp]
    \centering

    \begin{subfigure}{0.45\textwidth}
    \includegraphics[width=\linewidth]{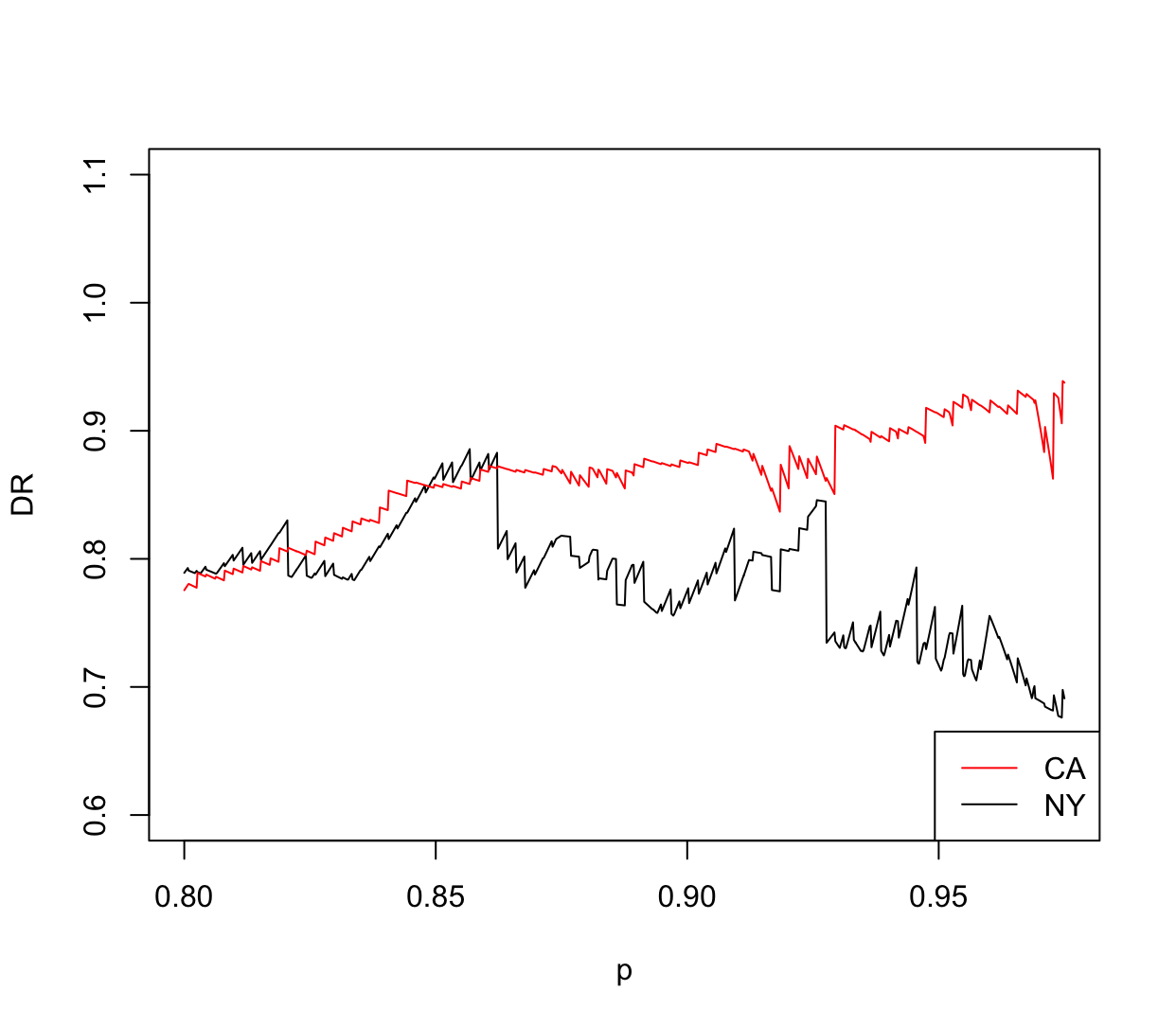}
    \caption{\small$\xi=0.1^{1/\widehat{\alpha}_{CA}}$}
    \label{fig:pool2_result1}
    \end{subfigure}
    \hfill
    \begin{subfigure}{0.45\textwidth}
    \includegraphics[width=\linewidth]{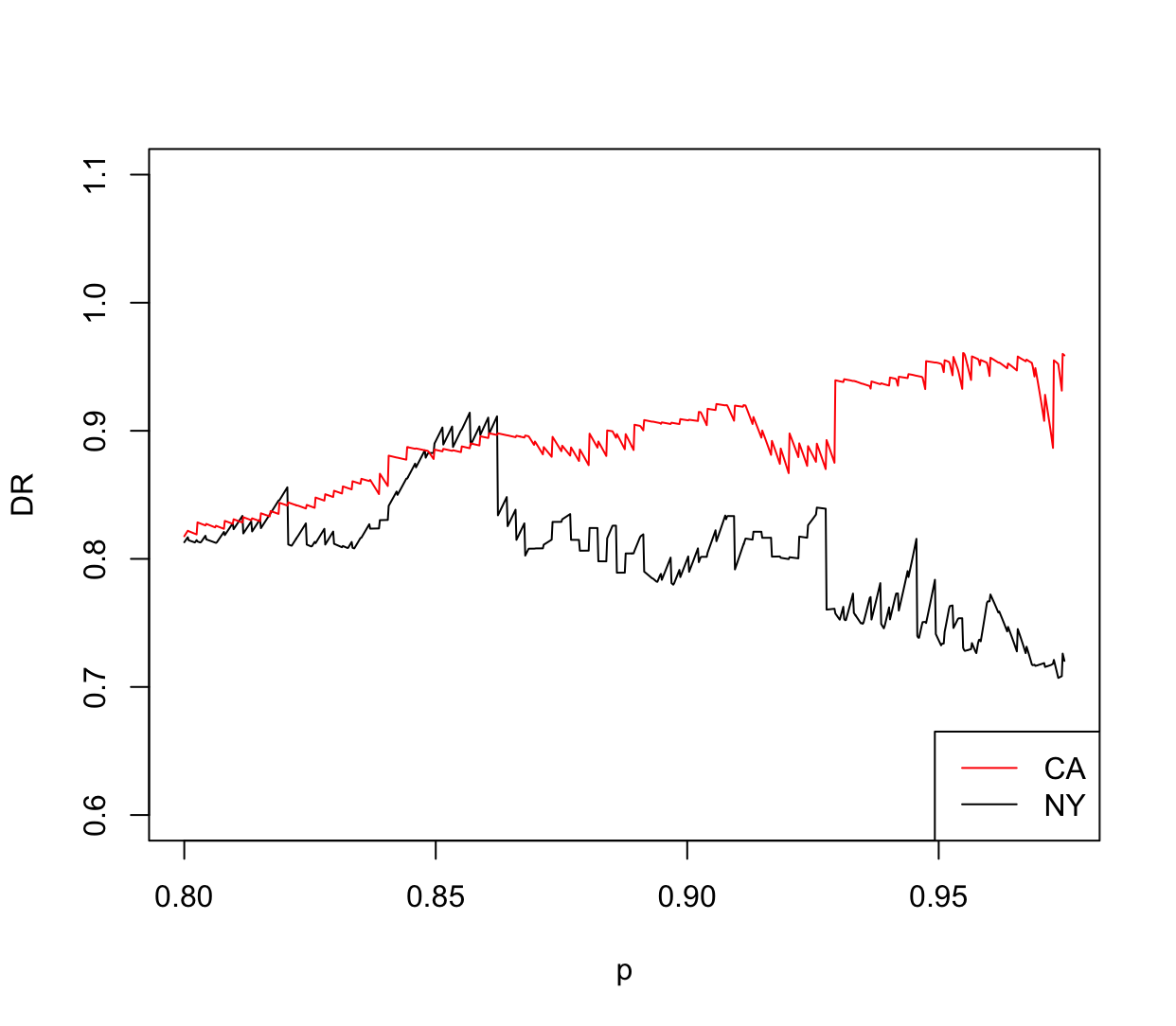}
    \caption{\small$\xi=0.3^{1/\widehat{\alpha}_{CA}}$}
    \label{fig:pool2_result3}
    \end{subfigure}

    \begin{subfigure}{0.45\textwidth}
    \includegraphics[width=\linewidth]{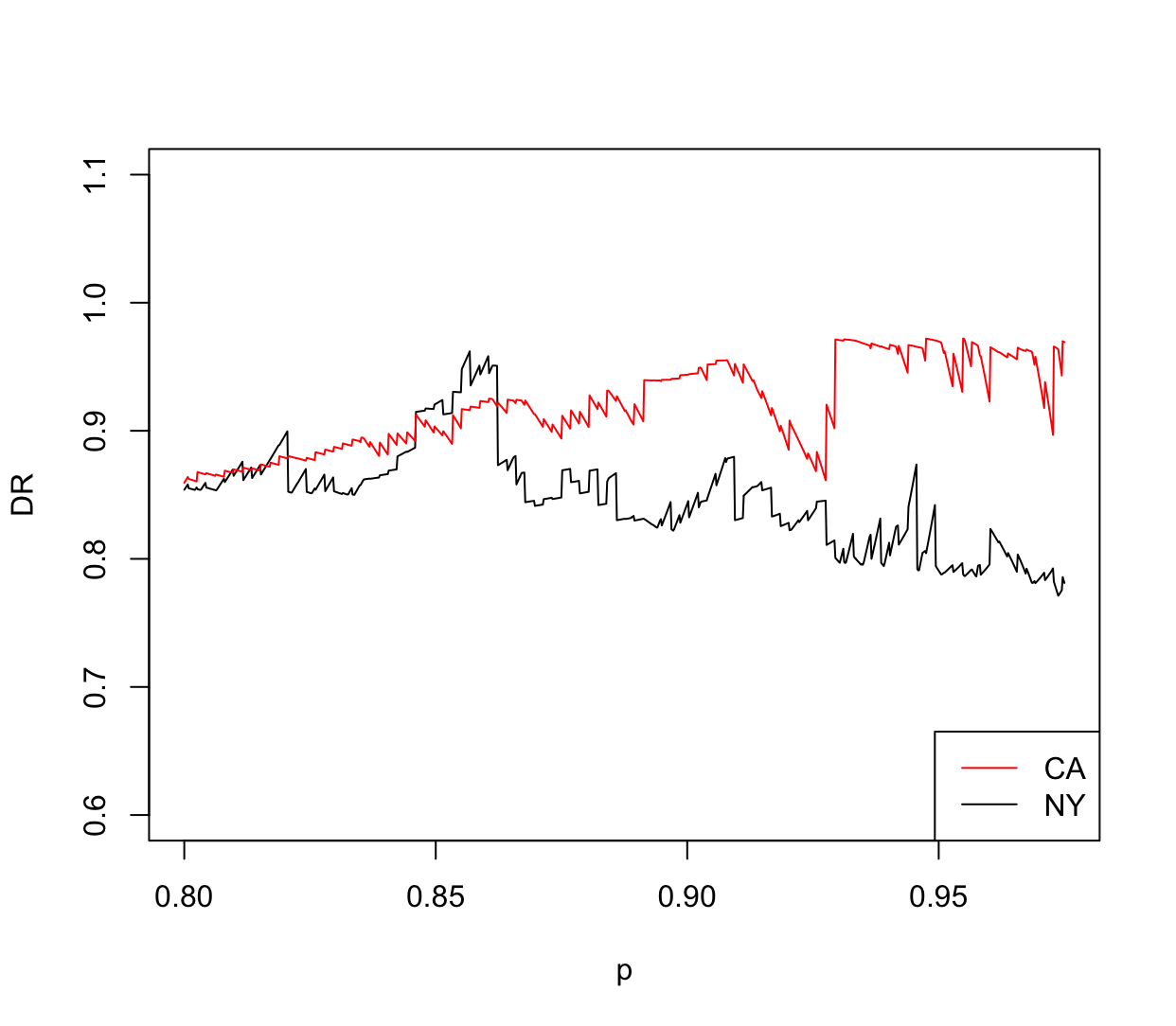}
    \caption{\small$\xi=0.5^{1/\widehat{\alpha}_{CA}}$}
    \label{fig:pool2_result5}
    \end{subfigure}
    \hfill
    \begin{subfigure}{0.45\textwidth}
    \includegraphics[width=\linewidth]{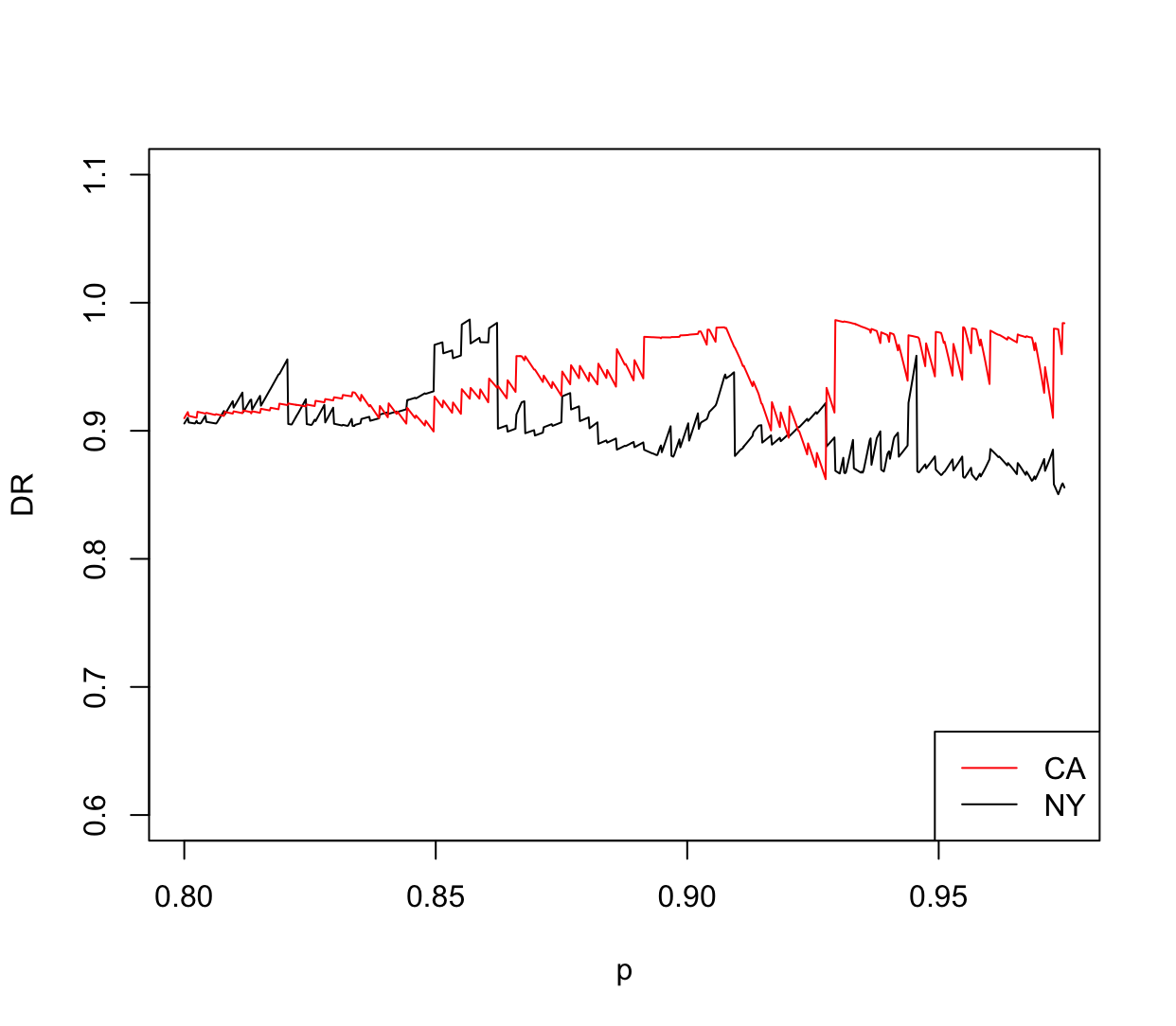}
    \caption{\small$\xi=0.7^{1/\widehat{\alpha}_{CA}}$}
    \label{fig:pool2_result7}
    \end{subfigure}

    \captionsetup{justification=raggedright,singlelinecheck=false}
    \caption{\small$\mathrm{DR}_i(p)$ in Pool 2 for selected values of $\xi$}
    
    \label{fig:pool2_result}
\end{figure}

Figure \ref{fig:pool2_result} shows $\mathrm{DR}_{i}(p)$ for each participant
in Pool 2 (CA and NY) for {$p$ ranging from 0.8 to 0.975 at a stepsize of 0.001}, where $\xi$ is set as $0.1^{1/\widehat{\alpha}_{1}},$
$0.3^{1/\widehat{\alpha}_{1}},$ $0.5^{1/\widehat{\alpha}_{1}},$ or
$0.7^{1/\widehat{\alpha}_{1}}$. Similarly to Pool 1, in all scenarios,
$\mathrm{DR}_{i}(p)$ is generally smaller than 1 for both CA and NY, and that
they are able to reach smaller values when $\xi$ is lower.

\begin{figure} [htbp]
    \centering

    \begin{subfigure}{0.45\textwidth}
    \includegraphics[width=\linewidth]{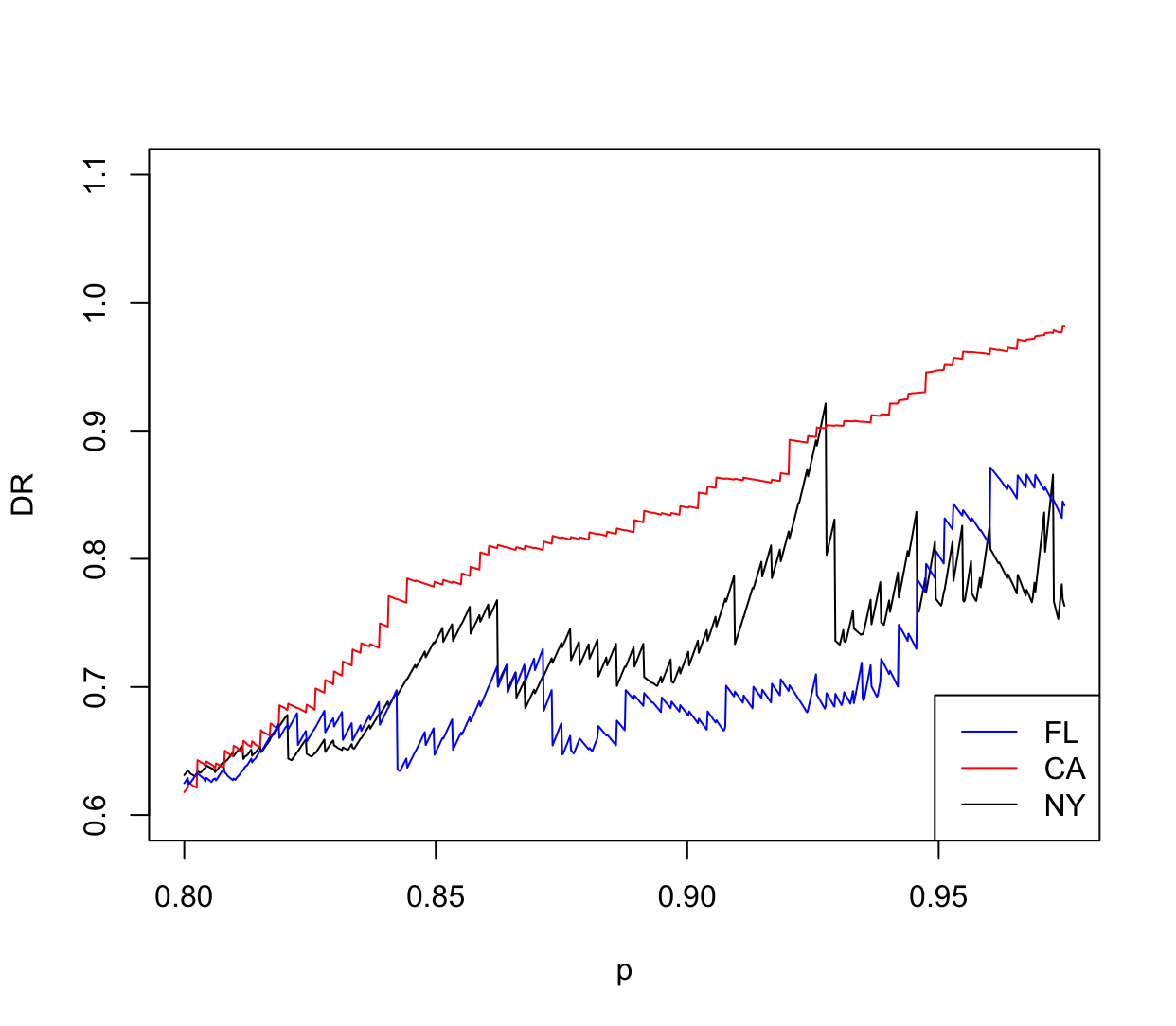}
    \caption{\small$\xi=0.1^{1/\widehat{\alpha}_{FL}}$}
    \label{fig:pool_all_result1}
    \end{subfigure}
    \hfill
    \begin{subfigure}{0.45\textwidth}
    \includegraphics[width=\linewidth]{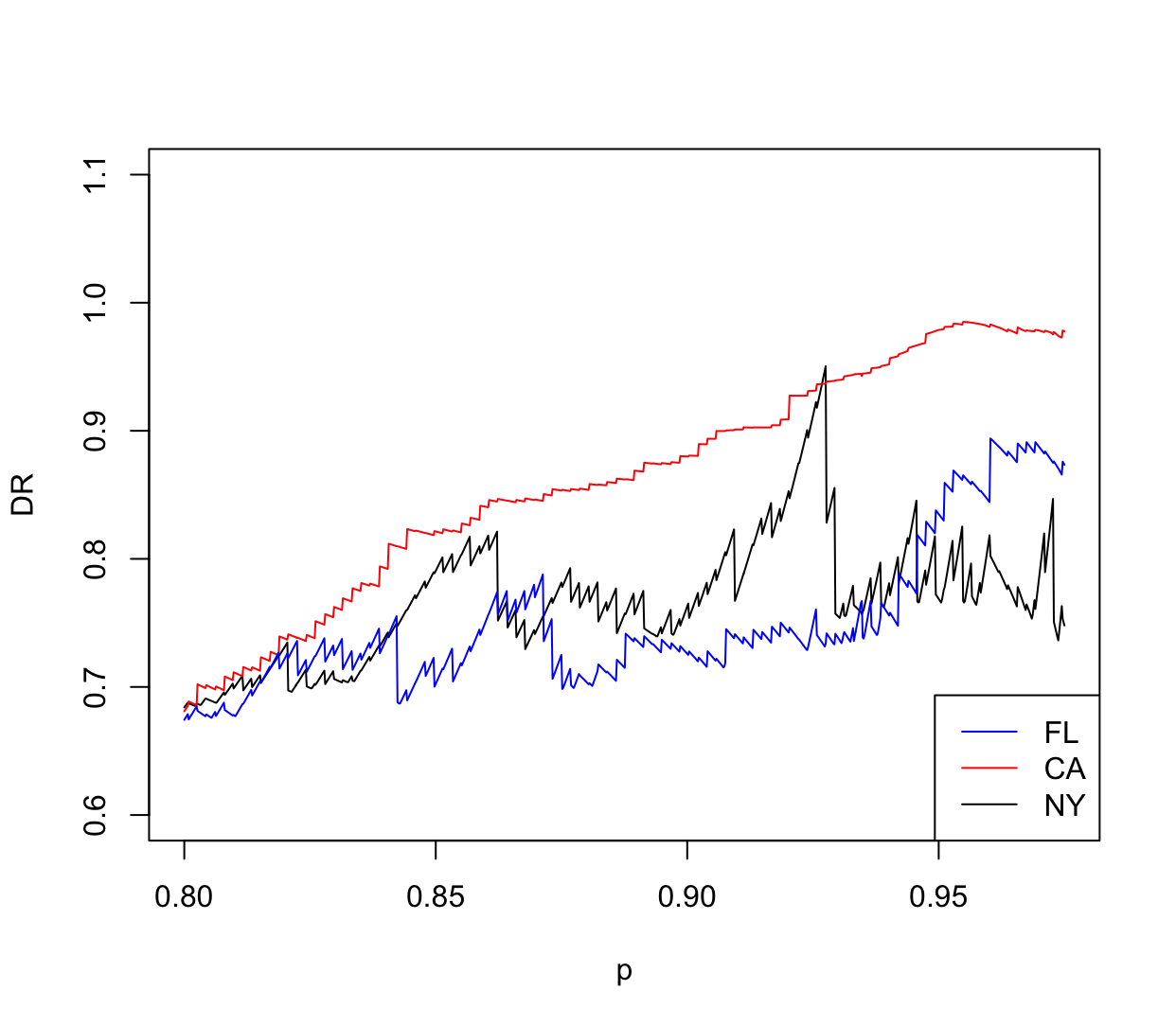}
    \caption{\small$\xi=0.3^{1/\widehat{\alpha}_{FL}}$}
    \label{fig:pool_all_result3}
    \end{subfigure}

    \begin{subfigure}{0.45\textwidth}
    \includegraphics[width=\linewidth]{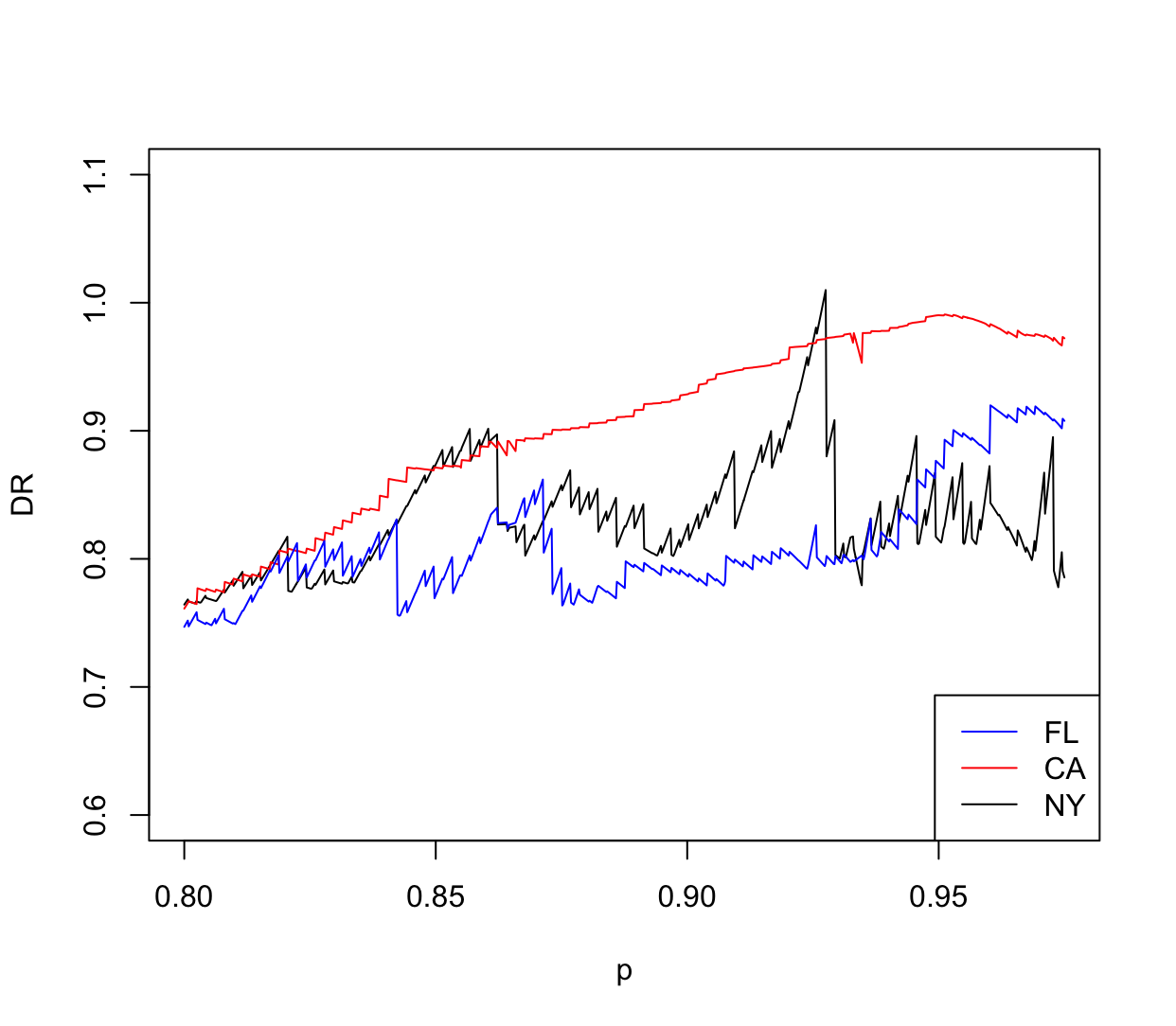}
    \caption{\small$\xi=0.5^{1/\widehat{\alpha}_{FL}}$}
    \label{fig:pool_all_result5}
    \end{subfigure}
    \hfill
    \begin{subfigure}{0.45\textwidth}
    \includegraphics[width=\linewidth]{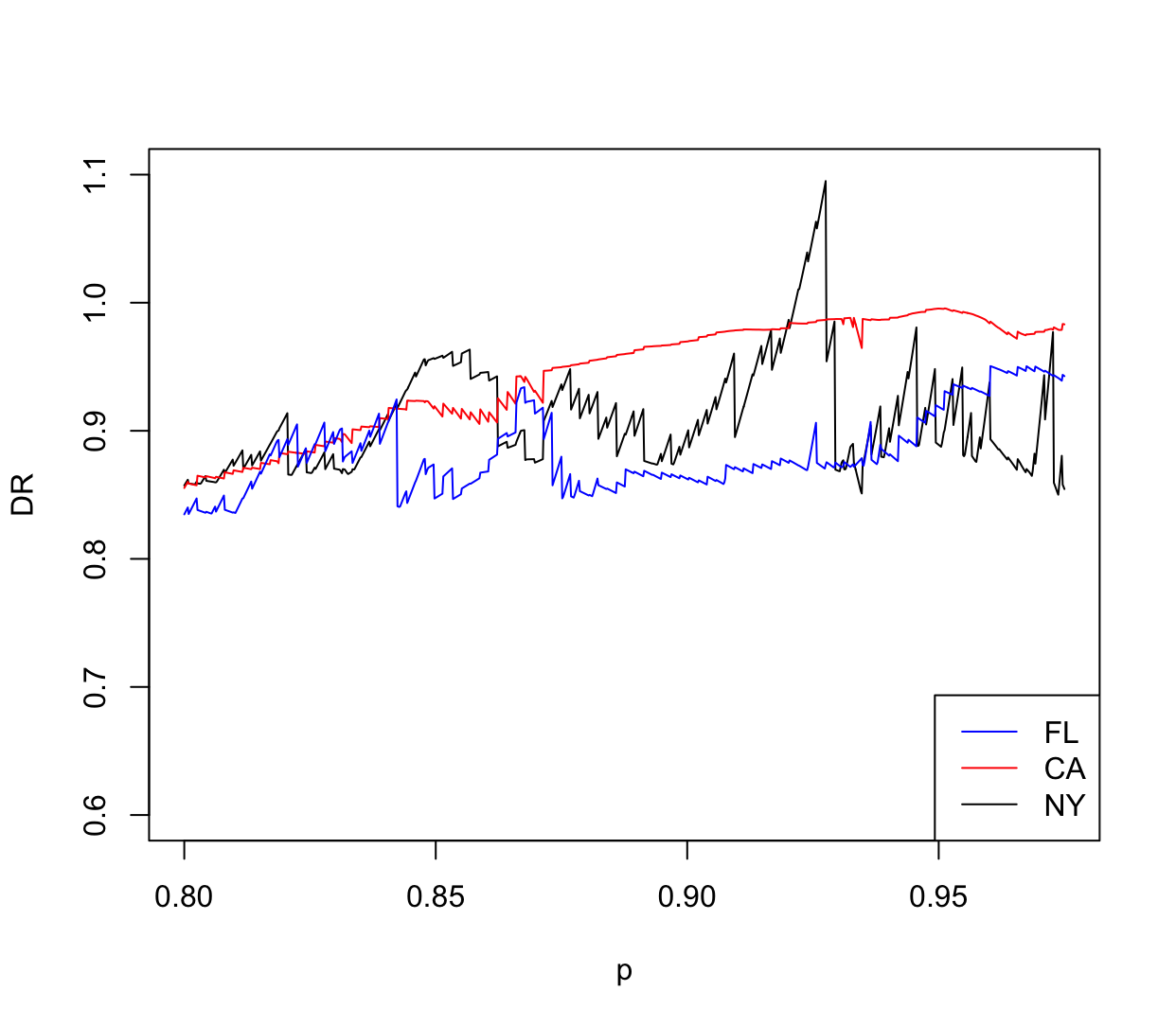}
    \caption{\small$\xi=0.7^{1/\widehat{\alpha}_{FL}}$}
    \label{fig:pool_all_result7}
    \end{subfigure}

    \captionsetup{justification=raggedright,singlelinecheck=false}
    \caption{\small$\mathrm{DR}_i(p)$ in Pool 3 for selected values of $\xi$}
    
    \label{fig:pool_all_result}
\end{figure}

Finally, Figure \ref{fig:pool_all_result} shows $\mathrm{DR}_{i}(p)$ for {$p$ ranging from 0.8 to 0.975 at a stepsize of 0.001}, where $\xi$ is set as $0.1^{1/\widehat{\alpha}_{1}},$
$0.3^{1/\widehat{\alpha}_{1}},$ $0.5^{1/\widehat{\alpha}_{1}},$ or
$0.7^{1/\widehat{\alpha}_{1}}$. Similarly to Pool 1 and Pool 2, in all
scenarios $\mathrm{DR}_{i}(p)$ is generally smaller than 1 and that they are
smaller and more different from each other when $\xi$ is closer to 0.

\section{Conclusion}
In this paper, we investigate how to optimally allocate the diversification
benefit among participants in a catastrophe risk pool. By joining the pool,
each participant receives coverage for a layer of loss characterized by an
attachment point and a limit. In return for this coverage, each participant
pays a premium proportional to the aggregated loss of the pool. The
diversification benefit is measured by a so-called diversification ratio
based on the VaR measure, which compares a participant's risk before and
after joining the pool. Achieving the optimal diversification benefit for
all participants at a given level $p$ of VaR requires solving a
high-dimensional optimization problem, for which an analytical solution is
generally unavailable, while numerical solutions can be time-consuming and
potentially unreliable. We therefore propose using the asymptotically
optimal pool, in which the diversification ratio is evaluated in the limit,
to approximate the practical optimal pool at any finite level $p$ close to
1. Through simulation and empirical studies, we show that the asymptotically
optimal pool provides an accurate and reliable approximation to the
practical optimal pool and that its implementation is relatively
straightforward.

\begin{appendices}
\renewcommand*\appendixpagename{\Large Appendices}
\appendixpage

\section{Proofs of analytical results}\label{Appendix_proof}

\subsection{Proofs of results under Model 1}

\begin{proof}[Proof of Lemma \ref{lem1}]

By (\ref{TE}) and Potter's bounds (see Proposition B.1.9(5) of \cite{haan2006extreme} for example), we have%
\[
1=\lim_{p\rightarrow1}\frac{\overline{F}_{i}\left(  \mathrm{VaR}_{p}%
(X_{i})\right)  }{\overline{F}_{1}\left(  \mathrm{VaR}_{p}(X_{1})\right)
}=\lim_{p\rightarrow1}\frac{\overline{F}_{i}\left(  \mathrm{VaR}_{p}%
(X_{i})\right)  }{\overline{F}_{1}\left(  \mathrm{VaR}_{p}(X_{i})\right)
}\frac{\overline{F}_{1}\left(  \mathrm{VaR}_{p}(X_{i})\right)  }{\overline
{F}_{1}\left(  \mathrm{VaR}_{p}(X_{1})\right)  }=\lim_{p\rightarrow1}%
\theta_{i}\left(  \frac{\mathrm{VaR}_{p}(X_{i})}{\mathrm{VaR}_{p}(X_{1}%
)}\right)  ^{-\alpha}.
\]
This leads to the following result
\[
\lim_{p\rightarrow1}\frac{\mathrm{VaR}_{p}(X_{i})}{\mathrm{VaR}_{p}(X_{1}%
)}=\theta_{i}^{1/\alpha}.
\]
Then under assumption (\ref{a1}), we have
\[
\lim_{p\rightarrow1}\frac{d_{i}(p)}{d_{1}(p)}=\lim_{p\rightarrow1}\frac
{d_{i}(p)}{\mathrm{VaR}_{p}(X_{i})}\frac{\mathrm{VaR}_{p}(X_{1})}{d_{1}%
(p)}\frac{\mathrm{VaR}_{p}(X_{i})}{\mathrm{VaR}_{p}(X_{1})}=\theta
_{i}^{1/\alpha}.
\]
Similarly, by Potter's bounds, we have
\[
\lim_{p\rightarrow1}\frac{\overline{F}_{i}(d_{i})}{\overline{F}_{1}(d_{1}%
)}=\lim_{p\rightarrow1}\frac{\overline{F}_{i}(d_{i})}{\overline{F}_{1}(d_{i}%
)}\frac{\overline{F}_{1}(d_{i})}{\overline{F}_{1}(d_{1})}=\lim_{p\rightarrow
1}\theta_{i}\left(  \frac{d_{i}(p)}{d_{1}(p)}\right)  ^{-\alpha}=1.
\]
Lastly, we have
\[
\lim_{p\rightarrow1}\frac{d_{i}(p)}{\mathrm{VaR}_{p}(X_{1})}=\lim
_{p\rightarrow1}\frac{d_{i}(p)}{\mathrm{VaR}_{p}(X_{i})}\frac{\mathrm{VaR}%
_{p}(X_{i})}{\mathrm{VaR}_{p}(X_{1})}=\theta_{i}^{1/\alpha}\xi
\]
and 
\[
\lim_{p\rightarrow1}\frac{l_{i}(p)}{\mathrm{VaR}_{p}(X_{1})}=\lim
_{p\rightarrow1}\frac{l_{i}(p)}{d_{i}(p)}\frac{d_{i}(p)}{\mathrm{VaR}%
_{p}(X_{1})}=\lambda_{i}\theta_{i}^{1/\alpha}\xi.
\]
by assumption (\ref{a1}). This completes the proof.\bigskip
\end{proof}

\begin{proof}[Proof of Theorem \ref{theorem3}]

Consider the split
\begin{align}
\mathrm{DR}_{i} (p) &  =\frac{\mathrm{VaR}_{p}(X_{i}-Y_{i})}{\mathrm{VaR}%
_{p}(X_{i})}+\frac{E\left[  Y_{i}\right]  }{E\left[  S_{n}\right]  }\cdot
\frac{\mathrm{VaR}_{p}(S_{n})}{\mathrm{VaR}_{p}(X_{i})}\nonumber \\
&  :=I_{1}+I_{2}\cdot I_{3}. \label{split}%
\end{align}
Next we analyze each term separately.

First we consider $I_{1}$. For any $0<p<1$, we have
\[
I_{1}=\left \{
\begin{array}
[c]{lll}%
1, &  & p\leq F_{i}(d_{i}),\\
\frac{d_{i}}{\mathrm{VaR}_{p}(X_{i})}, &  & F_{i}(d_{i})<p\leq F_{i}(l_{i}),\\
1-\frac{l_{i}-d_{i}}{\mathrm{VaR}_{p}(X_{i})}, &  & p>F_{i}(l_{i}).
\end{array}
\right.
\]
This leads to the following expression
\[
\lim_{p\rightarrow1}I_{1}=\left \{
\begin{array}
[c]{lll}%
1, &  & \xi \geq1,\\
\xi, &  & 1/\lambda_{i}\leq \xi<1,\\
1-(\lambda_{i}-1)\xi, &  & \xi<1/\lambda_{i}.
\end{array}
\right.
\]

Now we turn to $I_{2}$. Note that%
\[
E[Y_{i}]=\int_{d_{i}}^{l_{i}}\overline{F}_{i}(x)dx=d_{i}\int_{1}^{l_{i}/d_{i}%
}\overline{F}_{i}\left(  d_{i}x\right)  dx.
\]
By Potter's bounds (see Proposition B.1.9(5) of \cite{haan2006extreme} for example) and assumption (\ref{a1}),
for any $\varepsilon,\delta>0$, there exists $p_{0}>0$ and $d_{0}>0$, such
that for $p_{0}<p<1$, we have $d_{i}(p)>d_{0}$, $l_{i}(p)/d_{i}(p)<\lambda
_{i}+\varepsilon$ and
\[
\frac{\overline{F}_{i}\left(  d_{i}x\right)  }{\overline{F}_{i}\left(
d_{i}\right)  }\leq(1+\varepsilon)x^{-\alpha+\delta}%
\]
for all $x>1$. Then for $p_{0}<p<1$, we have
\[
0 \leq \frac{E[Y_{i}]}{d_{i}\overline{F}_{i}\left(  d_{i}\right)  }\leq \int
_{1}^{\lambda_{i}+\varepsilon}(1+\varepsilon)x^{-\alpha+\delta}dx<\infty
\]
for all $\alpha>0$. Adhering to the Lebesgue dominated convergence theorem gives us
\begin{equation}
\lim_{p\rightarrow1}\frac{E[Y_{i}]}{d_{i}\overline{F}_{i}\left(  d_{i}\right)
}=\int_{1}^{\lambda_{i}}x^{-\alpha}dx=\frac{\lambda_{i}^{1-\alpha}-1}
{1-\alpha}. \label{limit1}%
\end{equation}
If $\alpha=1$, the right-hand side of (\ref{limit1}) is understood to be
$\ln \lambda_{i}$. Then, by (\ref{r1}), we have%
\[
\lim_{p\rightarrow1}\frac{E[Y_{i}]}{d_{1}\overline{F}_{1}\left(  d_{1}\right)
}=\lim_{p\rightarrow1}\frac{E[Y_{i}]}{d_{i}\overline{F}_{i}\left(
d_{i}\right)  }\frac{d_{i}\overline{F}_{i}\left(  d_{i}\right)  }%
{d_{1}\overline{F}_{1}\left(  d_{1}\right)  }=\frac{\left(  \lambda
_{i}^{1-\alpha}-1\right)  \theta_{i}^{1/\alpha}}{1-\alpha}.
\]
Moreover, we also have
\[
\lim_{p\rightarrow1}\frac{E\left[  S_{n}\right]  }{d_{1}\overline{F}%
_{1}\left(  d_{1}\right)  }=\lim_{p\rightarrow1}\sum_{j=1}^{n}\frac{E[Y_{j}%
]}{d_{1}\overline{F}_{1}\left(  d_{1}\right)  }=\sum_{j=1}^{n}\frac{\left(
\lambda_{j}^{1-\alpha}-1\right)  \theta_{j}^{1/\alpha}}{1-\alpha}.
\]
It follows that
\[
\lim_{p\rightarrow1}I_{2}=\frac{\left(  \lambda_{i}^{1-\alpha}-1\right)
\theta_{i}^{1/\alpha}}{\sum_{j=1}^{n}\left(  \lambda_{j}^{1-\alpha}-1\right)
\theta_{j}^{1/\alpha}}.
\]

Lastly, we consider $I_{3}$. To analyze $\mathrm{VaR}_{p}(S_{n})$, we start
with $\mathbb{P}(S_{n}>t)$ and compare it with $\overline{F}_1(t)$, where $t=\mathrm{VaR}_{p}(X_{1})$. Then $p\rightarrow1$ can be replaced by $t\rightarrow\infty$. Denote $J=\{1,....n\}$ and $\mathbb{J}=\{
\left(  J_{1},J_{2},J_{3}\right)  :J_{1}\cup J_{2}\cup J_{3}=J$ and
$J_{1},J_{2},J_{3}$ mutually exclusive$\}$. Note that at most two of $J_{1}$,
$J_{2}$ and $J_{3}$ can be empty sets. Let $\mathbb{J}_{1}=\left \{
X_{i}\leq d_{i}:i\in J_{1}\right \}  $, $\mathbb{J}_{2}=\left \{  d_{i}<
X_{i}\leq l_{i}:i\in J_{2}\right \}  $ and $\mathbb{J}_{3}=\left \{  X_{i}>
l_{i}:i\in J_{3}\right \}  $. Then we have
\begin{align*}
\mathbb{P}(S_{n}>t)  &  =\underset{\left(  J_{1},J_{2},J_{3}\right)  \in \mathbb{J}%
}{\sum}\mathbb{P} \left(  \underset{i\in J_{1}}{\sum}Y_{i}+\underset{j\in J_{2}}{\sum
}Y_{j}+\underset{k\in J_{3}}{\sum}Y_{k}>t\right) \\
&  =\underset{\left(  J_{1},J_{2},J_{3}\right)  \in \mathbb{J}}{\sum}\mathbb{P} \left(
\underset{j\in J_{2}}{\sum}X_{j}>t+\underset{j\in J_{2}}{\sum}d_{j}%
-\underset{k\in J_{3}}{\sum}(l_{k}-d_{k}),\mathbb{J}_{1}\cap \mathbb{J}_{2}%
\cap \mathbb{J}_{3}\right) \\
&  =\underset{\left(  J_{1},J_{2},J_{3}\right)  \in \mathbb{J}}{\sum}I.
\end{align*}
Note that if $J_{2}=J_{3}=\varnothing$, then $I=0$. Next, we show that if at
least one of $J_{2}$ and $J_{3}$ is not an empty set, then the cardinality of
$J_{2}\cup J_{3}$ is at most 1 such that
\[
\lim_{t\rightarrow \infty}\frac{I}{\overline{F}_{1}\left(  t\right)  }\not =0.
\]
It suffices to show the following cases for some $1\leq i\not =j\leq n$ and
$i$ or $j\not \in J_{1}$: (1) $J_{2}=\{i\}$ and $J_{3}=\varnothing$; (2)
$J_{2}=\{i,j\}$ and $J_{3}=\varnothing$; (3) $J_{2}=\varnothing$ and
$J_{3}=\{i\}$; (4) $J_{2}=\varnothing$ and $J_{3}=\{i,j\}$; (5) $J_{2}=\{i\}$
and $J_{3}=\{j\}$.

(1) In this case, since $X_{i}$'s are independent, we obtain the following result
\[
I=\mathbb{P} \left(  X_{i}>t+d_{i},d_{i}<X_{i}\leq l_{i}\right)  \mathbb{P} \left(
\mathbb{J}_{1}\right)  .
\]
For $t>l_{i}-d_{i}$ which is equivalent to $\xi<\theta_{i}^{-1/\alpha}\left(
\lambda_{i}-1\right)  ^{-1}$, we have $I=0$. For $0<t<l_{i}-d_{i}$ which is
equivalent to $\xi>\theta_{i}^{-1/\alpha}\left(  \lambda_{i}-1\right)  ^{-1}$,
since $t=\mathrm{VaR}_{p}(X_{1})$, by Lemma \ref{lem1} we have%
\begin{align*}
\lim_{t\rightarrow \infty}\frac{I}{\overline{F}_{1}\left(  t\right)  }  &
=\lim_{t\rightarrow \infty}\left(  \frac{\overline{F}_{i}\left(  t+d_{i}%
\right)  }{\overline{F}_{1}\left(  t\right)  }-\frac{\overline{F}_{i}\left(
l_{i}\right)  }{\overline{F}_{1}\left(  t\right)  }\right)
{\displaystyle \prod \limits_{k\in J_{1}}}
F_{k}(d_{k})\\
&  =\lim_{t\rightarrow \infty}\left(  \frac{\overline{F}_{i}\left(
t+d_{i}\right)  }{\overline{F}_{i}\left(  t\right)  }\frac{\overline{F}%
_{i}\left(  t\right)  }{\overline{F}_{1}\left(  t\right)  }-\frac{\overline
{F}_{i}\left(  l_{i}\right)  }{\overline{F}_{i}\left(  t\right)  }%
\frac{\overline{F}_{i}\left(  t\right)  }{\overline{F}_{1}\left(  t\right)
}\right)  \lim_{t\rightarrow \infty}%
{\displaystyle \prod \limits_{k\in J_{1}}}
F_{k}(d_{k})\\
&  =\left(  \theta_{i}^{-1/\alpha}+\xi \right)  ^{-\alpha}-\left(  \xi
\lambda_{i}\right)  ^{-\alpha}.
\end{align*}

(2) In this case, since $X_{i}$'s are independent, we obtain the following result
\begin{align*}
I  &  =\mathbb{P} \left(  X_{i}+X_{j}>t+d_{i}+d_{j},d_{i}< X_{i}\leq l_{i},d_{j}<
X_{j}\leq l_{j}\right)  \mathbb{P} \left(  \mathbb{J}_{1}\right) \\
&  =\mathbb{P} \left(  X_{i}+X_{j}>t+d_{i}+d_{j}\left \vert d_{i}< X_{i}\leq l_{i}%
,d_{j}< X_{j}\leq l_{j}\right.  \right)  \mathbb{P} \left(  d_{i}< X_{i}%
\leq l_{i}\right)  \mathbb{P} \left(  d_{j}< X_{j}\leq l_{j}\right)  \mathbb{P} \left(
\mathbb{J}_{1}\right)  .
\end{align*}
Then, following a similar proof to that for case (1), we have
\[
\lim_{t\rightarrow \infty}\frac{\mathbb{P} \left(  d_{i}< X_{i}\leq l_{i}\right)
}{\overline{F}_{1}\left(  t\right)  }\mathbb{P} \left(  d_{j}<X_{j}\leq l_{j}\right)
=0,
\]
which leads to
\begin{equation}
\lim_{t\rightarrow \infty}\frac{I}{\overline{F}_{1}\left(  t\right)  }=0.
\label{I0}%
\end{equation}
Thus, when the cardinality of $J_{2}$ is greater than $1$, (\ref{I0}) holds true.

(3) In this case, since $X_{i}$'s are independent, we obtain the following result
\[
I=\mathbb{P} \left(  l_{i}-d_{i}>t,X_{i}> l_{i}\right)  \mathbb{P} \left(  \mathbb{J}%
_{1}\right)  .
\]
For $t>l_{i}-d_{i}$ which is equivalent to $\xi<\theta_{i}^{-1/\alpha}\left(
\lambda_{i}-1\right)  ^{-1}$, we have $I=0$. For $0<t<l_{i}-d_{i}$ which is
equivalent to $\xi>\theta_{i}^{-1/\alpha}\left(  \lambda_{i}-1\right)  ^{-1}$,
we have%
\begin{align*}
\lim_{t\rightarrow \infty}\frac{I}{\overline{F}_{1}\left(  t\right)  }  &
=\lim_{t\rightarrow \infty}\frac{\overline{F}_{i}\left(  l_{i}\right)
}{\overline{F}_{1}\left(  t\right)  }%
{\displaystyle \prod \limits_{k\in J_{1}}}
F_{k}(d_{k})\\
&  =\left(  \xi \lambda_{i}\right)  ^{-\alpha}.
\end{align*}

(4) In this case, since $X_{i}$'s are independent, we have%
\begin{align*}
I  &  =\mathbb{P} \left(  l_{i}+l_{j}>t+d_{i}+d_{j},X_{i}> l_{i},X_{j}>
l_{j}\right)  \mathbb{P} \left(  \mathbb{J}_{1}\right) \\
&  =\mathbb{P} \left(  l_{i}+l_{j}>t+d_{i}+d_{j}\left \vert X_{i}> l_{i},X_{j}>
l_{j}\right.  \right)  \mathbb{P} \left(  X_{i}> l_{i}\right)  \mathbb{P} \left(  X_{j}>
l_{j}\right)  \mathbb{P} \left(  \mathbb{J}_{1}\right)  .
\end{align*}
Then, following a similar proof to that for case (3), we have
\[
\lim_{t\rightarrow \infty}\frac{\mathbb{P} \left(  X_{i}> l_{i}\right)  }%
{\overline{F}_{1}\left(  t\right)  }\mathbb{P} \left(  X_{j}> l_{j}\right)  =0,
\]
which leads to (\ref{I0}). Moreover, when the cardinality of $J_{3}$ is
greater than $1$, (\ref{I0}) holds as well.

(5) In this case, since $X_{i}$'s are independent, we have
\begin{align*}
I  &  =\mathbb{P} \left(  X_{i}+l_{j}>t+d_{i}+d_{j},d_{i}< X_{i}\leq l_{i},X_{j}>
l_{j}\right)  \mathbb{P} \left(  \mathbb{J}_{1}\right) \\
&  =\mathbb{P} \left(  X_{i}+l_{j}>t+d_{i}+d_{j}\left \vert d_{i}< X_{i}\leq l_{i}%
,X_{j}> l_{j}\right.  \right)  \mathbb{P} \left(  d_{i}< X_{i}\leq l_{i}\right)
\mathbb{P} \left(  X_{j}> l_{j}\right)  \mathbb{P} \left(  \mathbb{J}_{1}\right)  .
\end{align*}
Then following a similar proof to that for case (1), we have
\[
\lim_{t\rightarrow \infty}\frac{\mathbb{P} \left(  d_{i}< X_{i}\leq l_{i}\right)
}{\overline{F}_{1}\left(  t\right)  }\mathbb{P} \left(  X_{j}> l_{j}\right)  =0,
\]
which leads to (\ref{I0}). Moreover, when the cardinality of $J_{2} \cup J_{3}$ is greater than $1$, (\ref{I0}) holds as well.

To sum up, we obtain the following result
\begin{equation}
\lim_{t\rightarrow \infty}\frac{\mathbb{P}(S_{n}>t)}{\overline{F}_{1}\left(  t\right)
}=\sum_{j\in Z}\left(  \theta_{j}^{-1/\alpha}+\xi \right)  ^{-\alpha},
\label{sum}%
\end{equation}
where $Z=\left \{  j=1,2,...,n:\xi>\theta_{j}^{-1/\alpha}\left(  \lambda
_{j}-1\right)  ^{-1}\right \}  $.

Next, we obtain the limit of $\frac{\mathrm{VaR}_{p}(S_{n})}{\mathrm{VaR}%
_{p}(X_{1})}$ as $p\rightarrow1$. If $Z$ is an empty set, then $\lim
_{p\rightarrow1}\frac{\mathrm{VaR}_{p}(S_{n})}{\mathrm{VaR}_{p}(X_{1})}=0$.
When $Z$ is not empty, we denote the right-hand side of (\ref{sum}) by $a$, which
is not $0$. Then for any $\varepsilon>0$, there exits $t_{0}>0$ such that for
$t>t_{0}$, we have%
\begin{equation}
a\overline{F}_{1}\left(  t\right)  -\varepsilon \leq \mathbb{P}(S_{n}>t)\leq
a\overline{F}_{1}\left(  t\right)  +\varepsilon. \label{ineq}%
\end{equation}
Since
\[
\mathrm{VaR}_{p}(S_{n})=\inf \left \{  x:\mathbb{P}(S_{n}\leq x)\geq p\right \}
=\inf \left \{  x:\mathbb{P}(S_{n}>x)\leq1-p\right \}  ,
\]
(\ref{ineq}) implies that for $p$ close to 1, we have
\begin{equation}
\inf \left \{  x:\overline{F}_{1}\left(  x\right)  \leq \frac{1-p+\varepsilon}%
{a}\right \}  \leq \mathrm{VaR}_{p}(S_{n})\leq \inf \left \{  x:\overline{F}%
_{1}\left(  x\right)  \leq \frac{1-p-\varepsilon}{a}\right \}  . \label{ineq1}%
\end{equation}
Since (\ref{ineq1}) holds for arbitrary $\varepsilon$, we have
\[
\mathrm{VaR}_{p}(S_{n})\sim \mathrm{VaR}_{1-\left(  1-p\right)  /a}(X_{1}).
\]
This leads to the following result
\[
\lim_{p\rightarrow1}\frac{\mathrm{VaR}_{p}(S_{n})}{\mathrm{VaR}_{p}(X_{1}%
)}=\lim_{p\rightarrow1}\frac{\mathrm{VaR}_{p}(S_{n})}{\mathrm{VaR}_{1-\left(
1-p\right)  /a}(X_{1})}\frac{\mathrm{VaR}_{1-\left(  1-p\right)  /a}(X_{1}%
)}{\mathrm{VaR}_{p}(X_{1})}=a^{1/\alpha}
\]
{by the convergence of extreme quantile estimator (see e.g., Theorem 4.3.8 of \cite{haan2006extreme}).} When $a=0$, it reduces to the case that $Z$ is an empty set.
Similarly, from the assumption (\ref{TE}), we can show that
\[
\lim_{p\rightarrow1}\frac{\mathrm{VaR}_{p}(X_{i})}{\mathrm{VaR}_{p}(X_{1}%
)}=\theta_{i}^{1/\alpha}.
\]
Furthermore, we also have
\[
\lim_{p\rightarrow1}I_{3}=\lim_{p\rightarrow1}\frac{\mathrm{VaR}_{p}(S_{n}%
)}{\mathrm{VaR}_{p}(X_{1})}\frac{\mathrm{VaR}_{p}(X_{1})}{\mathrm{VaR}%
_{p}(X_{i})}=\left(  \frac{a}{\theta_{i}}\right)  ^{1/\alpha}=\theta
_{i}^{-1/\alpha}\left(  \sum_{j\in Z}\left(  \theta_{j}^{-1/\alpha}%
+\xi \right)  ^{-\alpha}\right)  ^{1/\alpha}.
\]
The desired result follows by combining all of the above three terms.
\end{proof}

\bigskip

\subsection{Proofs of results under Model 2}

\begin{proof}[Proof of Lemma \ref{lem3}]

If $\alpha_{i}=\alpha_{1}$, then all the results can be proved in the same way as that for Lemma \ref{lem1}. If $\alpha_{i}>\alpha_{1}$, then
\[
\lim_{t\rightarrow \infty}\frac{\overline{F}_{i}(t)}{\overline{F}_{1}%
(t)}=0,\qquad \text{and}\qquad \lim_{p\rightarrow1}\frac{\mathrm{VaR}_{p}%
(X_{i})}{\mathrm{VaR}_{p}(X_{1})}=0.
\]
Thus we obtain the following result
\[
\lim_{p\rightarrow1}\frac{d_{i}(p)}{d_{1}(p)}=\lim_{p\rightarrow1}\frac
{d_{i}(p)}{\mathrm{VaR}_{p}(X_{i})}\frac{\mathrm{VaR}_{p}(X_{1})}{d_{1}%
(p)}\frac{\mathrm{VaR}_{p}(X_{i})}{\mathrm{VaR}_{p}(X_{1})}=0,
\]
and similarly we have
\[
\lim_{p\rightarrow1}\frac{d_{i}(p)}{\mathrm{VaR}_{p}(X_{1})}=\lim
_{p\rightarrow1}\frac{l_{i}(p)}{\mathrm{VaR}_{p}(X_{1})}=0.
\]
Note that when $\alpha_{i}>\alpha_{1}$, we have $\theta_{i}=0$. Thus the
desired results hold. Now we are left to show $\lim_{p\rightarrow1}%
\frac{\overline{F}_{i}(d_{i})}{\overline{F}_{1}(d_{1})}=1$, which follows
from
\[
\lim_{p\rightarrow1}\frac{\overline{F}_{i}(d_{i})}{\overline{F}_{1}(d_{1}%
)}=\lim_{p\rightarrow1}\frac{\overline{F}_{i}(d_{i})/(1-p)}{\overline{F}%
_{1}(d_{1})/(1-p)}=\lim_{p\rightarrow1}\frac{\overline{F}_{i}(d_{i}%
)/\overline{F}_{i}\left(  \mathrm{VaR}_{p}(X_{i})\right)  }{\overline{F}%
_{1}(d_{1})/\overline{F}_{1}\left(  \mathrm{VaR}_{p}(X_{1})\right)  }=1
\]
by (\ref{a3}) and the definition of regular variation. This completes the proof.
\bigskip
\end{proof}

\begin{proof}[Proof of Theorem \ref{difthm}]

Consider the same split (\ref{split}) as in Theorem \ref{theorem3}. Applying the
analysis of $I_{1}$ to Model 2, we have%
\[
\lim_{p\rightarrow1}I_{1}=\left \{
\begin{array}
[c]{lll}%
1, &  & \xi^{\alpha_{1}/\alpha_{i}}\geq1,\\
\xi^{\alpha_{1}/\alpha_{i}}, &  & 1/\lambda_{i}\leq \xi^{\alpha_{1}/\alpha_{i}%
}<1,\\
1-(\lambda_{i}-1)\xi^{\alpha_{1}/\alpha_{i}}, &  & \xi^{\alpha_{1}/\alpha_{i}%
}<1/\lambda_{i}.
\end{array}
\right.
\]
For $I_{2} \cdot I_3$, on the one hand, the relation (\ref{limit1}) holds for $E[Y_{i}]$:
\[
\lim_{p\rightarrow1}\frac{E[Y_{i}]}{d_{i}\overline{F}_{i}\left(  d_{i}\right)
}=\frac{\lambda_{i}^{1-\alpha_{i}}-1}{1-\alpha_{i}}.
\]
For $E\left[  S_{n}\right]  $, first note that by Lemma \ref{lem3}, we have
\[
\lim_{p\rightarrow1}\frac{E[Y_{i}]}{d_{1}\overline{F}_{1}\left(  d_{1}\right)
}=\lim_{p\rightarrow1}\frac{E[Y_{i}]}{d_{i}\overline{F}_{i}\left(
d_{i}\right)  }\frac{d_{i}\overline{F}_{i}\left(  d_{i}\right)  }%
{d_{1}\overline{F}_{1}\left(  d_{1}\right)  }=\frac{\left(  \lambda
_{i}^{1-\alpha_{i}}-1\right)  \theta_{i}^{1/\alpha_{i}}}{1-\alpha_{i}} = \frac{\left(  \lambda
_{i}^{1-\alpha_{1}}-1\right)  \theta_{i}^{1/\alpha_{1}}}{1-\alpha_{1}},
\]
since $\theta_{i}=0$ if $\alpha_{i}>\alpha_{1}$. Then we have
\[
\lim_{p\rightarrow1}\frac{E\left[  S_{n}\right]  }{d_{1}\overline{F}%
_{1}\left(  d_{1}\right)  }=\lim_{p\rightarrow1}\sum_{j=1}^{n}\frac{E[Y_{j}%
]}{d_{1}\overline{F}_{1}\left(  d_{1}\right)  }=\sum_{j=1}^{n}\frac{\left(
\lambda_{j}^{1-\alpha_{1}}-1\right)  \theta_{j}^{1/\alpha_{1}}}{1-\alpha_{1}%
}.
\]
On the other hand, by Lemma \ref{lem3} and following the proof of Theorem \ref{theorem3}, we have%
\[
\lim_{p\rightarrow1}\frac{\mathrm{VaR}_{p}(S_{n})}{\mathrm{VaR}_{p}(X_{1}%
)}=\left(  \sum_{j\in Z}\left(  \theta_{j}^{-1/\alpha_{1}}+\xi^{\alpha
_{1}/\alpha_{j}}\right)  ^{-\alpha_{1}}\right)  ^{1/\alpha_{1}}.
\]
This is due to the fact that if $\alpha_{j}>\alpha_{1}$, we have $\theta_{j}=0$
and $j\not \in Z$. Then, it follows that
\begin{align*}
\lim_{p\rightarrow1} I_2 \cdot I_3& =\lim_{p\rightarrow1}\frac{E\left[  Y_{i}\right]  }{E\left[  S_{n}\right]
}\frac{\mathrm{VaR}_{p}(S_{n})}{\mathrm{VaR}_{p}(X_{i})}  \\
&  =\lim
_{p\rightarrow1}\frac{E[Y_{i}]}{d_{i}\overline{F}_{i}\left(  d_{i}\right)
}\frac{d_{1}\overline{F}_{1}\left(  d_{1}\right)  }{E\left[  S_{n}\right]
}\frac{\mathrm{VaR}_{p}(S_{n})}{\mathrm{VaR}_{p}(X_{1})}\frac{d_{i}%
}{\mathrm{VaR}_{p}(X_{i})}\frac{\mathrm{VaR}_{p}(X_{1})}{d_{1}}\frac
{\overline{F}_{i}\left(  d_{i}\right)  }{\overline{F}_{1}\left(  d_{1}\right)
}\\
&  =\frac{\left(  1-\alpha_{1}\right)  \left(  \lambda_{i}^{1-\alpha_{i}%
}-1\right)  \xi^{\alpha_{1}/\alpha_{i}-1}}{(1-\alpha_{i})\sum_{j=1}^{n}\left(
\lambda_{j}^{1-\alpha_{1}}-1\right)  \theta_{j}^{1/\alpha_{1}}}\left(
\sum_{j\in Z}\left(  \theta_{j}^{-1/\alpha_{1}}+\xi \right)  ^{-\alpha_{1}%
}\right)  ^{1/\alpha_{1}}.
\end{align*}
This completes the proof.
\end{proof}

\bigskip

\subsection{Proof of Theorem \ref{thm:opt_sol}}

\begin{proof}
When $\xi\geq 1$, the smallest value possible of $\mathrm{DR}_i(1)$ is 1, which is achieved when $\Delta_{\xi
,\boldsymbol{\lambda}}=0$. This is obtained when the set $Z$ is an empty set, which happens when $\xi
\leq \theta_{j}^{-1/\alpha_{1}}\left(  \lambda_{j}-1\right)  ^{-1}$, or
equivalently $\lambda_{j}\leq1+\theta_{j}^{-1/\alpha_{1}}\xi^{-1}$, for all
$j=1,2,...,n$. By (\ref{a2}), this means that $1<\lambda_{j}\leq1+\theta_{j}^{-1/\alpha_{1}}\xi^{-1}$ are optimal solutions in this case.

When $\xi<1$, the smallest possible
$\mathrm{DR}_{i}(1)$ is $\xi^{\alpha_{1}/\alpha_{i}}$, which is achieve when
$\lambda_{i}\geq \xi^{-\alpha_{1}/\alpha_{i}}$ and $\Delta_{\xi
,\boldsymbol{\lambda}}=0$. As seen in the previous case, $\Delta_{\xi
,\boldsymbol{\lambda}}=0$ when $\lambda_{j}\leq1+\theta_{j}^{-1/\alpha_{1}}\xi^{-1}$, for all
$j=1,2,...,n$. Moreover, this condition can be satisfied at the same time as $\lambda_{i}\geq \xi^{-\alpha_{1}/\alpha_{i}}$, since $0\leq \theta_{i} \leq1$ for all $i=1,...,n$ and in Model 2, $\alpha
_{1}=\min \left \{  \alpha_{1},...,\alpha_{n}\right \}  $, which yields
$\alpha_{1}/\alpha_{i}\leq1$ for all $i=1,...,n$.

Combining these two cases, we obtain the set of optimal $\boldsymbol{\lambda}$, which is
\[
\boldsymbol{\lambda} \in \{(\lambda_{1},...,\lambda_{n}):\min \{ \xi^{-\alpha
_{1}/\alpha_{i}},1\} \leq \lambda_{i} \text{ and } \lambda_{i} \leq1+\theta
_{i}^{-1/\alpha_{1}} \xi^{-1}\}.
\]

\end{proof}
\bigskip

\section{Comparison of algorithms for simulation study} \label{Appendix_algo}

In this section, we provide a comparison of candidates for the numerical optimization algorithm used in Section \ref{sec:sim}, and justify the choice of GSA used in this paper. In the experiments within this comparison, we use a set of 20 samples with parameters from the second study of risks of different indices, and set $\xi=0.3^{1/8.5}$ and $p=0.95$. We consider the following global optimization algorithms:
\begin{enumerate}
    \item Artificial Bee Colony (ABC) \citep{karaboga2005idea}: an optimizer inspired by the foraging behaviour of a honey bee colony which searches for an optimal food source.
    \item Differential Evolution (DE) \citep{storn1997differential,mullen2011deoptim}: an optimization algorithm where candidates for the optimal solution are improved through evolution mechanisms such as mutation, crossover, selection.
    \item Generalized Simulated Annealing (GSA) \citep{tsallis1996generalized,xiang_generalized_2013}: a stochastic optimizer inspired by the annealing process in Metallurgy, which uses a heat treatment to minimize the material's internal energy (thermodynamics).
    \item Harmony Search (HS) \citep{geem_new_2001}: an optimization algorithm which mimics the process of improvisation of musicians in search of harmony.
    \item Particle Swarm Optimization (PSO) \citep{kennedy1995particle,clerc_standard_2012}: an algorithm imitating the movements of a flock of birds or a school of fish towards the optimal solution.
\end{enumerate}
While DE and GSA are applied using their respective \texttt{R} packages (\texttt{DEoptim} and \texttt{GenSA}), the algorithms  ABC, HS and PSO are built from scratch due to the unstable performance and poor maintenance of their \texttt{R} packages. To ensure the fairness of the comparison between algorithms, four factors are kept unchanged across these five experiments. Firstly, all experiments use the same 20 samples of $X_i$'s for the estimations described in Section \ref{sec:sim}, where each sample leads to an approximation of the optimal solution $\boldsymbol{\lambda}(p)$. Secondly, all variations of these experiments use the same starting points for the algorithm (i.e. the initial guess for the optimal solution). Thirdly, the convergence criterion is consistent across algorithms, where algorithms are deemed to have found the optimal solution if the function value is not improved after 500 consecutive iterations, with an iteration refers to an update of the optimal solution. Lastly, the maximum number of iterations allowed for each algorithm is the same, where algorithms are terminated if they do not converge after 10 thousand iterations.

We compare the performance of the algorithms based on two criteria: efficiency and accuracy. On the one hand, the efficiency is measured by the time $T_k$ ($k\in\{\text{ABC, DE, GSA, HS, PSO}\}$) needed for the algorithm to either converge or be terminated. Using 20 samples of $X_i$'s, we obtain a sample $t_{1,k},...,t_{20,k}$ of $T_k$. Inspired by the comparison protocols presented by \cite{beiranvand2017best}, Figures \ref{fig:algo_runningtime} and \ref{fig:algo_runningtime2} show the box plots and the empirical cumulative distribution functions of $T_k$
\[
\widehat{F}_{T_k}(x)= \frac{1}{20} \sum_{j=1}^{20} 1_{\{t_{jk}\leq x\}}.
\]
The results show that GSA and HS are the most efficient algorithms, as 100\% of numerical optimization problems tested are solved within 5 hours using these algorithms. ABC is the most inefficient among five algorithms as it takes more than 90 hours for ABC to solve any problem.

\begin{figure}[htbp]
\centering
\begin{subfigure}{.45\textwidth}
  \centering
   \includegraphics[width=\linewidth]{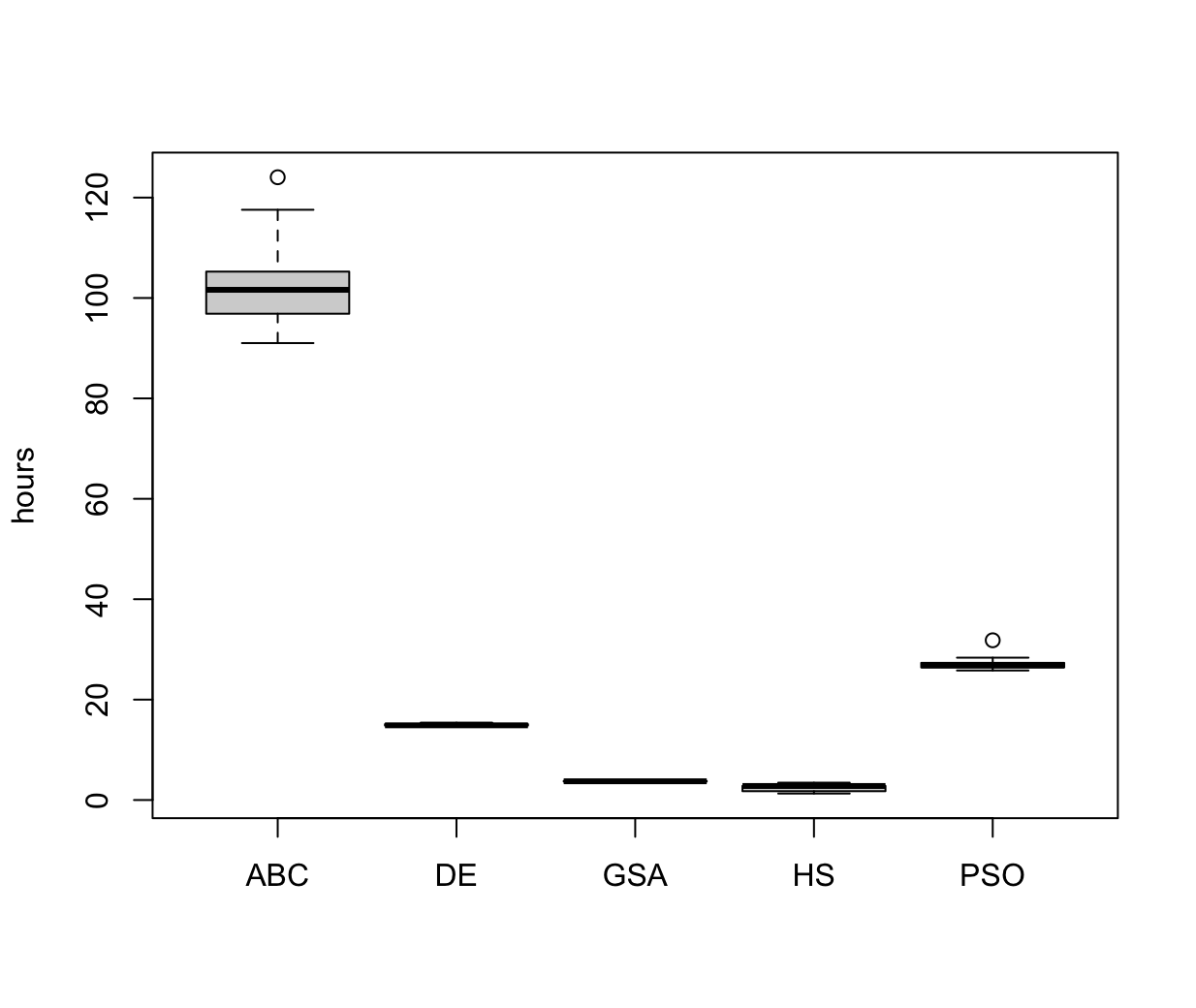}
 \caption{\small With ABC}%
    \label{fig:algo_runningtime_a}
\end{subfigure}%
\begin{subfigure}{.45\textwidth}
  \centering
   \includegraphics[width=\linewidth]{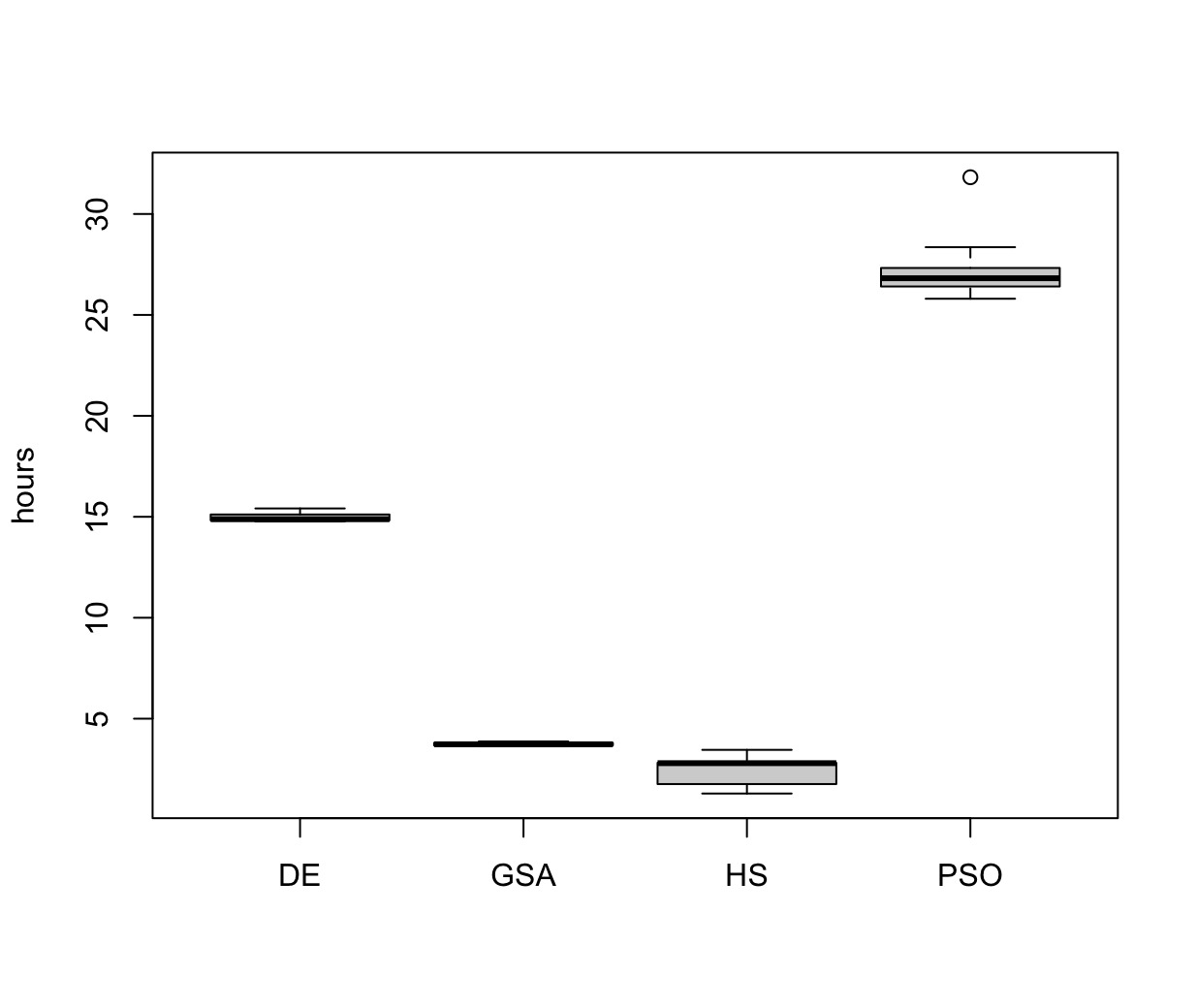}
 \caption{\small Without ABC}%
    \label{fig:algo_runningtime_b}
\end{subfigure}
\captionsetup{justification=raggedright,singlelinecheck=false}
\caption{\small Box plots of $T_k$}
\label{fig:algo_runningtime}
\end{figure}

\begin{figure}[htbp]
    \centering
    \includegraphics[width=0.7\linewidth]{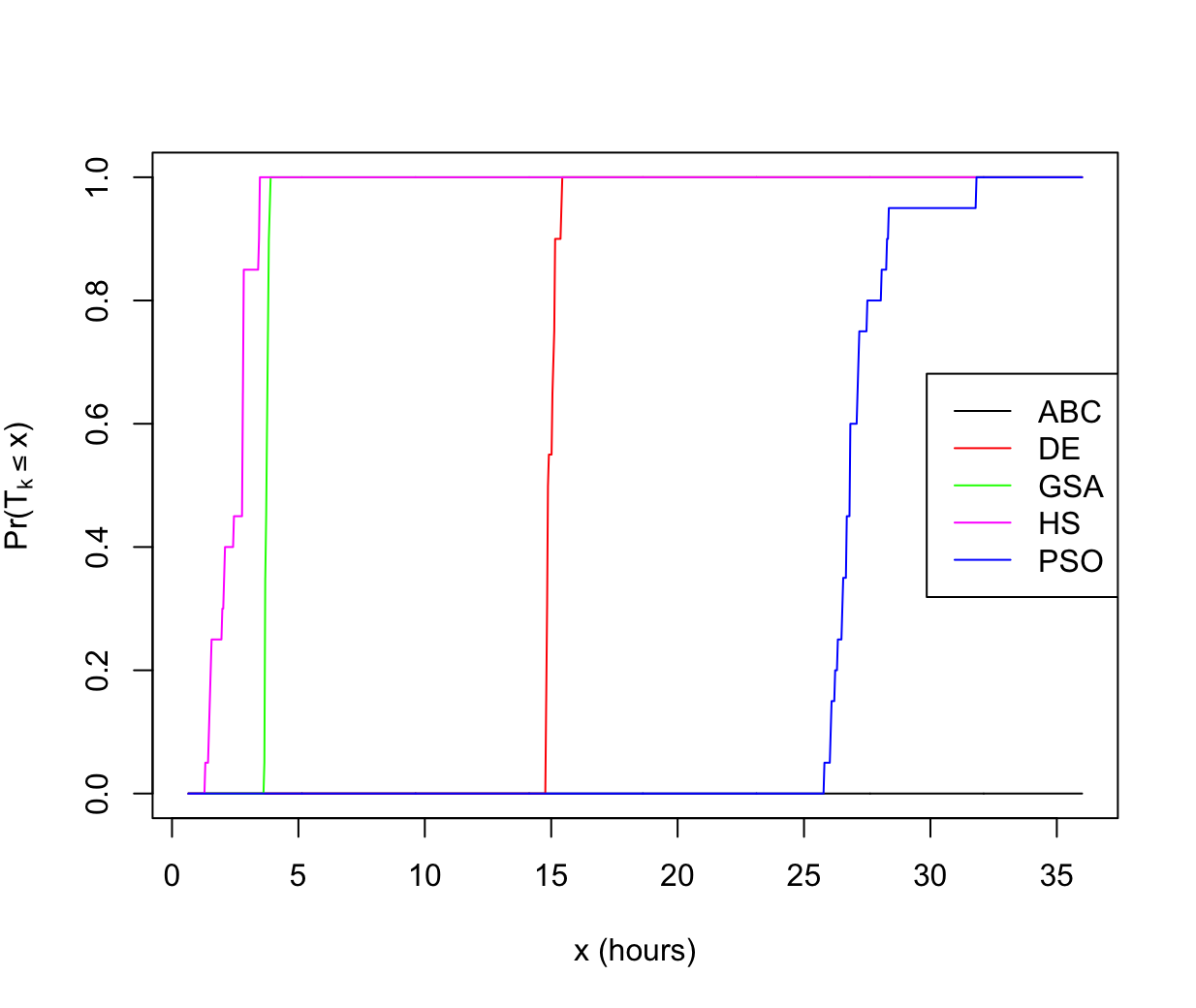}
    \captionsetup{justification=raggedright,singlelinecheck=false}
    \caption{Plots of empirical cumulative distribution functions of $T_k$}
    \label{fig:algo_runningtime2}
\end{figure}

On the other hand, the accuracy of algorithms is measured by the absolute error $ER_k$ ($k\in\{\text{ABC, DE, GSA, HS, PSO}\}$), with observations $er_{j,k}$ ($j=1,...,20$) defined as follows:
\[
er_{j,k} = |x_{j,k} - \min_{k} \{x_{j,k}\}|
\]
where $x_{j,k}$ is the optimal solution of the problem from the $j$-th sample obtained using algorithm $k$, and $\min_{k} \{x_{j,k}\}$ is assumed to be the true optimal solution of problem $j$. Figures \ref{fig:algo_error} and \ref{fig:algo_error2} show the box plots of $ER_k$ and the empirical cumulative distribution functions $\widehat{F}_{ER_k}$, which is defined similarly to $\widehat{F}_{T_k}$. It is clear that HS has the worst performance in terms of accuracy. Although GSA is not as accurate is ABC, DE and PSO, it still allows us to obtain a solution within $10^{-6}$ of the true optimal solution in all problems.

\begin{figure}[htbp]
\centering
\begin{subfigure}{.45\textwidth}
  \centering
   \includegraphics[width=\linewidth]{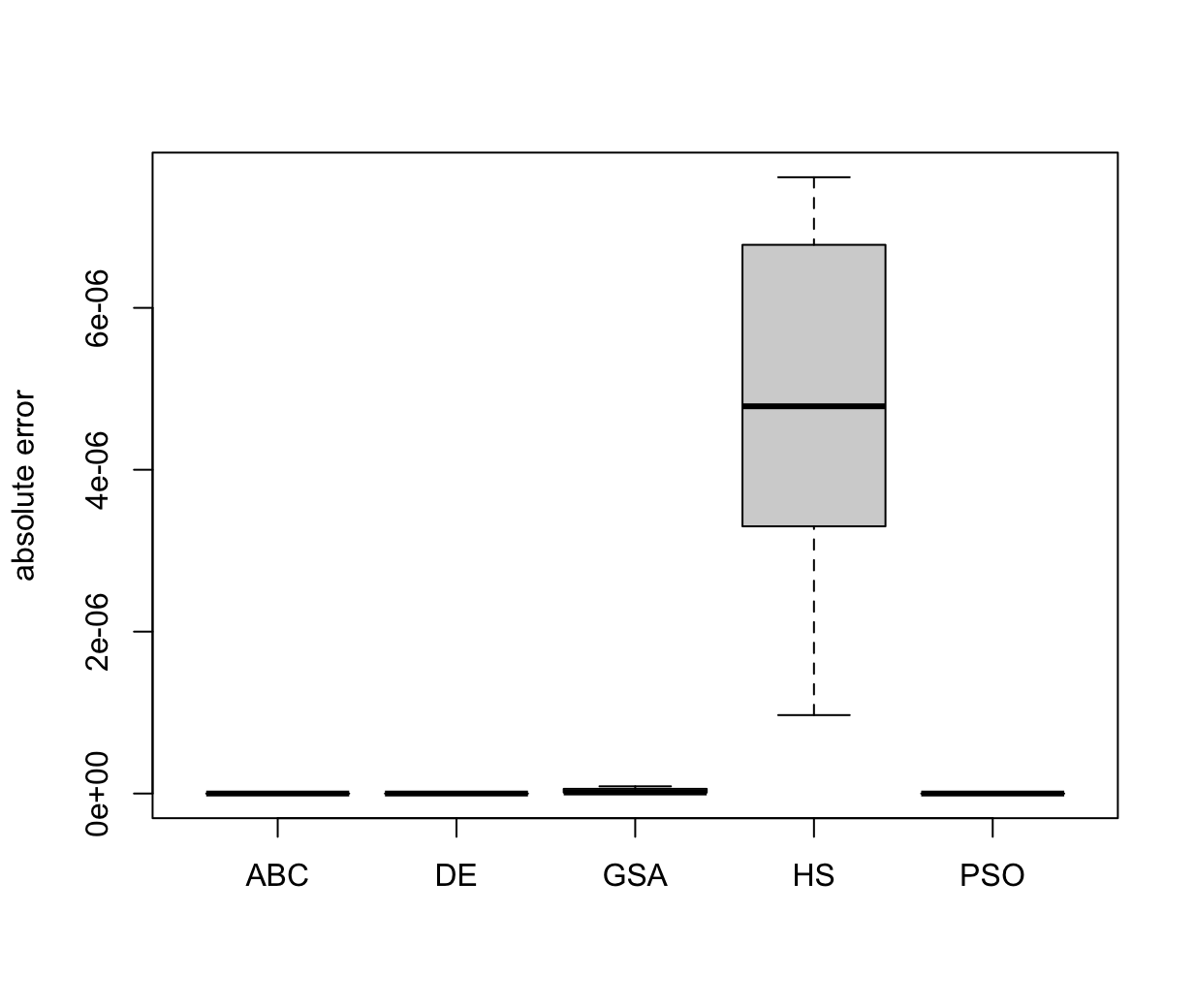}
 \caption{\small With HS}%
    \label{fig:algo_error_a}
\end{subfigure}%
\begin{subfigure}{.45\textwidth}
  \centering
   \includegraphics[width=\linewidth]{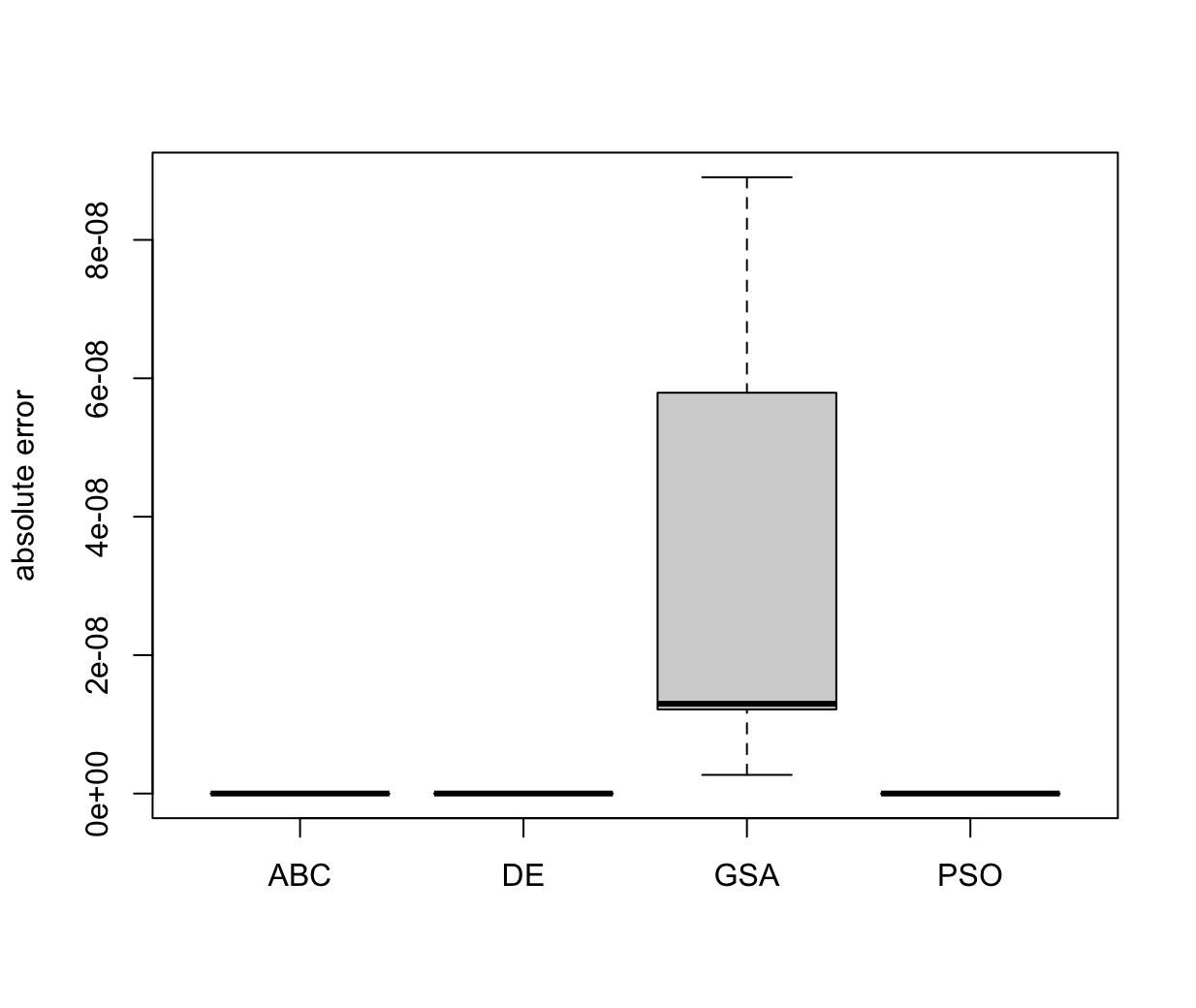}
 \caption{\small Without HS}%
    \label{fig:algo_error_b}
\end{subfigure}
\captionsetup{justification=raggedright,singlelinecheck=false}
\caption{\small Box plots of $ER_k$}
\label{fig:algo_error}
\end{figure}

\begin{figure}[htbp]
    \centering
    \includegraphics[width=0.7\linewidth]{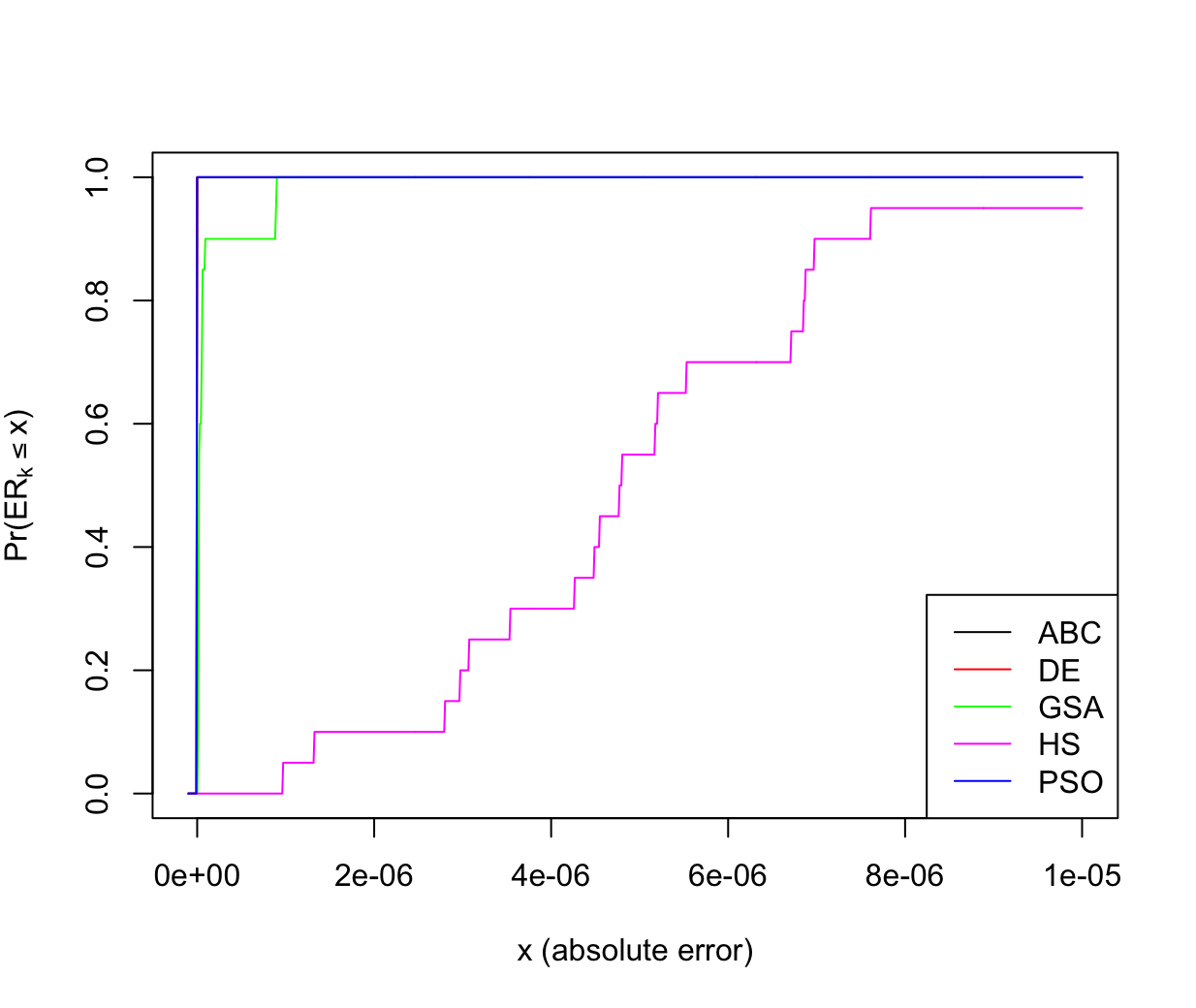}
    \captionsetup{justification=raggedright,singlelinecheck=false}
    \caption{Plots of empirical cumulative distribution functions of $ER_k$}
    \label{fig:algo_error2}
\end{figure}

Considering the performance of those five algorithms, GSA is the best candidate for the simulation study due to its ability to balance efficiency and accuracy.

\section{Exploratory analysis of the NFIP data}\label{Appendix_data_test}

In this section, we provide further details on the results of the statistical tests carried out in the exploratory analysis of the NFIP data in the main text. Denote $\widehat{\gamma}(k) = 1/\widehat{\alpha}(k)$ an inverse tail index estimated using the Hill's estimator with $k$ largest observations. The first test is for the regularly varying tail of the risks (Theorem 1 of \cite{dietrich2002testing}). A stricter version of the null hypothesis $H_0:\overline{F} \in RV_{-1/\gamma}$ for some $\gamma>0$ is not rejected if the test statistic
\[
k\int^1_0 \left( \frac{\log {X}_{(m-\lfloor kt \rfloor)} - \log {X}_{(m-k)}}{\widehat{\gamma}(k)} + \log t\right)^2 t^2 dt,
\]
where ${X}_{(1)}\leq {X}_{(2)}\leq \cdots \leq {X}_{(m)}$ denote the order statistics
for a sample ${X}_{1},...,{X}_{m}$ of distribution $F$, converges in distribution as $m \rightarrow \infty$ to
\[
T=\int_0^1 \left(B(t) + t \log t\int_0^1 \frac{B(s)}{s\text{ }ds} \right)^2 dt,
\]
where $\{B(t)\}$ is a Brownian bridge. Critical values of $T$ are available in \cite{dietrich2002testing}. The results in Figure \ref{fig:RV_test} show that there is strong evidence at the chosen significance level 0.01 that all three risks in our dataset have a regularly varying tail.

\begin{figure}[htbp]
\centering
\begin{subfigure}{0.33\textwidth}
  \centering
   \includegraphics[width=\linewidth]{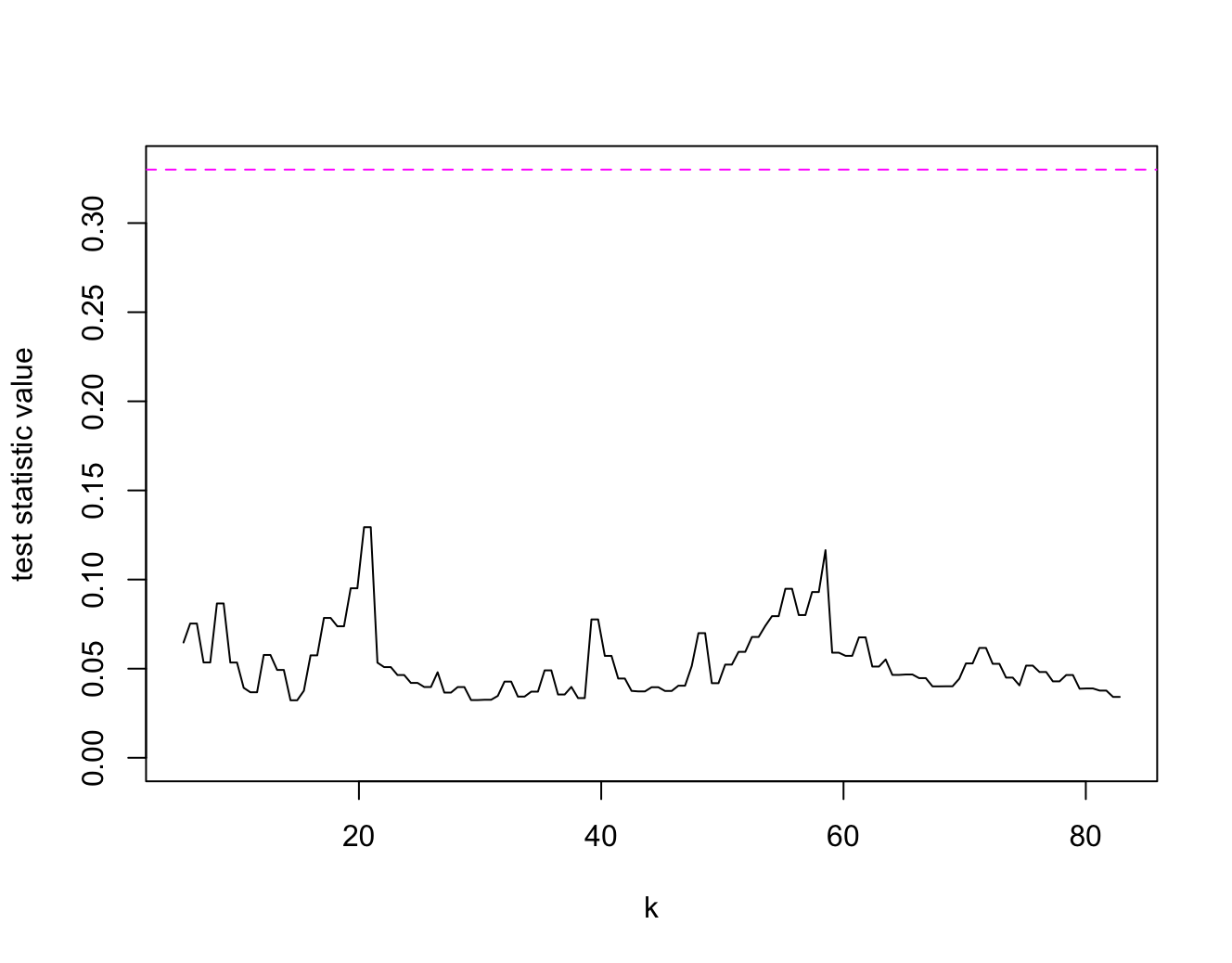}
    \caption{\small NY}
    \label{fig:RV_testNY}
\end{subfigure}%
\begin{subfigure}{.33\textwidth}
  \centering
   \includegraphics[width=\linewidth]{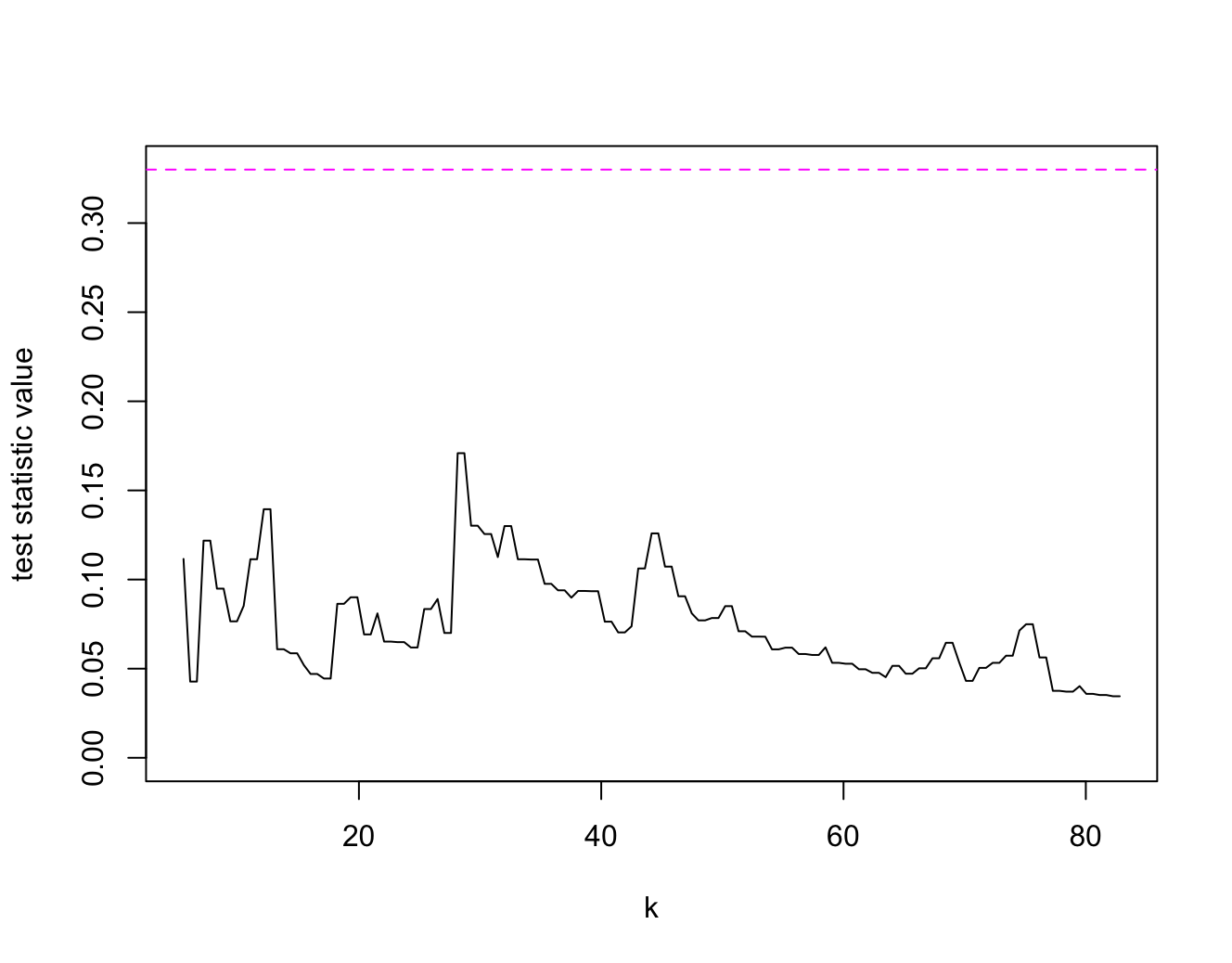}
 \caption{\small CA}%
    \label{fig:RV_testCA}
\end{subfigure}%
\begin{subfigure}{.33\textwidth}
  \centering
   \includegraphics[width=\linewidth]{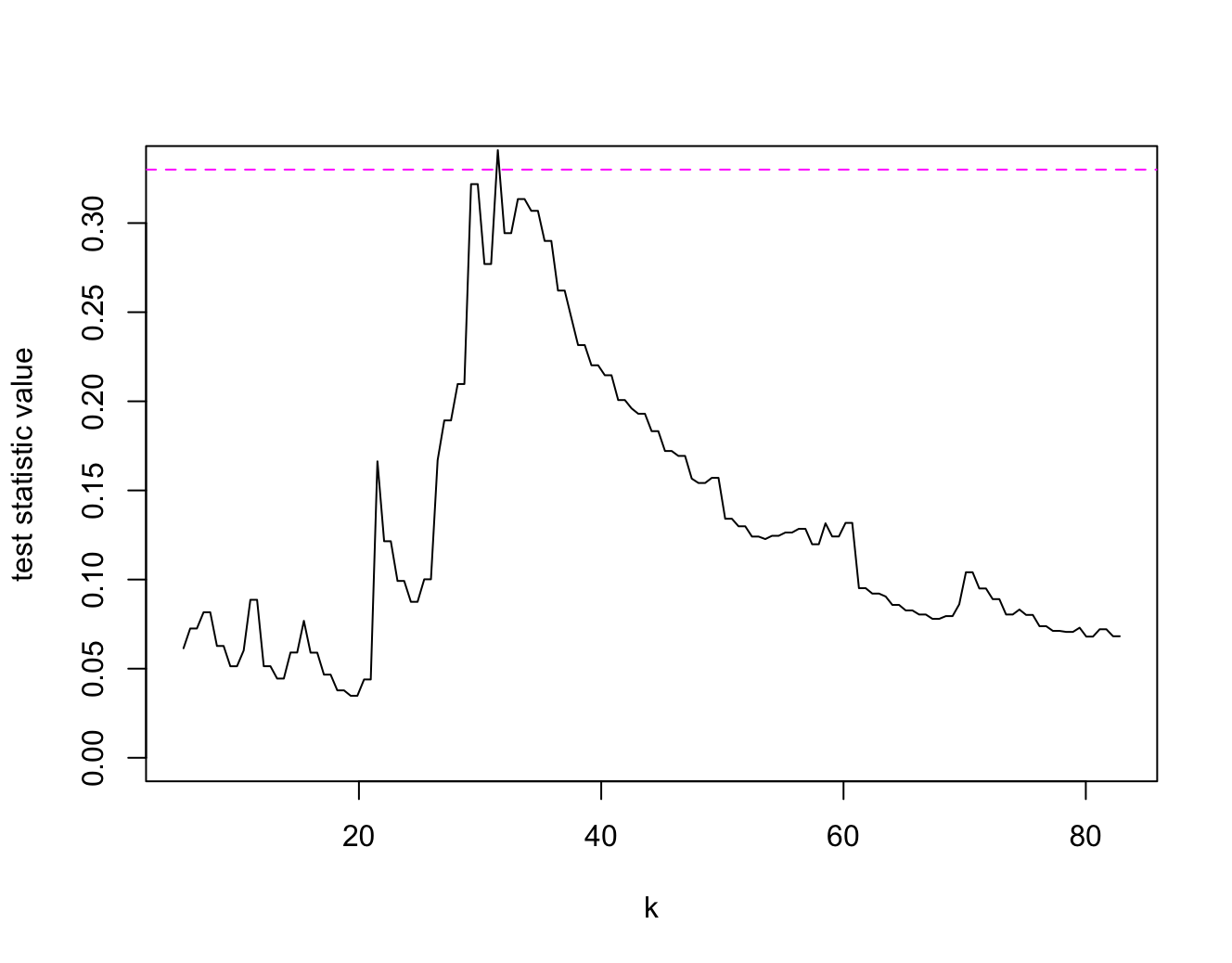}
 \caption{\small FL}%
    \label{fig:RV_testFL}
\end{subfigure}
\caption{\small Results of regularly varying distribution tests. If the test statistic is smaller than the critical value for a large range of $k$ then there is strong evidence that $F$ has a regularly varying tail. The dashed line indicates the critical value of $T$ of significance level 0.01.}
\label{fig:RV_test}
\end{figure}

Table \ref{indep_results} shows the results of the Pearson and Spearman correlation tests. Two pairs of risks NY - CA and CA - FL present strong evidence of independence at significance level 0.01, while NY - FL seems to have a non-linear monotonic relation. 

\begin{table}[htbp]
    \centering
    \begin{tabular}{|c|c|c|}
    \hline
         & Pearson correlation test & Spearman correlation test \\
          & \makecell{$H_0$:\textit{ two risks are} \\\textit{linearly independent}} & \makecell{$H_0$: \textit{two risks do not have}\\\textit{ monotonic relation}}\\ \hline
          NY vs. CA & \makecell{corr: -0.009524\\p-value: 0.8233} & \makecell{corr: 0.042562\\p-value: 0.1435} \\ \hline
          NY vs. FL & \makecell{corr: -0.002893\\p-value: 0.9459} & \makecell{corr: 0.131117\\p-value: $4.192\times10^{-6}$}\\ \hline
          CA vs. FL & \makecell{corr: -0.007786\\p-value: 0.8552} & \makecell{corr: -0.031568\\p-value: 0.2777}\\ \hline
    \end{tabular}
    \captionsetup{justification=raggedright,singlelinecheck=false}
    \caption{Results of independence tests}
    \label{indep_results}
\end{table}

Lastly, we carry out tests for tail equivalence proposed by \cite{daouia2024optimal} for two pairs of risks NY - CA and CA - FL (denoted below $X_1$ and $X_2$). For $i=1,2$, based on $k_i$ largest observations of $X_i$, we obtain the estimated inverse tail index $\widehat{\gamma}_i(k_i)$ of $X_i$ by using the Hill's estimator. Under the assumption of independence between risks, \cite{daouia2024optimal} define the pooled inverse tail index as 
\[
    \widehat{\gamma}_{pool} = \sum_{i=1}^2 w_i \widehat{\gamma}_i(k_i),
\]
where
\[
(w_1,w_2) = \boldsymbol{w} = \frac{\Sigma^{-1} \mathbf{1}}{\mathbf{1}^\mathrm{T} \Sigma^{-1} \mathbf{1}}, \qquad
\Sigma = \begin{bmatrix}
\widehat{\gamma}_1^2(k_1)/k_1 &  0   \\
0  &  \widehat{\gamma}_2^2(k_2)/k_2
\end{bmatrix}
\]
and $\mathbf{1} = (1,1)^\mathrm{T} \in \mathbb{R}^2$. Therefore, the test statistic is defined as
\[
     \sum_{i=1}^2 k_i \frac{(\widehat{\gamma}_i(k_i) - \widehat{\gamma}_{pool})^2}{\widehat{\gamma}_i^2(k_i)}
\]
and follows a $\chi ^2_{1}$ distribution under the null hypothesis $\mathrm{H}_0: \widehat{\gamma}_i(k_i) = \widehat{\gamma}_{pool} \forall i$. If the null hypothesis is not rejected, then $X_1$ and $X_2$ are assumed to have the same tail index defined as
\begin{equation}\label{Hill_pooled_formula}
\widehat{\alpha}_{pool} = 1/\widehat{\gamma}_{pool},
\end{equation} 
which is called the pooled tail estimator. Results of this test are discussed in Section \ref{sec:realdata} of the main text.

\end{appendices}

\bibliographystyle{apalike}
\bibliography{ref_paper.bib}

@book{woo2011calculating,
  title={Calculating catastrophe},
  author={Woo, Gordon},
  year={2011},
  publisher={World Scientific}
}

@article{malamud2005characterizing,
  title={Characterizing wildfire regimes in the United States},
  author={Malamud, Bruce D and Millington, James DA and Perry, George LW},
  journal={Proceedings of the National Academy of Sciences},
  volume={102},
  number={13},
  pages={4694--4699},
  year={2005},
  publisher={National Academy of Sciences}
}

@article{embrechts2009additivity,
  title={Additivity properties for Value-at-Risk under Archimedean dependence and heavy-tailedness},
  author={Embrechts, Paul and Ne{\v{s}}lehov{\'a}, Johanna and W{\"u}thrich, Mario V},
  journal={Insurance: Mathematics and Economics},
  volume={44},
  number={2},
  pages={164--169},
  year={2009},
  publisher={Elsevier}
}

@article{hogg1983estimation,
  title={On the estimation of long tailed skewed distributions with actuarial applications},
  author={Hogg, Robert V and Klugman, Stuart A},
  journal={Journal of Econometrics},
  volume={23},
  number={1},
  pages={91--102},
  year={1983},
  publisher={Elsevier}
}

@article{hsieh1999robustness,
  title={Robustness of tail index estimation},
  author={Hsieh, Ping-Hung},
  journal={Journal of Computational and Graphical Statistics},
  volume={8},
  number={2},
  pages={318--332},
  year={1999},
  publisher={Taylor \& Francis}
}

@article{sornette1996rank,
  title={Rank-ordering statistics of extreme events: Application to the distribution of large earthquakes},
  author={Sornette, Didier and Knopoff, Leon and Kagan, YY and Vanneste, Christian},
  journal={Journal of Geophysical Research: Solid Earth},
  volume={101},
  number={B6},
  pages={13883--13893},
  year={1996},
  publisher={Wiley Online Library}
}

@article{cui2022asymptotic,
  title={Asymptotic analysis of portfolio diversification},
  author={Cui, Hengxin and Tan, Ken Seng and Yang, Fan and Zhou, Chen},
  journal={Insurance: Mathematics and Economics},
  volume={106},
  pages={302--325},
  year={2022},
  publisher={Elsevier}
}

@book{feller1991introduction,
  title={An introduction to probability theory and its applications},
  author={Feller, William},
  volume={2},
  year={1971},
  publisher={John Wiley \& Sons}
}

@book{haan2006extreme,
  title={Extreme value theory: an introduction},
  author={de Haan, Laurens and Ferreira, Ana},
  volume={3},
  year={2006},
  publisher={Springer}
}

@article{denuit2022risk,
  title={Risk-sharing rules and their properties, with applications to peer-to-peer insurance},
  author={Denuit, Michel and Dhaene, Jan and Robert, Christian Y},
  journal={Journal of Risk and Insurance},
  volume={89},
  number={3},
  pages={615--667},
  year={2022},
  publisher={Wiley Online Library}
}

@book{boyd2004convex,
  title={Convex optimization},
  author={Boyd, Stephen and Vandenberghe, Lieven},
  publisher={Cambridge University Press},
  year={2004},
pages=133
}

@incollection{kreps1989nash,
  title={Nash equilibrium},
  author={Kreps, David M},
  booktitle={Game theory},
  pages={167--177},
  year={1989},
  publisher={Springer}
}

@article{bollmann2019international,
  title={International catastrophe pooling for extreme weather},
  author={Bollmann, Andreas and Wang, S},
  journal={Society of Actuaries},
  year={2019}
}

@article{cui2021diversification,
  title={Diversification in catastrophe insurance markets},
  author={Cui, Hengxin and Tan, Ken Seng and Yang, Fan},
  journal={ASTIN Bulletin: The Journal of the IAA},
  volume={51},
  number={3},
  pages={753--778},
  year={2021},
  publisher={Cambridge University Press}
}

@misc{swiss_re_sigma_2025,
	title = {sigma 1/2025: {Natural} catastrophes: insured losses on trend to {USD} 145 billion in 2025},
	shorttitle = {sigma 1/2025},
	url = {https://www.swissre.com/institute/research/sigma-research/sigma-2025-01-natural-catastrophes-trend.html},
	language = {en},
	author = {{Swiss Re}},
	month = apr,
	year = {2025}
}

@article{daouia2024optimal,
  title={Optimal weighted pooling for inference about the tail index and extreme quantiles},
  author={Daouia, Abdelaati and Padoan, Simone A and Stupfler, Gilles},
  journal={Bernoulli},
  volume={30},
  number={2},
  pages={1287--1312},
  year={2024},
  publisher={Bernoulli Society for Mathematical Statistics and Probability}
}

@article{dietrich2002testing,
  title={Testing extreme value conditions},
  author={Dietrich, Daniel and De Haan, Laurens and Husler, Jurg},
  journal={Extremes},
  volume={5},
  number={1},
  pages={71},
  year={2002},
  publisher={Springer Nature BV}
}

@article{tsallis1996generalized,
  title={Generalized simulated annealing},
  author={Tsallis, Constantino and Stariolo, Daniel A},
  journal={Physica A: Statistical Mechanics and its Applications},
  volume={233},
  number={1-2},
  pages={395--406},
  year={1996},
  publisher={Elsevier}
}

@article{xiang_generalized_2013,
	title = {Generalized {Simulated} {Annealing} for {Global} {Optimization}: {The} {GenSA} {Package}},
	volume = {5},
	shorttitle = {Generalized {Simulated} {Annealing} for {Global} {Optimization}},
	doi = {10.32614/RJ-2013-002},
	journal = {The R Journal Volume 5(1):13-29, June 2013},
	author = {Xiang, Yang and Gubian, Sylvain and Suomela, Brian and Hoeng, Julia},
	month = jun,
	year = {2013},
}

@article{karaboga2005idea,
  title={An idea based on honey bee swarm for numerical optimization},
  author={Karaboga, Dervis and others},
  year={2005},
  publisher={Technical report-tr06, Erciyes University, Engineering Faculty, Computer Engineering Department}
}

@article{storn1997differential,
  title={Differential evolution--a simple and efficient heuristic for global optimization over continuous spaces},
  author={Storn, Rainer and Price, Kenneth},
  journal={Journal of global optimization},
  volume={11},
  pages={341--359},
  year={1997},
  publisher={Springer}
}

@article{mullen2011deoptim,
  title={DEoptim: An R package for global optimization by differential evolution},
  author={Mullen, Katharine M and Ardia, David and Gil, David L and Windover, Donald and Cline, James},
  journal={Journal of Statistical Software},
  volume={40},
  pages={1--26},
  year={2011}
}

@article{geem_new_2001,
	title = {A {New} {Heuristic} {Optimization} {Algorithm}: {Harmony} {Search}},
	volume = {76},
	issn = {0037-5497},
	shorttitle = {A {New} {Heuristic} {Optimization} {Algorithm}},
	url = {https://doi.org/10.1177/003754970107600201},
	doi = {10.1177/003754970107600201},
	language = {EN},
	number = {2},
	urldate = {2025-06-24},
	journal = {Simulation},
	author = {Geem, Zong Woo and Kim, Joong Hoon and Loganathan, G.V.},
	month = feb,
	year = {2001},
	pages = {60--68},
}

@misc{clerc_standard_2012,
	title = {Standard {Particle} {Swarm} {Optimisation}},
	author = {Clerc, Maurice},
	year = {2012},
    howpublished = {HAL Open Access Archive}
}

@inproceedings{kennedy1995particle,
  title={Particle swarm optimization},
  author={Kennedy, James and Eberhart, Russell},
  booktitle={Proceedings of ICNN'95-international conference on neural networks},
  volume={4},
  pages={1942--1948},
  year={1995},
  organization={IEEE}
}

@article{beiranvand2017best,
  title={Best practices for comparing optimization algorithms},
  author={Beiranvand, Vahid and Hare, Warren and Lucet, Yves},
  journal={Optimization and Engineering},
  volume={18},
  pages={815--848},
  year={2017},
  publisher={Springer}
}

@article{potter1942mean,
  title={The mean values of certain Dirichlet series, II},
  author={Potter, HSA},
  journal={Proceedings of the London Mathematical Society},
  volume={2},
  year={1942},
  publisher={Wiley Online Library}
}

@article{karamata1930,
  author  = {Jovan Karamata},
  title   = {Sur un mode de croissance régulière des fonctions},
  journal = {Mathematica (Cluj)},
  volume  = {4},
  pages   = {38--53},
  year    = {1930}
}

@article{ibragimov2009nondiversification,
  title={Nondiversification traps in catastrophe insurance markets},
  author={Ibragimov, Rustam and Jaffee, Dwight and Walden, Johan},
  journal={The Review of Financial Studies},
  volume={22},
  number={3},
  pages={959--993},
  year={2009},
  publisher={Oxford University Press}
}

@article{boonen2024pareto,
  title={Pareto-efficient risk sharing in centralized insurance markets with application to flood risk},
  author={Boonen, Tim J and Chong, Wing Fung and Ghossoub, Mario},
  journal={Journal of Risk and Insurance},
  volume={91},
  number={2},
  pages={449--488},
  year={2024},
  publisher={Wiley Online Library}
}

@misc{natural_resources_canada_canadas_2023,
	title = {Canada’s record-breaking wildfires in 2023: {A} fiery wake-up call},
	url = {https://natural-resources.canada.ca/stories/simply-science/canada-s-record-breaking-wildfires-2023-fiery-wake-call},
	author = {{Natural Resources Canada}},
	year = {2023},
}

@article{hill1975,
  author  = {Bruce M. Hill},
  title   = {A Simple General Approach to Inference About the Tail of a Distribution},
  journal = {The Annals of Statistics},
  volume  = {3},
  number  = {5},
  pages   = {1163--1174},
  year    = {1975},
  doi     = {10.1214/aos/1176343247}
}

@article{ghossoub2024pareto,
  title={Pareto-optimal peer-to-peer risk sharing with robust distortion risk measures},
  author={Ghossoub, Mario and Zhu, Michael B and Chong, Wing Fung},
  journal={ASTIN Bulletin: The Journal of the IAA},
  pages={1--27},
  year={2024},
  publisher={Cambridge University Press}
}

@book{pisarenko2010heavy,
  title={Heavy-tailed distributions in disaster analysis},
  author={Pisarenko, V and Rodkin, M},
  volume={30},
  year={2010},
  publisher={Springer Science \& Business Media}
}

@article{froot2002pricing,
  title={The pricing of event risks with parameter uncertainty},
  author={Froot, Kenneth A and Posner, Steven E},
  journal={The Geneva Papers on Risk and Insurance Theory},
  volume={27},
  number={2},
  pages={153--165},
  year={2002},
  publisher={Springer}
}

@article{markowitz1952portfolio,
  title   = {Portfolio Selection},
  author  = {Markowitz, Harry M.},
  journal = {The Journal of Finance},
  volume  = {7},
  number  = {1},
  pages   = {77--91},
  year    = {1952},
  publisher = {Wiley}
}

@article{samuelson1967general,
  title={General proof that diversification pays},
  author={Samuelson, Paul A},
  journal={Journal of Financial and Quantitative Analysis},
  volume={2},
  number={1},
  pages={1--13},
  year={1967},
  publisher={Cambridge University Press}
}

@incollection{arrow1978uncertainty,
  title={Uncertainty and the evaluation of public investment decisions},
  author={Arrow, Kenneth J and Lind, Robert C},
  booktitle={Uncertainty in economics},
  pages={403--421},
  year={1978},
  publisher={Elsevier}
}

@article{holzheu2018natural,
  title={The natural catastrophe protection gap: Measurement, root causes and ways of addressing underinsurance for extreme events},
  author={Holzheu, Thomas and Turner, Ginger},
  journal={The Geneva Papers on Risk and Insurance-Issues and Practice},
  volume={43},
  number={1},
  pages={37--71},
  year={2018},
  publisher={Springer}
}

@article{ibragimov2016heavy,
  title={Heavy tails and copulas: Limits of diversification revisited},
  author={Ibragimov, Rustam and Prokhorov, Artem},
  journal={Economics Letters},
  volume={149},
  pages={102--107},
  year={2016},
  publisher={Elsevier}
}

@book{vincent1981optimality,
  author    = {Vincent, Thomas L. and Grantham, Walter J.},
  title     = {Optimality in Parametric Systems},
  year      = {1981},
  publisher = {John Wiley \& Sons},
  address   = {New York}
}

@article{marler2004survey,
  title={Survey of multi-objective optimization methods for engineering},
  author={Marler, R Timothy and Arora, Jasbir S},
  journal={Structural and multidisciplinary optimization},
  volume={26},
  number={6},
  pages={369--395},
  year={2004},
  publisher={Springer}
}

\nocite{*}

\end{document}